\documentclass[reqno]{amsart}
\usepackage{hyperref}
\usepackage{amsmath}
\usepackage{amssymb}
\usepackage{graphicx}
\usepackage{tikz}
\usetikzlibrary{calc}
\usepackage{xcolor}
\usetikzlibrary{snakes}
\usepackage{amsthm}

\newcommand{\EEE}{\color{black}}

\theoremstyle{plain}
\begingroup
\newtheorem{theorem}{Theorem}[section]
\newtheorem{lemma}[theorem]{Lemma}
\newtheorem{proposition}[theorem]{Proposition}

\endgroup

\theoremstyle{definition}
\begingroup
\newtheorem{definition}[theorem]{Definition}
\newtheorem{remark}[theorem]{Remark}

\endgroup

\usepackage[english]{babel}

\begin{document}

\title{Crystallization  in the hexagonal lattice for ionic dimers}


\author{Manuel Friedrich}
\address[Manuel Friedrich]{Applied Mathematics M\"unster, University of M\"unster\\
Einsteinstrasse 62, 48149 M\"unster, Germany.}
\email{manuel.friedrich@uni-muenster.de}

\author{Leonard Kreutz}
\address[Leonard Kreutz]{Applied Mathematics M\"unster, University of M\"unster\\
Einsteinstrasse 62, 48149 M\"unster, Germany.}
\email{lkreutz@uni-muenster.de}

\date{}

\begin{abstract}

 We consider finite discrete systems consisting of two different atomic types and investigate ground-state configurations  for configurational energies featuring  two-body short-ranged particle interactions. The atomic potentials favor some reference distance between different atomic types and include repulsive terms for atoms of the same type,  which are typical assumptions in models for ionic dimers.  Our goal is to show a two-dimensional crystallization result. More precisely, we give conditions in order to prove that   energy minimizers are connected subsets of the hexagonal lattice where the two atomic types are  alternately arranged in  the crystal lattice. We also provide explicit formulas for the  ground-state energy. Finally, we characterize the net charge, i.e., the difference of the number of the two atomic types. Analyzing the deviation of configurations from the hexagonal Wulff shape, we prove that for ground states consisting of $n$ particles the net charge is at most of order ${\rm O}(n^{1/4})$ where the scaling is sharp.  
\end{abstract}

\maketitle

\section{Introduction}

 A fundamental problem in crystallography is to understand why atoms or molecules typically arrange themselves into  crystalline order at low temperature. The challenge to prove rigorously that global minimizers of certain configurational energies arrange in a periodic lattice  is referred to as the \emph{crystallization problem} \cite{Blanc}.  Starting with seminal works in one dimension \cite{Gardner,Hamrick,Radi,Ventevogel}, the last decades have witnessed a remarkable interest in this mathematical problem  for systems consisting of \emph{identical} particles.

At zero or very low temperature, atomic interactions are expected to be governed solely by the geometry of the atomic arrangement. Identifying configurations  with the respective positions of identical atoms  $\lbrace x_1, \ldots, x_n \rbrace$, one is concerned with the  minimization of a configurational energy $\mathcal{E}( \lbrace x_1,\ldots,x_n \rbrace )$ which  comprises classical interaction potentials.   Crystallization then consists in proving  that any ground state of $\mathcal{E}$  is a subset of a periodic lattice. \EEE

Despite of its paramount theoretical and
applicative relevance, rigorous  results on crystallization are scarce and even under simplified assumptions the problem presents many difficulties. In this paper, we  contribute to these fundamental mathematical questions by investigating crystallization of particle systems consisting of \emph{two different atomic types}, also called \emph{dimers}. 

It is well known that the geometry of molecular compounds  often differs from that of their
components. An example in this direction  is sodium chloride (NaCl) which has a FCC structure (Face Centered Cubic), whereas
sodium crystals are BCC (Body Centered Cubic) and chlorine crystals are orthorhombic. The goal of this contribution is to illustrate such a phenomenon in a simplified setting by deriving a two-dimensional crystallization result in the hexagonal lattice for dimers whose components instead would likely assemble themselves in a triangular lattice. 

Similar to the problem for systems of identical particles, we follow the classical molecular-mechanical frame of configurational energy minimization. We identify configurations of $n$ particles with their respective \emph{positions} $\{x_1,\ldots,x_n\}  \subset  \mathbb{R}^{2} \EEE$ and additionally with their \emph{types} $\{q_1,\ldots,q_n\}  \in \{-1,1\}^n\EEE$. Our goal is to determine global minimizers of a corresponding  interaction energy $ \mathcal{E}(\{(x_1,q_1),\ldots,(x_n,q_n)\})$.

More specifically, the energy  $\mathcal{E} =  \mathcal{E}_{\mathrm{a}} + \mathcal{E}_{\mathrm{r}}$ is assumed to consist of two short-ranged two-body interaction potentials $\mathcal{E}_{\mathrm{a}}$ and  $\mathcal{E}_{\mathrm{r}}$, where $\mathcal{E}_{\mathrm{a}}$ represents the interactions between atoms of different type and  $\mathcal{E}_{\mathrm{r}}$ encodes the energy contributions of atoms having the  same type. The   potential $\mathcal{E}_{\mathrm{a}}$ is attractive-repulsive and  favors atoms sitting at some specific reference distance, whereas $\mathcal{E}_{\mathrm{r}}$  is a pure repulsive term. The latter is supposed to be strong enough at short distances such that the hexagonal lattice, with the
two atomic types alternating, is energetically preferred with respect to  closer packed structures.  (Note that weaker short-ranged repulsion may indeed favor crystallization in the square lattice, see \cite{FriedrichKreutzSquare}.)

An example of a dimer  satisfying qualitatively  the assumptions on the attractive and repulsive potentials is  graphitic \EEE boron nitride \cite{DresselhausM, Wang}, which, similar to graphene \cite{Geim}, is a one-atom thick layer of atoms arranged in a hexagonal lattice. Moreover, our choice of the interaction potentials is motivated by the modeling of ions in ionic compounds. Indeed, the two interaction energies may be interpreted as (simplified) Coulomb-interactions between ions of equal and opposite charge. In this spirit, we will frequently refer to the particles as \emph{ions} and to the atomic types $\{q_1,\ldots,q_n\}  \in \{-1,1\}^n$ as positive or negative \emph{charges}.

The paper contains the following three main results:
\begin{itemize}
\item[(1)]   Theorem \ref{TheoremEnergyGroundstates}:  Under suitable assumptions on the attractive and the repulsive potentials, we characterize the ground-state energy of $n$-particle configurations of ions in the plane. In particular, we quantify exactly the surface energy related to atoms at the boundary of the ensemble. Moreover, we show that in our modeling frame  ground states are repulsion free, i.e.,  $\mathcal{E}_{\rm r}(\{(x_1,q_1),\ldots,(x_n,q_n)\}) = 0$.    
\smallskip
\item[(2)] Theorem \ref{TheoremGroundstatesleq31}: We characterize the geometry of ground states. We show that each global minimizer of the configurational energy is essentially a connected subset of the regular hexagonal lattice with the two atomic types alternating. This characterization holds except for possibly one defect occurring at the boundary which can be identified explicitly either as a single atom or as a simple octagonal cycle in the structure. 
\smallskip
\item[(3)] Theorem \ref{TheoremCharge}:  We consider the \emph{net charge} defined by $\mathcal{Q}(\{(x_1,q_1),\ldots,(x_n,q_n)\}) = \sum_{i=1}^n q_i$ or, in other words, the  (signed) difference of the number of the two atomic types.    We provide a fine asymptotic characterization of the net charge as the number of particles grows. More precisely, we show that  its fluctuation can be at most of order ${\rm O}(n^{1/4})$, i.e., $|\mathcal{Q}(\{(x_1,q_1),\ldots,(x_n,q_n)\})| \leq c n^{1/4}$ for some constant  $c>0$ independent of $n$. By giving an explicit construction we further prove that the scaling is sharp.
\end{itemize}

To the best of our knowledge, the present paper provides a first  rigorous mathematical crystallization result for two-dimensional molecular compounds.  Let us mention that in case of different types of particles, simulations are abundant, but   rigorous results seem to be limited to   \cite{Betermin,B43,B195}.  We also refer to \cite{Betermin2} for a recent study about  optimal charge distributions on Bravais lattices. \EEE Our analysis is largely inspired by many contributions on identical particle systems. We include here only a minimal   crystallization literature overview and refer the reader to the recent review \cite{Blanc} for a more general perspective.

Mathematical results on crystallization in one dimension, proving or disproving
periodicity of ground-state configurations for different choices of
the energy, are by now quite classical, see, e.g., \cite{Gardner,Hamrick,Radi,Ventevogel}.  In two dimensions for a finite number of identical particles,  ground states have been proved to be
subsets of the triangular lattice by {\sc Heitman, Radin, and   Wagner} \cite{HR,Radin,Wagner83} for sticky and soft, purely two-body interaction energies. Recently, a new perspective on this problem was given by {\sc De Luca and Friesecke} \cite{Lucia} using a  discrete Gauss-Bonnet method from discrete differential geometry.  Including angular potentials favoring  $2\pi/3$ bond-angles, which is typically the case for graphene, crystallization in the hexagonal lattice has been proved by {\sc Mainini and Stefanelli} \cite{Mainini}, see also \cite{emergence} for a related classification of ground states. In a similar fashion, if $\pi/2$ bond-angles are favored, the
square lattice can be recovered \cite{Mainini-Piovano}. Under less restrictive assumptions on the potentials, various results have been obtained in the thermodynamic limit \cite{ELi, Smereka15}, namely as  the  number of
particles $n$ tends to infinity, most notably the seminal paper by {\sc Theil} \cite{Theil}. The crystallization problem in three dimension appears to be very difficult and only few rigorous results \cite{Flateley1,Flateley2, cronut, Suto06} are  currently available. 

In two dimensions, a fine characterization of ground-state geometries is possible: as  $n$ increases, one observes the emergence of a polygonal cluster, the {\it Wulff shape} \cite{Yuen}. For the triangular \cite{Schmidt-disk},
the hexagonal \cite{Davoli15}, and the square lattice \cite{Mainini-Piovano}, ground states differ from
the Wulff shape by at most ${\rm O}(n^{3/4})$ particles and at most by ${\rm O}(n^{1/4})$ in Hausdorff distance. This sharp bound is called the $n^{3/4}$ law \cite{Davoli16}.

Our general proof strategy follows the induction method on bond-graph layers developed in  \cite{HR, Mainini-Piovano, Mainini, Radin}. In particular, our study is related to  \cite{Mainini}, where crystallization for identical particles in the hexagonal lattice has been investigated under three-body angular potentials. Actually, the ground-state energy in the present context coincides with the one obtained there.  

Concerning the characterization of ground-state geometries, however, richer geometric structures with respect to the ones of \cite{Mainini} are possible since  configurations are in general less rigid without angular potentials: simple examples show that ground states may contain
distortions and defects at the boundary such as octagons. From a technical point of view, novel geometric arguments in the
induction step are necessary, complementing the available approaches, see Remark \ref{rem: method} for more details. Let us mention that, in any case, certain quantitative requirements on the potentials are indispensable since other   assumptions  might favor an assemblence of the atoms  in other periodic structures. (In particular, we refer to the forthcoming \EEE paper \cite{FriedrichKreutzSquare} for the square-lattice case.)

For the characterization of the net charge in terms of the sharp scaling ${\rm O}(n^{1/4})$, we use the fact that in ground states the two atomic types are alternately arranged in the hexagonal lattice. Morever, the argument fundamentally relies on the $n^{3/4}$-law \cite{Davoli15} which states that ground states differ from
the Wulff shape by at most ${\rm O}(n^{3/4})$ particles, or equivalently and more relevant in our context, by ${\rm O}(n^{1/4})$ in Hausdorff distance. Let us emphasize that in our model the net charge is not a priori given but a crucial part of the minimization problem. The investigation of ground-state geometries under preassigned net charge is a different challenging problem which we defer to future studies.

The article is organized as follows. In Section 2 we introduce the precise mathematical setting and state the main results. In Section 3 we discuss elementary geometric properties of ground states. In Section 4 we construct explicitly some configurations in order to  provide sharp upper bounds for the ground-state energy and the (signed) net charge. In Section 5 we introduce the concept of boundary energy and  prove corresponding  estimates which are instrumental for the induction method used in the proof of Theorem \ref{TheoremEnergyGroundstates} and \ref{TheoremGroundstatesleq31}.  In Section 6 we give a lower bound for the ground-state energy matching the one of the configurations constructed in Section 4. In Section 7 we characterize the geometry of ground states and finally Section 8 is devoted to the proof of a $n^{1/4}$-law for the net charge of ground states.

\section{Setting and  main results}

In this section we introduce our model, give some basic definitions, and present our main results. 

\subsection{Configurations and interaction energy}
 
 We consider two-dimensional particle systems consisting of two different atomic types and model their interaction by classical potentials in the frame of Molecular Mechanics \cite{Molecular, Friesecke-Theil15,Lewars}.

Let $n \in \mathbb{N}$.  We indicate the \textit{configuration} of $n$ particles by 
$$C_n = \{ (x_1,q_1),\ldots,(x_{n},q_{n}) \} \subset  (\mathbb{R}^2 \times \lbrace -1,1 \rbrace)^n,$$
identified with the respective \emph{atom positions} $X_n =  (x_1,\ldots,x_n) \in    \mathbb{R}^{2n}$ together with their \emph{types} $Q_n =  ( q_1,\ldots,q_n ) \in   \lbrace -1,1\rbrace^n$.  Referring to a model for ionic compounds, we will frequently call $Q_n$ the \emph{charges} of the atoms, $q = 1$ representing  \emph{cations} and $q=-1$ representing \emph{anions}. Indeed, our choice of the empirical potentials (see below)  is inspired by ions in ionic compounds which are  primarily held together by the electrostatic forces between the net negative charge of the anions and net positive charge of the cations \cite{Pauling}. 

 We \EEE define the  phenomenological energy $\mathcal{E}:  (\mathbb{R}^2 \times \lbrace -1,1 \rbrace)^{n}$ $ \to \overline{\mathbb{R}}$ of a given configuration $\lbrace (x_1,q_1),\ldots,(x_n,q_n) \rbrace  \in (\mathbb{R}^2 \times \lbrace -1,1 \rbrace)^{n}$   by
\begin{align}\label{Energy}
\mathcal{E}(C_n) = \frac{1}{2}\sum_{\underset{q_i = q_j}{i \neq j}} V_{\mathrm{r}}(|x_i-x_j|) + \frac{1}{2}\sum_{\underset{q_i \neq q_j}{i \neq j}} V_{\mathrm{a}}(|x_i-x_j|),
\end{align}
where $V_{\mathrm{r}}, V_{\mathrm{a}} : [0,+\infty)\to \overline{\mathbb{R}}$ are a \emph{repulsive potential} and an \emph{attractive-repulsive potential}, respectively. The factor $\frac{1}{2}$  accounts for the fact that every contribution is counted twice in the sum.  The two potentials are pictured schematically in Fig.~\ref{FigurePotentials}.  Let $r_0 \in [1,(2\sin(\frac{\pi}{7}))^{-1})$  and note that $r_0 \le \sqrt{2}$.  The  attractive-repulsive potential $V_{\mathrm{a}}$ satisfies
\begin{align*}
{\rm [i]}& \ \ V_{\mathrm{a}}(r) =+\infty \text{ for all $r < 1$},\\
{\rm [ii]}& \ \ V_{\mathrm{a}}(r)=-1 \text{ if and only if $r=1$ and $V_{\mathrm{a}}(r) >-1$ otherwise},\\
{\rm [iii]}& \ \  V_{\mathrm{a}}(r) \leq 0 \text{ for all } r\geq 1   \text{ with equality for all $r > r_0$.}
\end{align*}
The distance $r =1 $ represents the (unique) equilibrium distance of two atoms with opposite charge. The constraint $V_{\mathrm{a}}(r) =+\infty $ is usually called the \emph{hard-interaction}  or the  \emph{hard sphere}  assumption, see e.g.\ \cite{Blanc, Mainini}. \EEE The choice of $V_{\rm a}$ reflects a balance between a long-ranged Coulomb attraction and  a short-ranged repulsive force when a pair of ions comes too close to each other. Assumption [iii] restricts the interaction range and ensures that the \emph{bond graph} (see Section \ref{sec: prelimi}) is planar. \EEE

The repulsive potential $V_{\mathrm{r}}$ satisfies
\begin{align*}
{\rm [iv]} & \ \  \text{$V_{\mathrm{r}}(r) = +\infty$ for all $r<1$ and $0 \leq V_{\mathrm{r}}(r) <+\infty$ for all $r\geq 1$},\\
{\rm [v]} & \ \  \text{$V_{\mathrm{r}}$ is non-increasing and convex for $r \geq 1$,}\\
{\rm [vi]} & \ \   \text{$V_{\mathrm{r}}(\sqrt{2}r_0) > 3$,}  \\
{\rm [vii]} & \ \  \text{$V_{\mathrm{r}}(r) = 0 $ iff $r \geq \sqrt{3}$.}
\end{align*}
Assumption [iv] is the hard-interaction assumption for the repulsive potential. The natural assumption [v] is satisfied for example for repulsive Coulomb interactions. We remark that some quantitative requirement of the form [vi] and [vii] are indispensable to obtain a crystallization result in the hexagonal lattice.  Indeed, for vanishing $V_{\rm r}$, ground states could be patches of the triangular lattice. Other quantitative assumptions on the repulsive potential will favor an assemblence of the atoms  in the square-lattice as we prove in \cite{FriedrichKreutzSquare}. Also note that two-body pair interactions for identical particle systems  often \EEE favor crystallization in the triangular lattice \cite{Radin, Theil},  when the interaction is of attractive-repulsive-type (instead of pure repulsive-type),  e.g.\ for the   Lennard Jones or the Morse potential. \EEE  This reflects the  aforementioned \EEE fact that the geometry of molecular compounds often differs from that of their
components.

A main assumption is [vii], i.e., the repulsion vanishes at $\sqrt{3}$ (the distance between second neighbors in the hexagonal lattice). Note that, if $V_{\rm r}(\sqrt{3})$ is  instead assumed  to be positive and sufficiently large, then crystallization in the hexagonal lattice is not expected as, e.g., a one-dimensional chain of atoms  with alternating charges \EEE is energetically favorable. For small, positive $V_{\rm r}(\sqrt{3})$ we still expect crystallization in the hexagonal lattice, but the analysis is much more demanding and beyond the scope of the present contribution.

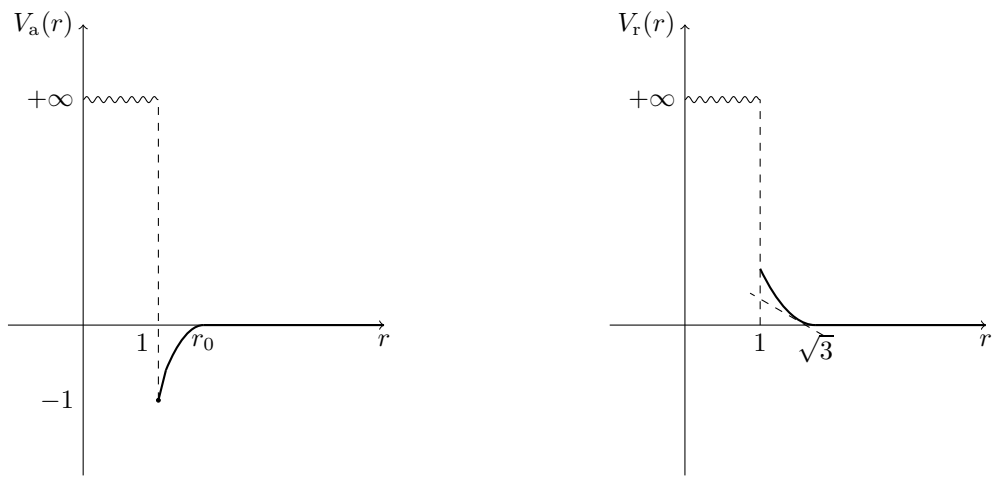
\begin{figure}[h]
\centering
\begin{tikzpicture}
\draw[->](0,-2)--++(0,6) node[anchor= east] {$V_{\mathrm{a}}(r)$};
\draw[->](-1,0)--++(5,0) node[anchor =north] {$r$};
\draw[decorate, decoration={snake,amplitude=.4mm,segment length=1.5mm}] (0,3) node[anchor =east]{$+\infty$}--++(1,0);
\draw[dashed,thin](1,-1) --++(0,1) node[anchor =north east] {$ 1 $}--++(0,3);
\draw[fill=black](0,-1) node[anchor = east] {$-1$}++(1,0) circle(.025);
\draw[thick](1.6,0)--++(2.4,0);
\draw[thick](1,-1)--(1.1,-0.6);
\draw[thick](1.1,-0.6) parabola[bend at end] (1.6,0) node[anchor=north]{$r_0$};

\begin{scope}[shift={(8,0)}]
\draw[->](0,-2)--++(0,6) node[anchor= east] {$V_{\mathrm{r}}(r)$};
\draw[->](-1,0)--++(5,0) node[anchor =north] {$r$};
\draw[decorate, decoration={snake,amplitude=.4mm,segment length=1.5mm}] (0,3) node[anchor =east]{$+\infty$}--++(1,0);
\draw[dashed,thin](1,0) node[anchor =north] {$ 1 $}--++(0,3);
\draw[thick](1,.75) parabola[bend at end] ({sqrt(3)},0) node[anchor=north]{$\sqrt{3}$};
\draw[thick]({sqrt(3)},0)--(4,0);
\draw[dashed,thin]({sqrt(3)},-.075)++(150:-.125)--++(150:1.125);
\end{scope}

\end{tikzpicture}
\caption{The potentials $V_{\mathrm{a}}$ and $V_{\mathrm{r}}$.}
\label{FigurePotentials}
\end{figure}

Finally, we require the following \emph{slope conditions} \EEE
\begin{itemize}
\item[{\rm[viii]} ] $ \displaystyle
V'_{\mathrm{r},-}(\sqrt{3}) < {-\frac{3}{\pi}},  \ \ \ \ \ \ \    \frac{1}{r-1}(V_{\mathrm{a}}(r)- V_{\mathrm{a}}(1)) > -2V'_{\mathrm{r},+}(1)   \ \ \ \text{for all } r \in (1,r_0],
$
\end{itemize}
where the functions $V_{\mathrm{r},+}',V_{\mathrm{r},-}'$  denote the right and left derivative, respectively.  (They exist due to convexity of $V_{\rm r}$.)  These conditions are reminiscent of  the \emph{soft-interaction} assumption by {\sc Radin} \cite{Radin} and the \emph{slope condition} for an angular potential by {\sc Mainini and Stefanelli} \cite{Mainini}. In particular, we assume that the repulsion grows linearly out of $\sqrt{3}$ and that, roughly speaking,   the slope of $V_{\mathrm{a}}$ is steep enough compared to the slope of $V_{\mathrm{r}}$. We highlight that without assumptions of this kind  finite crystallization is presently not known, not even for identical particle systems.  Assumption [viii] \EEE is only needed in Lemma \ref{LemmaBoundaryEnergy} where the energy contribution of atoms at the boundary of the configuration  is \EEE estimated. See also Remark \ref{rem: angle/bond}  for an in-depth explanation of the assumption and \EEE for a comparison of our model to \EEE \cite{Mainini}.  Note that [viii] also restricts the choice of possible $r_0$. 
 Assumptions  [ii]   and [vii] ensure that the (infinite) hexagonal lattice (see  \eqref{eq: hex-lattice} below) is \emph{stress free}. An assumption of this kind is necessary as otherwise for finite particle systems surface relaxation at the boundary of the configuration could occur, eventually ruling out \EEE finite crystallization in any lattice. \EEE

 We remark that the assumptions are chosen here for the sake of maximizing simplicity rather than generality. Some conditions, e.g., about the hard-interaction assumption or the exact value in [vi], could be weakened at the expense of more elaborated arguments. \EEE In particular, two different repulsive potentials for cations and anions could be considered as long as [iv]-[viii] hold for both potentials. 
  In the following we assume that conditions $[\mathrm{i}]$-$[\mathrm{viii}]$ are always satisfied.

\subsection{Basic notions}\label{sec: prelimi} 

In this section we collect some basic notions. Consider a configuration $C_n  \in  (\mathbb{R}^{2} \times \{\ -1,1\})^{n}$ with finite energy consisting of the positions $X_n  =  (x_1,\ldots,x_n) \EEE \in \mathbb{R}^{2n}$ and the charges $Q_n =  ( q_1,\ldots,q_n )  \EEE \in \lbrace -1,1 \rbrace^n$.

\textbf{Neighborhood, bonds, angles:}  For $i \in \{1,\ldots,n\}$ we define the \textit{neighborhood of} $x_i\in \mathbb{R}^2$ by
\begin{align}\label{eq: neighborhood}
\mathcal{N}(x_i) = \left(X_n \setminus \{x_i\}\right) \cap \lbrace x \in \mathbb{R}^2: \ |x-x_i| \le r_0 \rbrace,
\end{align}
 where $r_0$ is defined in [iii]. \EEE   If $x_j \in \mathcal{N}(x_i)$, we say that $x_i$ and $x_j$ are \emph{bonded}. We will say that $x_i$ is $k$-bonded if $\# \mathcal{N}(x_i) = k$. Given $x_j,x_k \in \mathcal{N}(x_i)$, we define the \textit{bond-angle} between $x_j,x_i,x_k$ as the angle between the two vectors $x_k-x_i$ and $x_j-x_i$ (choose anti-clockwise orientation, for
definiteness). In general, we say that it is an angle at $x_i$. 
\EEE

\textbf{The bond graph:}  The set of   atomic positions $X_n \subset \mathbb{R}^{2n}$ together with the set of  bonds $\{\{x_i,x_j\}: x_j \in \mathcal{N}(x_i)\}$   forms a graph which we call the  \textit{bond graph}. Since   for finite energy configurations we get \EEE $\mathrm{dist}(x_i,X_n \setminus \{x_i\})\geq 1$ and $x_j \in \mathcal{N}(x_i)$ only if $|x_i-x_j| \leq r_0<\sqrt{2}$, we have that   their bond   graph is  planar.  Indeed, given a quadrangle with all sides and one diagonal in $[1,r_0]$, the second diagonal is at least $\sqrt{2} >r_0$. Note that assumption [iii] states that the attractive interactions are restricted to nearest neighbors in the bond graph only. If no ambiguity arises, the number of bonds in the bond graph will be denoted by $b$, i.e.,
\begin{align}\label{eq:b}
b= \# \{\{x_i,x_j\}: x_j \in \mathcal{N}(x_i)\}.
\end{align} 
We say a configuration is \emph{connected} if each two atoms are joinable through a simple path in the bond graph.  \EEE Any simple cycle of the bond graph is a \textit{polygon}.

 \textbf{Acyclic bonds:} We say that a bond is \textit{acyclic} if it is not contained in any simple cycle of the bond graph. Among acyclic bonds we distinguish between \textit{flags} and \textit{bridges}. We call a bridge an acyclic bond which is contained in some simple path connecting two vertices which are included in two distinct cycles. All other acyclic bonds are  called flags, see Fig.~\ref{FigureFlagsBridges}.
 
 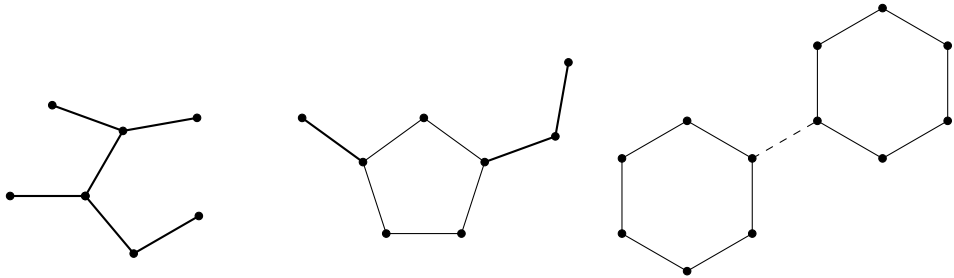
\begin{figure}[htp]
\centering
\begin{tikzpicture}
\draw[thick](0,0)--++(1,0)--++(60:1)--++(10:1);
\draw[thick](1,0)--++(-50:1)--++(30:1);
\draw[thick](1,0)++(60:1)--++(160:1);

\draw[fill=black](0,0)circle(.05)++(1,0)circle(.05)++(60:1)circle(.05)++(10:1)circle(.05);
\draw[fill=black](1,0)circle(.05)++(-50:1)circle(.05)++(30:1)circle(.05);
\draw[fill=black](1,0)++(60:1)++(160:1)circle(.05);

\draw[fill=black](5,-.5)circle(.05)++(1,0)circle(.05)++(72:1)circle(.05)++(144:1)circle(.05)++(216:1)circle(.05)++(288:1);

\draw(5,-.5)--++(1,0)--++(72:1)--++(144:1)--++(216:1)--++(288:1);

\draw[thick](5,-.5)++(1,0)++(72:1)--++(20:1)--++(80:1);

\draw[fill=black](5,-.5)++(1,0)++(72:1)++(20:1)circle(.05)++(80:1)circle(.05);

\draw[fill=black](5,-.5)++(1,0)++(72:1)++(144:1)++(216:1)++(144:1)circle(.05);
\draw[thick](5,-.5)++(1,0)++(72:1)++(144:1)++(216:1)--++(144:1);

\draw(9,-1)--++(30:1)--++(90:1)--++(150:1)--++(210:1)--++(270:1)--++(330:1);

\draw(9,-1)++(30:1)++(90:1)++(210:-1)--++(330:1)--++(30:1)--++(90:1)--++(150:1)--++(210:1)--++(270:1);

\draw[fill=black](9,-1)circle(.05)++(30:1)circle(.05)++(90:1)circle(.05)++(150:1)circle(.05)++(210:1)circle(.05)++(270:1)circle(.05)++(330:1);

\draw[fill=black](9,-1)++(30:1)++(90:1)++(210:-1)circle(.05)++(330:1)circle(.05)++(30:1)circle(.05)++(90:1)circle(.05)++(150:1)circle(.05)++(210:1)circle(.05)++(270:1);

\draw[dashed](9,-1)++(30:1)++(90:1)--++(210:-1);
\end{tikzpicture}
\caption{Examples of flags (bold) and a bridge (dashed)}
\label{FigureFlagsBridges}
\end{figure}

 \textbf{Defects:}  By elementary polygons we denote polygons which do not contain any non-acyclic bonds in its interior region. \EEE We call an  elementary \EEE polygon in the bond graph which is not a hexagon a \emph{defect}. A configuration is said to be \emph{defect-free} if all its  elementary \EEE polygons are hexagons. We also introduce the  \textit{excess of edges} $\eta$ by
\begin{align}\label{Excess}
\eta= \sum_{j\geq 6} (j-6)   f_j,  
\end{align}
 where $f_j$ denotes the number of  elementary \EEE polygons with $j$ vertices in the bond graph. The excess of edges is a tool to quantify the number of {defects} in the bond graph. \EEE { Note that the summation in (\ref{Excess}) runs over $j\geq 6$. This is due to the fact that we use this definition only for configurations whose bond graph contains only $k$-gons with $k \geq 6$,  cf.~Lemma \ref{LemmaPolygon}.  We remark that in \cite{Lucia} the  \textit{excess of edges} (with respect to the triangular lattice) is referred to as \emph{defect measure}.}
 
 In the following we refer to $C_n$ instead of $X_n$ when speaking about its bond graph, acyclic bonds, or connectedness properties, when no confusion may occur. 
 
\textbf{Charges:} We say that a configuration satisfying
\begin{align}\label{SamechargeNeighbourhood}
 \mathcal{N}(x_i) \cap \left\{x_j \in X_n: q_j=q_i\right\} = \emptyset  \text{ for all } i \in \{1,\ldots,n\}
\end{align}
 has \textit{alternating charge distribution}. Moreover, we call a configuration \emph{repulsion-free} if $|x_i - x_j| \ge \sqrt{3}$ for all $x_i \neq x_j$ with $q_i = q_j$. The \emph{net charge} of a configuration is defined  as the (signed) difference of the number of the two atomic types, i.e.,   
 \begin{align} \label{DefinitionNetCharge}
  \mathcal{Q}(C_n) := \sum_{i=1}^n q_i. 
 \end{align}

\EEE

\textbf{Ground state:} A configuration $C_n$ is called \textit{ground state} for the interaction energy (\ref{Energy}) if for all $C'_n \subset ( \mathbb{R}^{2} \times \{ -1, 1\})^{n}$ there holds
\begin{align*}
\mathcal{E}(C_n) \leq \mathcal{E}(C_n').
\end{align*}
In other words, $C_n$ minimizes (\ref{Energy}) among all possible configurations consisting of $n$ atoms.

\begin{figure}[htp]
\centering
\begin{tikzpicture}[scale=0.8]
\foreach \i in {0,...,3}{
\foreach \k in {0,...,4}{
\foreach \j in {0,2,4}{
\draw[thin](\k*1.5,{\k*sqrt(3)*.5})++(-\i*1.5,{\i*sqrt(3)*.5})++(\j*60:1)--++(\j*60+120:1);
\draw[fill=black](\k*1.5,{\k*sqrt(3)*.5})++(-\i*1.5,{\i*sqrt(3)*.5})++(\j*60:1) circle(.05);
}
\foreach \j in {1,3,5}{
\draw[thin](\k*1.5,{\k*sqrt(3)*.5})++(-\i*1.5,{\i*sqrt(3)*.5})++(\j*60:1)--++(\j*60+120:1);
\draw[fill=white](\k*1.5,{\k*sqrt(3)*.5})++(-\i*1.5,{\i*sqrt(3)*.5})++(\j*60:1) circle(.05);

}
}
}
\end{tikzpicture}
\caption{A configuration with atomic positions forming a subset of the hexagonal lattice. The configuration has alternating charge distribution where white indicates $q=1$ and black indicates $q=-1$. \EEE }
\label{FigureHexagonal-new}
\end{figure}
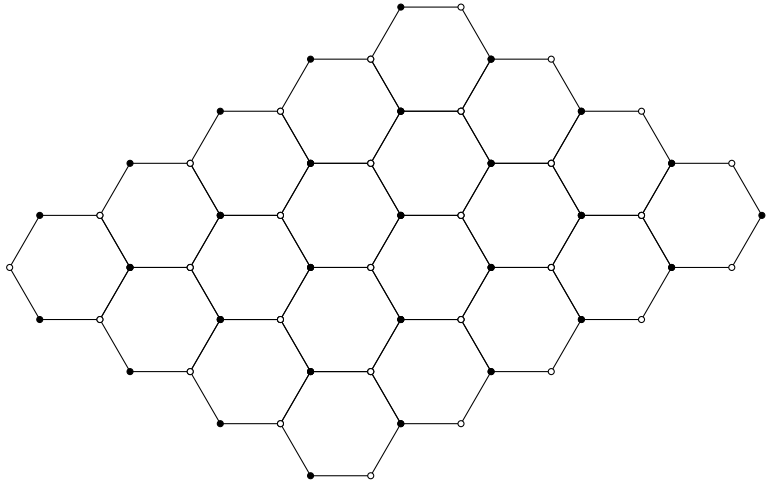

\textbf{The hexagonal lattice:} We define the \textit{hexagonal lattice}  by 
\begin{align}\label{eq: hex-lattice}
 \mathcal{L} := \EEE \left\{p v_1 + q v_2 + r v_0 : p,q \in \mathbb{Z}, r \in \{0,1\}\right\},
\end{align}
with $v_1 = (0,\sqrt{3})$, $v_2=\frac{1}{2}(3,\sqrt{3})$ and $v_0=(1,0)$. We identify it with its bond graph defined above.

Note that the hexagonal lattice is planar, connected, and all edges have unit length. Since it is bipartite, we can associate to all positions $X_n\subset \mathcal{L}$ of atoms a competitor for the minimization of (\ref{Energy}) by two-coloring (and hereby choosing the charge of)  the atoms.   The corresponding configuration $C_n$ \EEE is pictured in Fig.~\ref{FigureHexagonal-new}.  \EEE Then by assumption [ii], [iii], and [vii] \EEE we have  
\begin{align*}
\mathcal{E}(C_n) = -b,
\end{align*}
since all atoms of same charge have at least distance $\sqrt{3}$ and all atoms of $X_n$ are  bonded only to atoms of \EEE opposite charge. This means that for subsets of the hexagonal lattice the energy is computed (up to sign) by counting the number of bonds.

\subsection{Main results}

In this section we state our main results. We will derive a rigorous planar crystallization result in the  spirit of \cite{HR,  Mainini-Piovano, Mainini, Radin}. Afterwards, we will investigate the \emph{net charge} of ground-state configurations. Our first result characterizes the energy of ground states and shows that  all ground states are connected. To this end, we introduce the function 
\begin{align}\label{eq: beta definition}
\beta(n) := \frac{3}{2}n - \sqrt{\frac{3}{2}n} 
\end{align} 
for $n \in \mathbb{N}$. By $\lfloor t \rfloor$ we denote the integer part of $t \in \mathbb{R}$.  \EEE

\begin{theorem}[Ground-state energy]\label{TheoremEnergyGroundstates} Let $n \geq 1$ and let $C_n$ \EEE be a ground state. Then $C_n$ is connected and satisfies \EEE
\begin{align}\label{Energygroundstates}
\mathcal{E}(C_n) = \EEE -b  = -\lfloor \beta(n)\rfloor.
\end{align}
\end{theorem}

\begin{remark}\label{rem: repulsionsfree}
{\normalfont 
In view of assumptions ${\rm [ii]}$ and [vii], we have that $\mathcal{E}(C_n) \geq -b$ with equality if and only if the configuration is repulsion-free and all bonds have unit length. In particular, Theorem \ref{TheoremEnergyGroundstates}   implies that ground states satisfy both   properties.
}
\end{remark}

\EEE

The next result states that ground states are essentially subsets of the hexagonal lattice with alternating charge distribution. Moreover, they satisfy some topological properties. Without further notice, all following statements regarding the geometry of ground states hold up to isometry. \EEE

\begin{theorem}[Characterization of ground states]\label{TheoremGroundstatesleq31}  Let $n \geq 30$ and let $C_n$ be a ground state. \EEE Then except for possibly one flag, \EEE $C_n$ is a connected subset of the hexagonal lattice  with alternating charge distribution. Moreover, $C_n$  is defect-free and does not contain any bridges. \EEE
\end{theorem}

\begin{remark}\label{rem: main}
{\normalfont  (i)  More precisely, we show that the bond graph of ground states  consists only of hexagonal cycles except for at most two flags. Since ground states are repulsion-free and all bonds have unit length, at least one of the flags is contained in the hexagonal lattice.  If two flags exist and are connected, we note that one of them can be rotated in a continuous way without changing the energy.

(ii) The ground states for $n \le 29$ can also be characterized, but due to the smallness of the structures, more degeneracies can occur. In particular, for $n=8,9,12,15, 18, 21, 29$ \EEE there might be one octagon at the boundary. (We refer to Remark \ref{rem: octogons} for more details.) For  $n < 10$   the structure can be much more flexible, see Fig.~\ref{FigureFlexible}. 
}
\end{remark}
\EEE

Our third main result addresses the net charge. Recall  (\ref{DefinitionNetCharge}).
 
\begin{theorem}[Net charge]\label{TheoremCharge} 
The following properties for the net charge of ground states hold:
\begin{itemize}
\item[(i)] (Charge control) There is a universal constant $c>0$ such that for all $n \in \mathbb{N}$ and all ground states $C_n$ the charge satisfies $|\mathcal{Q}(C_n)| \le cn^{1/4}$.

\item[(ii)] (Sharpness of the $n^{1/4}$-scaling) There exists  an increasing \EEE sequence of integers $(n_j)_j$ and ground states $(C_{n_j})_j$ such that 
$$\liminf_{j \to +\infty} n_j^{-1/4}|\mathcal{Q}(C_{n_j})|>0. $$
\end{itemize}
\end{theorem}

The proof of our main results will be given in Sections \EEE 3-8. In particular, in Section 4 we construct explicitly configurations and give the upper bound for the ground-state energy  (Theorem \ref{TheoremEnergyGroundstates}) as well as the proof of Theorem \ref{TheoremCharge}(ii). In Section 6 we conclude the proof of Theorem \ref{TheoremEnergyGroundstates} by providing the lower bound. Section 7 is then devoted to the characterization of ground-state geometries for  $n \ge 30$ \EEE (Theorem \ref{TheoremGroundstatesleq31}). Finally, in Section 8 we prove Theorem \ref{TheoremCharge}(i).

\section{Elementary geometric properties}

 In this section we provide   some elementary geometric facts  for ground states independently of $n \in \mathbb{N}$.  \EEE
In the proofs, we will use the following convention: we say that  we \textit{relocate} $(x,q) \in C_n$,   and write $C_n-\{(x,q)\}$, \EEE   by considering the configuration   $(C_n\cup \{(x+\tau,q)\})\setminus \{(x,q)\} $ \EEE, where $\tau \in \mathbb{R}^2$ is chosen such that
\begin{align*}
\mathrm{dist}(X_n,x+\tau)\geq \sqrt{3}.
\end{align*}
We relocate a set of atoms $A \subset C_n$ by  relocating successively every $(x,q) \in A$ and write $C_n -A$. Note that $ C_n-A$ still consists of $n$ particles. \EEE

The first lemma addresses the neighbors. Recall \eqref{eq: neighborhood} and \eqref{SamechargeNeighbourhood}.\EEE
\begin{lemma}[Neighbors and charge distribution\EEE]\label{LemmaNeighborhood}
 Let $C_n$ be a ground state. Then $C_n$ has alternating charge distribution and \EEE
\begin{align}\label{NeighbourhoodCard}
\#\mathcal{N}(x_i) \leq 3 \text{ for all } i \in \{1,\ldots,n\}.
\end{align}
\end{lemma}

\begin{proof} Since  $C_n$ \EEE is a ground state, it holds that $\mathcal{E}(C_n)<+\infty$. Therefore,  by assumption [i],[iv] \EEE
\begin{align}\label{eq: distance}
\mathrm{dist}(x_i,  X_n \EEE \setminus\{x_i\}) \geq 1 \text{ for all }i\in \{1,\ldots,  n \EEE \}.
\end{align} 
 For brevity, we define $\mathcal{N}_{\rm rep}(x_i) =   \mathcal{N}(x_i) \cap \left\{x_j \in   X_n: q_j=q_i\right\}$  for all $i \in \lbrace 1 ,\ldots, n \rbrace$. We give the proof of the statement in two \EEE steps. \EEE

\medskip

\textbf{Claim 1:}  We have \EEE
\begin{align*}
\# \mathcal{N}_{\rm rep}(x_i)  \leq 1 \text{ for all } i \in \{1,\ldots,n\},
\end{align*}
and if $\#\mathcal{N}(x_i) \leq 4$, then
\begin{align*}
\# \mathcal{N}_{\rm rep}(x_i)  =0.
\end{align*}
\noindent \emph{Proof of Claim 1}:  First, \eqref{eq: distance} \EEE and $r_0 < (2\sin(\frac{\pi}{7}))^{-1}$  entail by an elementary geometric argument that \EEE
\begin{align}\label{eq: six neighbors}
\#\mathcal{N}(x_i) \leq 6 \text{ for all } i \in \{1,\ldots,n\}.
\end{align}
 In fact, if $\#\mathcal{N}(x_i) \ge 7$, two neighbors of $x_i$ would necessarily have distance smaller than $1$. \EEE 
Now assume by contradiction that $\# \mathcal{N}_{\rm rep}(x_i) \geq 2$. Note that \EEE every bond between points of different charge contributes at least $-1$ to the energy by $[\mathrm{ii}]$.  This along with  $V_{\rm r} \ge 0$  and the fact that   the energy per neighbor of same charge is  larger than   $3$ (see  $[\mathrm{v}]$   and   $[\mathrm{vi}]$)  allows us to relocate $(x_i,q_i)$:   we \EEE obtain
 \begin{align*}
\mathcal{E}(  C_n - \{(x_{i},q_i)\} ) <   \mathcal{E}(C_n) + \#\left(\mathcal{N}(x_i) \setminus \mathcal{N}_{\rm rep}(x_i)\right) - 3\#\mathcal{N}_{\rm rep}(x_i) \le \mathcal{E}(C_n) + 4 - 6    < \EEE  \mathcal{E}(C_n).
\end{align*} 
 This contradicts the fact that $C_n$ is a ground state.  The argument in the case $\#\mathcal{N}(x_i)\leq 4$ and $\#\mathcal{N}_{\rm rep}(x_i) \geq 1$ is similar: \EEE due to the fact that every bond between points of different charge contributes at least $-1$ to the energy and the contribution of neighbors of the same charge is larger than $3$, \EEE   we have  
 \begin{align*}
\mathcal{E}(  C_n - \{(x_{i},q_i)\} \EEE ) <   \mathcal{E}(C_n) + \#\left(\mathcal{N}(x_i) \setminus \mathcal{N}_{\rm rep}(x_i)\right) - 3\#\mathcal{N}_{\rm rep}(x_i) \le \mathcal{E}(C_n) + 3 - 3 \EEE   = \EEE  \mathcal{E}(C_n).
\end{align*}
 This contradiction shows that \EEE Claim 1 holds true.

 \textbf{Claim 2:} \EEE $\#\mathcal{N}(x_i) \leq 3$ for all $i \in \{1,\ldots,n\}$.

\noindent \emph{Proof of Claim 2}: Assume by contradiction that there exists $i \in \{1,\ldots,n\}$ such that $\#\mathcal{N}(x_i) \geq 4$. By Claim 1 we may suppose that there exist  \EEE  $\{x_0,\ldots,x_3\}\subset \mathcal{N}(x_{i})$ with $q_j =-q_i$ $j=0,\ldots,3$. We let $\theta_j \in [0,2\pi)$ be the angle between $x_j,x_i,x_{j+1}$.  (Here and in the following the indices have to be understood modulo $4$.) \EEE We can choose $j_0 \in \{0,\ldots,3\}$ such that
\begin{align*}
\theta_{j_0}+\theta_{j_0-1} \leq \frac{1}{4}\sum_{j=0}^3 (\theta_{j}+\theta_{j-1} ) = \pi.
\end{align*} 
 Note that  $|x_{j_0}-x_{j_0+1}| \leq  2r_0 \sin\left(\theta_{j_0}/2\right), |x_{j_0}-x_{j_0-1}| \leq  2r_0\sin\left(\theta_{j_0-1}/2\right)$. By concavity of $\sin(x)$ for $x \in [0,\pi]$ and the fact that $\sin(x)$ is   increasing  for $x\in [0,\frac{\pi}{2}]$,  we have 
\begin{align*}
|x_{j_0}-x_{j_0+1}| + |x_{j_0}-x_{j_0-1}| &\leq r_0\left( 2\sin\left(\frac{\theta_{j_0}}{2}\right) + 2\sin\left(\frac{\theta_{j_0-1}}{2}\right)\right)\notag \\&\leq 4r_0 \sin\left(\frac{\theta_{j_0}+\theta_{j_0-1}}{4}\right)\leq 4r_0 \sin\left(\frac{\pi}{4}\right)=2\sqrt{2}r_0.
\end{align*} 
Since  $V_{\mathrm{r}}$ is convex and non-increasing (see $[\mathrm{v}]$), we find 
\begin{align*}
V_{\mathrm{r}}\left(|x_{j_0}-x_{j_0+1}|\right) + V_{\mathrm{r}}\left(|x_{j_0}-x_{j_0-1}|\right) \geq 2V_{\mathrm{r}}\left(\frac{1}{2}\left(|x_{j_0}-x_{j_0+1}|+ |x_{j_0}-x_{j_0-1}|\right)\right) \geq 2V_{\mathrm{r}}(\sqrt{2}r_0).
\end{align*}

  Using  $[\mathrm{ii}]$,  $[\mathrm{vi}]$, \EEE $V_{\rm r} \ge 0$,  and \eqref{eq: six neighbors} \EEE we finally observe 
 \begin{align*}
 \mathcal{E}(C_n-\{ (x_{j_0},q_{j_0})\})&\leq  \mathcal{E}(C_n)-2V_{\mathrm{r}}(\sqrt{2}r_0) +  \# \mathcal{N}(x_{j_0}) \EEE  \leq  \mathcal{E}(C_n)-2V_{\mathrm{r}}(\sqrt{2}r_0) +  6    < \mathcal{E}(C_n).
\end{align*}
 This contradicts the fact that $C_n$ is a ground state and concludes the proof of    Claim 2. \EEE

 Now, \eqref{NeighbourhoodCard}  holds, see \EEE Claim 2. \EEE The property of alternating charge (see \eqref{SamechargeNeighbourhood}) \EEE follows from \eqref{NeighbourhoodCard} and  the second statement of Claim 1.  This   concludes \EEE the proof of Lemma \ref{LemmaNeighborhood}.
\end{proof}

 We now investigate simple cycles in the bond graph. \EEE

\begin{lemma}[Polygons\EEE]\label{LemmaPolygon} Let $C_n$ be a ground state. Then all polygons have at least $6$ edges and the number of edges is always even. \EEE 
\end{lemma}
\begin{proof}

As $C_n$ has alternating charge distribution,  two \EEE successive vertices of a path of the bond graph of a ground state have to be of different charge. This prohibits cycles of odd length. In particular, \EEE this rules out triangles and pentagons. We are thus   left to prove that there is no simple square. \EEE

Now assume by contradiction that the bond graph of a ground state contains a simple square. As two successive vertices of the square have  different charge, \EEE the vertices connected by the diagonal have the same charge. Denote the lengths of the two diagonals by $d_1$ and \EEE $d_2$, respectively. An elementary computation shows
\begin{align} \label{ineq: di}
d_1^2+d_2^2 \leq 4r_0^2.
\end{align}
In fact, denoting the vertices of the square by $x_1,\ldots,x_4$ (counterclockwise) one can use the elementary expansion 
$$|x_1 - x_3|^2 + |x_2 - x_4|^2 = |x_1 - x_2|^2 + |x_2 - x_3|^2 + |x_3 - x_4|^2 + |x_4 - x_1|^2 - |x_1 - x_2 + x_3 - x_4|^2    $$
along with the fact that each bond length is smaller or equal to $r_0$. \EEE Now (\ref{ineq: di}) together with Young's inequality gives 
\begin{align*}
(d_1 +d_2)^2 \leq 2(d_1^2 +d_2^2)\leq 8r_0^2.
\end{align*}
Since $V_{\mathrm{r}}$ is convex and non-increasing  (see ${\rm [v]}$), \EEE we obtain
\begin{align*}
V_{\mathrm{r}}(d_1) +V_{\mathrm{r}}(d_2)\geq 2 V_{\mathrm{r}}\left(\frac{1}{2}( d_1+d_2)\right) \geq 2V_{\mathrm{r}}(\sqrt{2}r_0).
\end{align*}

\begin{figure}[htp]
\centering
\begin{tikzpicture}[scale=.8]
\draw[thin](0,0)node[anchor =east]{$x_1$}--++(1,0)--++(70:1)--++(-1,0)--++(70:-1);
\draw[thin](0,0)--++(220:1);
\draw[fill= white](0,0)++(220:1) circle(.05);
\draw[thin](1,0)--++(-20:1);
\draw[fill= black](1,0)++(-20:1) circle(.05);
\draw[thin](70:1)--++(100:1);
\draw[fill= black](70:1)++(100:1) circle(.05);
\draw[thin](1,0)++(70:1)--++(50:1);
\draw[fill= white](1,0)++(70:1)++(50:1) circle(.05);
\draw[fill=black](0,0)circle(.05)++(1,0)++(70:1)circle(.05);
\draw[fill=white](0,0)++(1,0)circle(.05)++(70:1)++(180:1)circle(.05) node[anchor =east]{$x_2$};
\end{tikzpicture}
\caption{ Relocating $(x_1,q_1)$ and $(x_2,q_2)$ \EEE strictly decreases the energy.}
\label{FigureBondgraphsquare}
\end{figure}
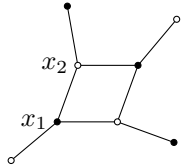

We relocate two  neighboring particles in the square, denoted by $(x_1,q_1)$  and $(x_2,q_2)$.  \EEE  By \eqref{NeighbourhoodCard} and the fact that  $x_1$ and $x_2$ share a bond, we observe that hereby \EEE we remove at most $5$ bonds between atoms \EEE of different charge, see Fig.~\ref{FigureBondgraphsquare}. Using $[\mathrm{ii}]$,   $[\mathrm{vi}]$, \EEE and $V_{\rm r} \ge 0$, we obtain
\begin{align*}
\mathcal{E}(C_n- \{(x_1,q_1),(x_2,q_2)\}) \leq  \mathcal{E}(C_n) -2V_{\mathrm{r}}(\sqrt{2}r_0) + 5 <   \mathcal{E}(C_n) -1 <  \mathcal{E}(C_n).
\end{align*}
This contradicts the fact that $C_n$ is a ground state and concludes the proof. 
\end{proof}

\begin{lemma}\label{lemma:bondangles} Let $C_n$ be a ground state. Then all the bond angles $\theta$ satisfy
\begin{align}\label{ineq:bondangles}
\frac{\pi}{3}\leq \theta \leq \frac{5\pi}{3}.
\end{align}
\end{lemma}
\begin{proof} Assume  by contradiction  that there exists a bond angle $\theta$ between $x_1,x_0,x_2$ satisfying $\theta < \frac{\pi}{3}$.  As $|x_1-x_0|, |x_2-x_0|\leq r_0$,  an easy trigonometric argument shows that then also $|x_1-x_2| \leq r_0$. Therefore,  the points $\{x_0,x_1,x_2\}$ form a triangle in the bond graph. This contradicts Lemma \ref{LemmaPolygon} and the first inequality in \eqref{ineq:bondangles} follows. Now, if there exists a bond angle  $\theta > 5\pi/3$, then, since the bond angles at every point  sum to $2\pi$, there exists a bond angle $\tilde{\theta} < \pi/3$,  which is impossible.  Therefore, also the second inequality in \eqref{ineq:bondangles} is proved.
\end{proof}

\EEE

\section{Upper bound on the ground-state energy}
This section is devoted to the explicit construction of  configurations $D_n$ \EEE with alternating charge distribution which are \EEE  subsets of the hexagonal lattice. These configurations provide a reference energy value for every $n$, namely $\mathcal{E}(D_n) \EEE=-\lfloor\beta(n)\rfloor$, cf. \eqref{eq: beta definition}.  This already gives the upper bound in \eqref{Energygroundstates}. Moreover, we will construct explicitly configurations  $C_n$ \EEE with net charge (see  \eqref{DefinitionNetCharge}) \EEE of the order $n^{1/4}$ which establishes \EEE Theorem \ref{TheoremCharge}(ii). We defer the lower bound on the ground-state energy, the upper bound on the net charge, and the characterization of the ground states to the subsequent sections. \EEE

By the special geometry of the hexagonal lattice, it is quite natural to give an interpretation of the two terms appearing in $\beta$. The leading order term of the energy is given by $-\frac{3}{2}n$. It is related to the fact that each atom which is not contained in the exterior face of the bond graph is bonded to exactly \EEE three atoms of opposite \EEE charge and every bond connects two atoms. The repulsive term in the energy is zero for such configurations since the distance of each \EEE two atoms \EEE of the same charge is bigger or equal than $\sqrt{3}$. The additional lower order correction term in the energy  is \EEE due to the fact that a certain proportion of the \EEE atoms touching the exterior face of the bond graph is \EEE only bonded to  two atoms of opposite  charge. Their cardinality scales like $\sqrt{n}$.

\subsection{Special subsets of the hexagonal lattice}\label{sec: daisy}

We exhibit special configurations that are subsets of the hexagonal lattice with energy  \EEE $-\lfloor \beta(n)\rfloor$. This provides an upper bound for the ground-state energy. We give \EEE a recursive construction for these geometries following the ideas  in \cite[Section 4,5]{Mainini}. \EEE

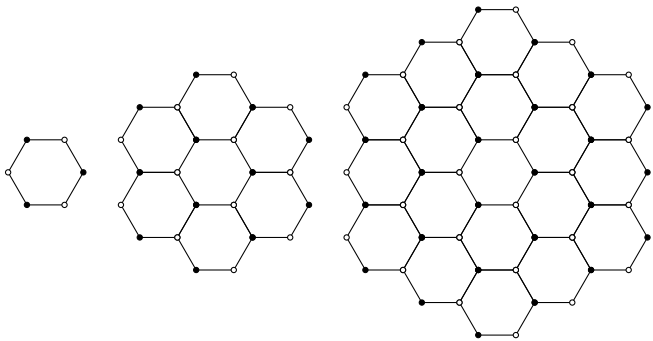
\begin{figure}[htp]
\centering
\begin{tikzpicture}[scale=.5]
\foreach \j in {0,2,4}{
\draw[thin](\j*60:1)--(\j*60+60:1);
\draw[fill=black](\j*60:1) circle(.07);
}
\foreach \j in {1,3,5}{
\draw[thin](\j*60:1)--(\j*60+60:1);
\draw[fill=white](\j*60:1) circle(.07);
}

\begin{scope}[shift={(4.5,0)}]
\foreach \k in {0,...,5}{
\foreach \j in {0,2,4}{
\draw[thin](\k*60+30:{sqrt(3)})++(\j*60:1)--++(\j*60+120:1);
\draw[fill=black](\k*60+30:{sqrt(3)})++(\j*60:1) circle(.07);
}
\foreach \j in {1,3,5}{
\draw[thin](\k*60+30:{sqrt(3)})++(\j*60:1)--++(\j*60+120:1);
\draw[fill=white](\k*60+30:{sqrt(3)})++(\j*60:1) circle(.07);
}
}

\end{scope}

\begin{scope}[shift={(12,0)}]
\foreach \i in {1,2}{
\foreach \k in {0,...,5}{
\foreach \j in {0,2,4}{
\draw[thin](\k*60+30:{\i*sqrt(3)})++(\j*60:1)--++(\j*60+120:1);
\draw[fill=black](\k*60+30:{\i*sqrt(3)})++(\j*60:1) circle(.07);
}
\foreach \j in {1,3,5}{
\draw[thin](\k*60+30:{\i*sqrt(3)})++(\j*60:1)--++(\j*60+120:1);
\draw[fill=white](\k*60+30:{\i*sqrt(3)})++(\j*60:1) circle(.07);

}
}
}

\foreach \k in {0,...,5}{
\foreach \j in {0,2,4}{
\draw[thin](\k*60:3)++(\j*60:1)--++(\j*60+120:1);
\draw[fill=black](\k*60:3)++(\j*60:1) circle(.07);
}
\foreach \j in {1,3,5}{
\draw[thin](\k*60:3)++(\j*60:1)--++(\j*60+120:1);
\draw[fill=white](\k*60:3)++(\j*60:1) circle(.07);

}
}

\end{scope}
\end{tikzpicture}
\caption{ $D_6$, $D_{24}$ and $D_{54}$.\EEE}
\label{FigureDaisies}
\end{figure}

First, we \EEE provide the construction for $n=6k^2$, $k\in \mathbb{N}$. For $k=1$ we define  $D_{6}$ \EEE to be a regular hexagon $\{x_1,\ldots,x_6\}$, where the points $x_1,\ldots,x_6$ are arranged in a counter-clockwise sense as the vertices of the hexagon and with charges $q_i=(-1)^{i+1}, i=1,\ldots,6$. Once we have constructed  $D_{6(k-1)^2}$\EEE, we construct  $D_{6k^2}$ \EEE by attaching hexagons on all boundary sides of  $D_{6(k-1)^2}$ \EEE such that every  atom  has two or three   atoms \EEE of opposite charge in its neighborhood. (This is possible since the hexagonal lattice is bipartite.) \EEE These configurations are pictured in Fig.~\ref{FigureDaisies}.  Due to their symmetry, these configurations are called \emph{daisies} and we will indicate their atomic positions also by $X_{6k^2}^{\rm daisy}$. Note that daisies have net charge zero. \EEE

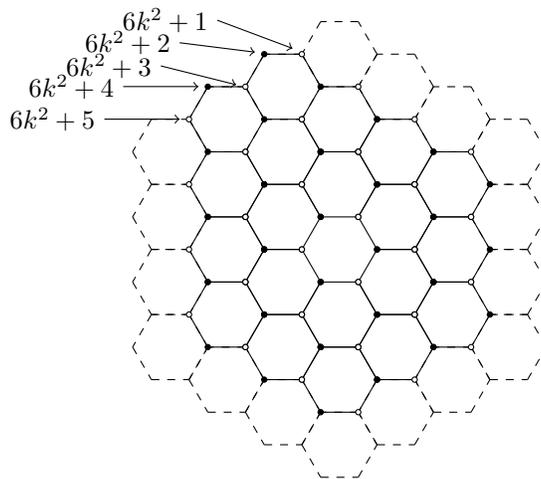
\begin{figure}[htp]
\centering
\begin{tikzpicture}[scale=0.5]

\foreach \i in {0,1,2}{
\foreach \k in {0,...,5}{
\foreach \j in {0,2,4}{
\draw[thin,dashed](\k*60+30:{3*sqrt(3)})++(\j*60:1)++(\k*60+150:{\i*sqrt(3)})--++(\j*60+120:1);
}
\foreach \j in {1,3,5}{
\draw[thin,dashed](\k*60+30:{3*sqrt(3)})++(\k*60+150:{\i*sqrt(3)})++(\j*60:1)--++(\j*60+120:1);

}
}
}

\draw[thin](120:3)++(60:2)--++(120:1)--++(180:1)--++(60:-1)--++(180:1)--++(60:-1);
\draw[thin](120:3)++(60:2)++(120:1)++(180:1)++(60:-1)--++(-60:1);
\draw[thin](120:3)++(60:2)++(120:1)++(180:1)++(60:-2)++(180:1)--++(-60:1);
\draw[fill=white](120:3)++(60:2)++(120:1) circle(.07);
\draw[fill=black](120:3)++(60:2)++(120:1)++(180:1) circle(.07);
\draw[fill=white](120:3)++(60:2)++(120:1)++(180:1)++(60:-1) circle(.07);
\draw[fill=black](120:3)++(60:2)++(120:1)++(180:2)++(60:-1) circle(.07);
\draw[fill=white](120:3)++(60:2)++(120:1)++(180:2)++(60:-2) circle(.07);
\foreach \i in {1,2}{
\foreach \k in {0,...,5}{
\foreach \j in {0,2,4}{
\draw[thin](\k*60+30:{\i*sqrt(3)})++(\j*60:1)--++(\j*60+120:1);
\draw[fill=black](\k*60+30:{\i*sqrt(3)})++(\j*60:1) circle(.07);
}
\foreach \j in {1,3,5}{
\draw[thin](\k*60+30:{\i*sqrt(3)})++(\j*60:1)--++(\j*60+120:1);
\draw[fill=white](\k*60+30:{\i*sqrt(3)})++(\j*60:1) circle(.07);

}
}
}

\foreach \k in {0,...,5}{
\foreach \j in {0,2,4}{
\draw[thin](\k*60:3)++(\j*60:1)--++(\j*60+120:1);
\draw[fill=black](\k*60:3)++(\j*60:1) circle(.07);
}
\foreach \j in {1,3,5}{
\draw[thin](\k*60:3)++(\j*60:1)--++(\j*60+120:1);
\draw[fill=white](\k*60:3)++(\j*60:1) circle(.07);

}
}
\draw[->](-3.25,6.1)node[anchor =east]{$6k^2+1 $}--++(2,-.75);
\draw[->](-4.25,5.5)node[anchor =east]{$6k^2+2 $}--++(2,-.25);
\draw[->](-4.25,5.75)++(60:-1)node[anchor =east]{$6k^2+3 $}--++(2,-.5);
\draw[->](-5.25,5.2)++(60:-1)node[anchor =east]{$6k^2+4 $}--++(2,0);
\draw[->](-5.25,5.2)++(60:-2)node[anchor =east]{$6k^2+5 $}--++(2,0);
\end{tikzpicture}
\caption{Construction of the $(6k^2+m)$-configuration for $k=3$ and $m=5$.}
\label{FigureCkn}
\end{figure}

Now for $n\neq 6k^2$  we can assume that $n=6k^2+m$ for some $1\leq m< 12k+6$. We start \EEE from  $D_{6k^2}$ \EEE and add a new atom \EEE to the bond graph in such a way that it gets bonded to the leftmost \EEE among the uppermost atoms \EEE of  $D_{6k^2}$ \EEE and that it has distance  larger \EEE or equal than $\sqrt{3}$ to all the other atoms. We choose the charge to be the opposite charge of the leftmost \EEE among the uppermost atoms. Then we add atoms in a counter-clockwise fashion such that the latest atom added is bonded to the atom \EEE added in the previous step and possibly to some other atom of  $D_{6k^2}$.\EEE  Moreover,  its distance to all the other atoms is at least $\sqrt{3}$. \EEE We choose the charge to be the opposite of the charge of the atom \EEE added in the previous step. One can realize that this defines uniquely \EEE a procedure in order to add $m$ atoms as shown in Fig.~\ref{FigureCkn}.

The first main result of this section is the following upper bound for the ground-state energy. \EEE

\begin{proposition}[Upper bound for ground-state energy\EEE]\label{PropositionDaisy} For  each $n \in \mathbb{N}$, let $D_n$ be the configuration  introduced above. There holds  
$
\mathcal{E}(D_n) = -\lfloor \beta(n)\rfloor.
$
In particular, for all ground states $C_n$ there holds 
$
\mathcal{E}(C_n) \leq  -\lfloor \beta(n)\rfloor.
$\EEE
\end{proposition}
\begin{proof} The proof follows as in \cite[Proposition 4.1, Proposition 5.1]{Mainini}, additionally \EEE noting that in the provided construction, all atoms \EEE in the bond graph are bonded \EEE to points of opposite charge only.
\end{proof}

\subsection{Daisies with an additional  trapezoid\EEE}\label{sec: daisy + something}

Note that the above configurations have net charge in $\lbrace -1 ,0 , 1\rbrace$. Starting with a daisy and attaching a  trapezoid \EEE in a suitable way, we can also construct configurations with energy $- \lfloor \beta(n) \rfloor$ having a net charge of order $n^{1/4}$. The following construction is inspired by the one of \cite{Davoli15} applied in connection to the derivation of the so-called $n^{3/4}$-law.

We choose $k \in \mathbb{N}$ with  $k \ge 252$, $r \in 6 \mathbb{N}$ with $r \le \sqrt{k/7}$, and let 
\begin{align}\label{eq: the-n-def}
n := 6k^2 + 2kr + \frac{1}{6}r^2 + 1 = 6\hat{k}^2+1,   
\end{align}
where $\hat{k} = k + r/6$. We construct a configuration $C_n$ as follows. We   start  from $C_{6k^2}$ and add a new  atom  to the bond graph in such a way that it gets bonded to the  leftmost  among the uppermost atoms  of $C_{6k^2}$ and that it has distance greater or equal than $\sqrt{3}$ to all the other atoms. We choose the charge to be the opposite charge of the  leftmost  among the uppermost atoms.  (Say, charge $1$ for definiteness.) \EEE Then we add atoms in a counter-clockwise way similar to the previous construction until `the line is completed'.
%
Overall, we add $2k-1$ atoms by this procedure where $k$ atoms have charge $1$ and $k-1$ atoms have charge $-1$. We then repeat this construction by adding another row of atoms on the top where we need to add $2(k-1)-1$ atoms. We repeat this until we have added $r$ rows of atoms corresponding to 
$$ \sum_{j= k - r+1}^k (2j-1) = 2kr - r^2$$
added atoms. Note that in each row the number of added atoms of charge $1$ exceeds the number of added atoms of charge $-1$ by exactly one.  The construction is sketched in Fig.~\ref{FigureDaisyParallelogram}. \EEE

 Finally, in a last row we add $m :=7r^2/6 +1$ atoms. Note that by the assumptions $k \ge 5$ and  $r \le \sqrt{k/7}$  we have
\begin{align*}
\frac{7}{6}r^2 + 1 \le k/6+ 1 \le  \frac{k - r}{2}, \EEE
\end{align*}
 i.e.,  the chain of atoms added in row $r$ is `long enough' so that $m$ atoms can be added in the last row. Note that the resulting configuration consists of $n$ atoms, see \eqref{eq: the-n-def}. 
 
 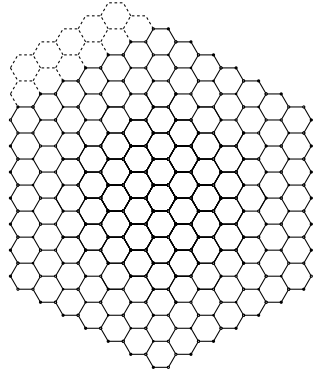
\begin{figure}[htp]
\centering
\begin{tikzpicture}[scale=0.2]

\foreach \k in {0,...,5}{
\begin{scope}[rotate=60*\k]
\draw[thin](8,0)++(60:2)++(-60:-1)++(0:1)++(60:2)++(120:2)--++(0:-1);
\foreach \j in {0,...,5} {
\draw[thin](8,0)++(60:-2+\j)++(-60:3-\j)--++(0:1)--++(60:1)--++(120:1);
\draw[fill=white](8,0)++(60:-2+\j)++(-60:3-\j)++(0:1)circle(.07)++(60:1)++(120:1)circle(.07);
\draw[fill=black](8,0)++(60:-2+\j)++(-60:3-\j)++(0:1)++(60:1)circle(.07)++(120:1);
}
\end{scope}
}

\begin{scope}[rotate=60*2]
\draw[thin, dash pattern=on 1pt off 1pt on 1pt off 1pt](9,0)++(60:-2+5)++(-60:4-5)++(0:1)++(60:1)++(120:1)--++(0:-1);
\foreach \j in {0,...,5} {
\draw[thin,dash pattern=on 1pt off 1pt on 1pt off 1pt](9,0)++(60:-2+\j)++(-60:4-\j)--++(0:1)--++(60:1)--++(120:1);
}

\draw[thin, dash pattern=on 1pt off 1pt on 1pt off 1pt](10,0)++(60:3)++(-60:5-5)++(0:1)++(60:1)++(120:1)--++(0:-1);
\foreach \j in {0,...,4} {
\draw[thin,dash pattern=on 1pt off 1pt on 1pt off 1pt](10,0)++(60:-1+\j)++(-60:4-\j)--++(0:1)--++(60:1)--++(120:1);
}
\end{scope}

\foreach \k in {0,...,5}{
\begin{scope}[rotate=60*\k]
\draw[thin](7,0)++(60:2)++(-60:-2)++(0:1)++(60:1)++(120:1)--++(0:-1);
\foreach \j in {0,...,4} {
\draw[thin](7,0)++(60:-2+\j)++(-60:2-\j)--++(0:1)--++(60:1)--++(120:1);
\draw[fill=white](7,0)++(60:-2+\j)++(-60:2-\j)++(0:1)circle(.07)++(60:1)++(120:1)circle(.07);
\draw[fill=black](7,0)++(60:-2+\j)++(-60:2-\j)++(0:1)++(60:1)circle(.05)++(120:1);
}
\end{scope}
}

\foreach \k in {0,...,5}{
\foreach \j in {-1,1}{
\draw[thin](60*\k:\j*7)--++(120+60*\k:\j*1);
\draw[thin](60*\k:\j*7)--++(240+60*\k:\j*1);
}
}

\foreach \k in {0,...,2}{
\draw[fill=black](120*\k:7)circle(.07);
\draw[fill=white](120*\k+60:7)circle(.07);
}

%
%
%

%

\foreach \j in {0,...,2}{
\draw[thin](\j*60:1)--++(120+\j*60:1);
\draw[thin](\j*60+180:1)--++(300+\j*60:1);
}
\foreach \j in {0,...,2}{
\draw[fill=black](\j*120:1)circle(.05);
\draw[fill=white](\j*120+60:1)circle(.07);
}
\foreach \m in {0,1,2,3,4,5}{
\foreach \k in {1,...,3}{
\foreach \i in {1,...,6}{
\begin{scope}[shift={(\m*60+90:{-sqrt(3)})}]
\begin{scope}[shift={(\i*60+30:{\k*sqrt(3)})}]
\foreach \j in {0,...,2}{
\draw[thin](\j*60:1)--++(120+\j*60:1);
\draw[thin](\j*60+180:1)--++(300+\j*60:1);
}
\foreach \j in {0,...,2}{
\draw[fill=black](\j*120:1)circle(.07);
\draw[fill=white](\j*120+60:1)circle(.07);
}
\end{scope}
\end{scope}
}
}
}

\end{tikzpicture}
\caption{Construction of a daisy \EEE with an additional  trapezoid. \EEE Two rows of atoms have been added where the first and the last atom in the added rows have charge $+1$. Thus, the net charge of the configuration is $+2$.\EEE}
\label{FigureDaisyParallelogram}
\end{figure}

 Concerning the energy, we observe that in the row where $2j-1$ atoms are added \EEE we add exactly $3j-2$ bonds to the bond graph for $k-r+1 \le j \le k$. \EEE In the ultimate row we add $1 + 3(m-1)/2  = 1+ 7r^2/4$ bonds. Consequently, using Proposition \ref{PropositionDaisy} the energy of $C_n$ is given by
\begin{align*}
\mathcal{E}(C_n) & = \mathcal{E}(C_{6k^2}) - \sum_{j=k-r+1}^k (3j-2)   - 1 - 7r^2/4 =   -9k^2 + 3k - 3kr  + 3r^2/2  + r/2 - 1 -  7r^2/4 \\& = -\frac{3}{2}6(k + r/6)^2 + \sqrt{\frac{3}{2}6(k+r/6)^2} -1 \EEE = -\lfloor \beta( 6\hat{k}^2) \rfloor  -  1 \EEE = -\lfloor \beta(n) \rfloor. 
\end{align*}  
We now determine the \emph{net charge} of the configuration. Recall that the configuration $C_{6k^2}$, from which we started our construction, has net charge zero. As explained above, in each added row the number of added atoms of charge $1$ exceeds the number of added atoms of charge $-1$ by exactly one, i.e., $\mathcal{Q}(C_n) = r+1$. 

We are now in the position to give the proof of Theorem \ref{TheoremCharge}(ii). To this end, consider the sequence of integers $n_j = 6(j+r_j/6)^2+1$, where $r_j = 6 \lfloor \frac{1}{6} \sqrt{j/7} \rfloor$, and the ground states $C_{n_j}$ constructed above. Note that by  $r_j \le j$ we  get $n_j \le 9j^2$. Thus, we calculate
$$\mathcal{Q}(C_{n_j}) = r_j+1 \ge  \sqrt{j/7}-5 \ge \frac{1}{\sqrt{21}}n_j^{1/4} - 5.$$
This yields  $\displaystyle\liminf_{j \to +\infty} n_j^{-1/4}|\mathcal{Q}(C_{n_j})|>0$.

\EEE

\section{Boundary Energy}

In this section  we introduce the concept of boundary energy and prove a corresponding estimate which will be fundamental for the characterization of ground states and their energy in Section \ref{sec: characterization-energy}, \ref{sec: characterization}. \EEE

 \textbf{Boundary atoms, boundary energy:} Within the bond graph, we say that an atom   is a \textit{boundary atom\EEE} if it is not contained in the interior region of any simple cycle. Otherwise, we call it \emph{bulk atom}. \EEE We denote the union of the boundary atoms \EEE by $\partial X_n$. A \textit{boundary bond} is a bond containing a boundary atom. All other bonds are called \textit{bulk bonds}. Similarly, a bond-angle will be called \emph{boundary angle} or \emph{bulk angle}, depending on whether the associated atom is a boundary atom or not. \EEE  

We denote by $d$ the number of boundary atoms of $X_n$. Given $C_n$, we define its \textit{bulk}, denoted by $C_n^{\mathrm{bulk}}$, as the sub-configuration obtained by dropping all boundary atoms (and the corresponding charges). Similarly, the particle positions are indicated by $X_n^{\mathrm{bulk}}$. \EEE With the above definition, we have that the bulk is an \EEE $(n-d)$-atom configuration. We define $\mathcal{E}^{\mathrm{bulk}}(C_n)$ as the energy of $C_n^{\mathrm{bulk}}$. We then have two contributions to the energy of $C_n$, namely $\mathcal{E}^{\mathrm{bnd}}$ and $\mathcal{E}^{\mathrm{bulk}}$, where
\begin{align*}
\mathcal{E}^{\mathrm{bnd}}(C_n) := \mathcal{E}(C_n) - \mathcal{E}^{\mathrm{bulk}}(C_n).
\end{align*}

  \textbf{Maximal polygon:} We introduce an additional notion in the case that $C_n$ is connected and does not contain acyclic bonds. In this case, the bond graph is  delimited by a simple cycle which we call the \emph{maximal polygon}.  We denote the atoms of the maximal polygon  by $\{x_1,\ldots,x_d\}$   and  the interior angle  at the point  $x_i \in \partial X_n$ by $\theta_i$. \EEE Moreover, we indicate by
\begin{align*}
&I_2=\{  x_i \in \partial X_n: \EEE \#\mathcal{N}(x_i)=2\},\\&I_3=\{  x_i \in \partial X_n: \EEE \#\mathcal{N}(x_i)=3\},
\end{align*}
the set of $2$-bonded and $3$-bonded boundary atoms, respectively.  If $C_n$ is a ground state, we note that $\# I_2 + \# I_3 = d$ by Lemma \ref{LemmaNeighborhood}. \EEE

 We now provide an estimate for the boundary energy. Its proof is inspired by \cite[Lemma 6.2]{Mainini}. The precise estimates, however, deviate significantly from the study in \cite{Mainini}   due to the presence of the repulsive potential $V_{\rm r}$ instead of an angular potential. We defer a discussion in that direction after the proof, see Remark \ref{rem: angle/bond}. \EEE

\begin{lemma}[Boundary energy\EEE]
\label{LemmaBoundaryEnergy}
Let $n \geq 6$ and let $C_n$ be a connected ground state with no acyclic bonds.   Then
\begin{align}\label{BoundaryEnergyEstimate}
\mathcal{E}^{\mathrm{bnd}}(C_n) \geq -\frac{3}{2}d +3
\end{align}
with equality only if  all the three \EEE  following conditions are satisfied:
\begin{align}
& \text{All boundary bonds are of unit length},\label{eq: 1}\\
&\#I_3 = \frac{1}{2}d -3,\label{eq: 2}\\
&\theta_i = \frac{4\pi}{3} \text{ if } x_i \in I_3 \text{ and } \theta_i = \frac{2\pi}{3} \text{ if } x_i \in I_2.\label{eq: 3}
\end{align}

\end{lemma}

\begin{remark}\label{rem: bdy}
{\normalfont (i) Note  by Lemma \ref{LemmaNeighborhood} that  $C_n$ has alternating charge distribution and thus  $d$ is even. 
(ii)  Observe by  \eqref{eq: 2} that equality in \eqref{BoundaryEnergyEstimate} implies that $\frac{3}{2}d-3$ bonds contribute to the boundary energy. Thus, equality in \eqref{BoundaryEnergyEstimate} together with [ii] and [vii] imply that for   all boundary atoms $x_i$ one has $\min\lbrace |x_i  - x_j|: \  j \in \lbrace 1, \ldots, n \rbrace, j\neq i, \ q_j = q_i \rbrace \ge \sqrt{3}$.\EEE}
\end{remark}

\begin{proof}  Suppose that  $\{x_1,\ldots,x_d\}$ are ordered such that \EEE  $x_i\in \mathcal{N}(x_{i+1})$, $i=1,\ldots,d$.  Here and in the following, we use \EEE the identification $x_{d+1} = x_1$ and $x_{0} = x_d$. For a $3$-bonded atom $x_i$, denote by  $x_i^b \in X_n \setminus \{x_{i-1},x_{i+1}\}$ the atom that is connected to $x_i$ with the third bond.  The boundary energy can be  estimated by
\begin{align}\label{eq: boundary energy}
\begin{split}
\mathcal{E}^{\mathrm{bnd}}(C_n) &\geq   \sum_{i=1}^d \Big(\frac{1}{2}\big( V_{\mathrm{a}}(|x_i-x_{i+1}|)+V_{\mathrm{a}}(|x_i-x_{i-1}|)\big)+V_{\mathrm{r}}(|x_{i+1}-x_{i-1}|) \Big) \\&\quad+ \sum_{x_i \in I_3}\Big(  V_{\mathrm{a}}(|x_i-x_i^b|) + V_{\mathrm{r}}(|x_{i-1}-x_i^b|) + V_{\mathrm{r}}(|x_{i+1}-x_i^b|)  \Big).
\end{split}
\end{align}
Here, we split the interactions into two sums. The first sum represents \EEE the energy between successive boundary atoms. The second sum is a lower bound for the energy of  $3$-bonded \EEE atoms that may be bonded to other boundary atoms or to bulk atoms:   note that we might double count (negative) attractive interaction if also $x_i^b$ is a boundary atom. However, we never double count  (positive) repulsive interaction.  To see this, suppose that a repulsive interaction would be double counted in the sum of the right hand side, i.e., there exist $i,j  \in \{1,\ldots,d\}$, $i \neq j$, such that,  e.g.,   $\{x_{j-1},x_j^b\} = \{x_{i-1},x_i^b\}$.  More precisely, as $i\neq j$, this means $x_{j-1} = x^b_i$ and $x_{i-1} =x_j^b$.  This implies that $x_i$ as well as $x_j$  are  both  bonded to $x_{j-1}$ and   $x_{i-1}$.  Then, however, the atoms $\{x_j,x_{j-1},  x_{i-1},  x_i\}$ form a square in the bond graph which contradicts Lemma \ref{LemmaPolygon}.  \EEE

For a $3$-bonded atom $x_i$, denote by $\theta_i^1,\theta_i^2\in [0,2\pi]$ the two angles forming $\theta_i$ enclosed by the  three bonds at $x_i$. Finally, we define $\delta : = \frac{\#I_2+2\#I_3}{d}$ and note that $\delta \in [1,2]$ since $\# I_2 + \# I_3  =d$. \EEE

We will prove that 
\begin{align}\label{boundarylowerbound-new}
\mathcal{E}^{\mathrm{bnd}}(C_n) & \ge -\delta d + \sum_{x_i \in I_2} V_{\mathrm{r}}\Big(2\sin\Big(\frac{\theta_i}{2}\Big)\Big) + \sum_{x_i \in I_3}   \sum_{j=1,2} V_{\mathrm{r}}\Big(2\sin\Big(\frac{\theta_i^j}{2}\Big)\Big)  \EEE \nonumber \\& \geq  - \delta d +\delta d V_{\mathrm{r}}\Big(2\sin\Big(\frac{\pi(d-2)}{2\delta d}\Big)\Big),
\end{align}
where the first inequality is strict if not all  lengths of boundary bonds \EEE  are equal to $1$. We defer the proof of this estimate  and its strict version \EEE and first show that it implies the statement of the lemma. \EEE

First, introducing \EEE
\begin{align*}
\alpha(\delta) := \frac{\pi(d-2)}{2\delta d},
\end{align*}
estimate \eqref{boundarylowerbound-new}  can be written as \EEE
\begin{align}\label{Ebndestimate}
\begin{split}
\mathcal{E}^{\mathrm{bnd}}(C_n) &\geq -\delta d + \delta dV_{\mathrm{r}}\Big(2\sin\Big(\frac{\pi(d-2)}{2\delta d}\Big)\Big) =  \delta d\Big(
V_{\mathrm{r}}\Big(2\sin\Big(\alpha(\delta)\Big)\Big)-1\Big).
\end{split}
\end{align}
We obtain (\ref{BoundaryEnergyEstimate}) by minimizing the right hand side of (\ref{Ebndestimate}) with respect to $\delta$. To see this, \EEE set $\delta_0 = \frac{3}{2}-\frac{3}{d}$. For $\delta \leq \delta_0$ we have 
\begin{align*}
\alpha(\delta) \geq \alpha(\delta_0) = \frac{\pi}{3}.
\end{align*}
By $[\mathrm{vii}]$ \EEE we get $V_{\mathrm{r}}( 2\sin(\alpha(\delta)))=0$ for all $1\leq \delta\leq \delta_0$\EEE. Therefore, we find
\begin{align}\label{deltageqdelta1}
\delta d\Big( \EEE V_{\mathrm{r}}\Big(2\sin\Big(\alpha(\delta)\Big)\Big)-1\Big)= -\delta d\geq -\delta_0 d=-\Big(\frac{3}{2}d-3\Big)
\end{align}
and we obtain estimate (\ref{BoundaryEnergyEstimate})  \EEE for $\delta\leq \delta_0$. Now for $\delta >\delta_0$, we have $\alpha(\delta) < \alpha(\delta_0)$.  By  $(\mathrm{v})$ \EEE we get
$$V_{\mathrm{r}}(2\sin(\alpha(\delta))) \geq V_{\mathrm{r}}(2\sin(\alpha(\delta_0))) + 2V_{\mathrm{r},-}^\prime(2\sin(\alpha(\delta_0))(\sin(\alpha(\delta))-\sin(\alpha(\delta_0)).$$
 Then by \EEE the fact that $\sin(\theta)$ is concave for $\theta \in [0,\pi]$, $V_{\mathrm{r},-}^\prime(\sqrt{3}) < - 3/\pi<0$ by $[\mathrm{viii}]$,  $\alpha(\delta_0) = \frac{\pi}{3}$, \EEE and $\alpha(\delta)-\alpha(\delta_0) < 0$ \EEE we derive
\begin{align}\label{deltageqdelta0}
\begin{split}
V_{\mathrm{r}}(2\sin(\alpha(\delta))) &\geq V_{\mathrm{r}}(\sqrt{3}) + 2V_{\mathrm{r},-}^\prime (\sqrt{3})\cos(\alpha(\delta_0))\big(\alpha(\delta)-\alpha(\delta_0)\big) \\&=V_{\mathrm{r}}(\sqrt{3}) + V_{\mathrm{r},-}^\prime(\sqrt{3})\Big(\frac{\pi(d-2)}{2\delta d}-\frac{\pi}{3}\Big) > -\frac{3}{\pi}\Big(\frac{\pi(d-2)}{2\delta d}-\frac{\pi}{3}\Big) \\&=\frac{1}{\delta d}\Big(\delta d-\frac{3}{2}d+3\Big).
\end{split}
\end{align}
 Here, we also used that $\cos(\alpha(\delta_0)) = \frac{1}{2}$ and   $V_{\mathrm{r}}(\sqrt{3}) = 0$. \EEE From the previous calculation and (\ref{Ebndestimate}),  estimate (\ref{BoundaryEnergyEstimate})  \EEE follows  also for $\delta >\delta_0$. \EEE 

 We now show that we have strict inequality in (\ref{BoundaryEnergyEstimate})  if one of the conditions \eqref{eq: 1}-\eqref{eq: 3} is violated.  First,  we have a strict inequality in \eqref{boundarylowerbound-new} if a boundary bond is not of unit length and therefore also in \eqref{Ebndestimate}.  (Recall that we defer the proof of \eqref{boundarylowerbound-new} and its strict version to the end of the proof.) \EEE If \eqref{eq: 2} is violated, we find $\delta \neq \delta_0$ after a short computation. Then we obtain strict inequalities from \eqref{deltageqdelta1} and \eqref{deltageqdelta0}, respectively. Finally, let use suppose that \eqref{eq: 3} is violated. We can assume that $\delta=\delta_0$ and \eqref{eq: 1}-\eqref{eq: 2} hold as otherwise the inequality in  (\ref{BoundaryEnergyEstimate}) is strict. If  \EEE equality holds in (\ref{BoundaryEnergyEstimate}), then equality also holds in \eqref{boundarylowerbound-new}.  As $V_{\rm r}(2\sin(\alpha(\delta_0)))= 0$, \EEE this implies  
\begin{align*}
\sum_{ x_i \in I_2} V_{\mathrm{r}}\Big(2\sin\Big(\frac{\theta_i}{2}\Big)\Big) + \sum_{x_i \in I_3} \Big( V_{\mathrm{r}}\Big(2\sin\Big(\frac{\theta_i^1}{2}\Big)\Big) + V_{\mathrm{r}}\Big(2\sin\Big(\frac{\theta_i^2}{2}\Big)\Big) \Big)=0.
\end{align*} 
In view of  ${\rm [iv]}$ and  ${\rm [vii]}$\EEE, this gives 
\begin{align}\label{eq: anglerange}
\theta_i \in [\tfrac{2\pi}{3},\tfrac{4\pi}{3}] \ \ \ \text{for all $x_i \in I_2$}, \ \ \ \ \ \theta^1_i, \theta^2_i \in [\tfrac{2\pi}{3},\tfrac{4\pi}{3}] \ \ \ \text{for all $x_i \in I_3$}.
\end{align}
  Under the assumption that \eqref{eq: 3} is violated, using  \eqref{eq: anglerange} and  $\theta_i = \theta_i^1 + \theta_i^2$ we find some  $x_i \in I_3$ with  $\theta_i > \frac{4\pi}{3}$ or some $x_i \in I_2$ with $\theta_i > \frac{2\pi}{3}$. Then \eqref{eq: anglerange} implies \EEE
\begin{align*}
  \pi(d-2) = \EEE \frac{4\pi}{3}\Big(\frac{1}{2}d-3\Big)+ \frac{2\pi}{3}\Big(\frac{1}{2}d +3\Big) =\frac{4\pi}{3}\#I_3 + \frac{2\pi}{3}\#I_2 < \sum_{x_i \in I_2} \theta_i + \sum_{x_i \in I_3} \theta_i = \pi(d-2),
\end{align*}
 where the last step follows from the fact that the maximal polygon has $d$ vertices. \EEE This is a contradiction and shows strict inequality in \eqref{BoundaryEnergyEstimate} if \eqref{eq: 3} is violated.

To complete the proof, it remains to show \eqref{boundarylowerbound-new}  and its strict version. \EEE In the case of  a \EEE $2$-bonded $x_i$,  define $r^1_i = |x_i-x_{i-1}|$, $r^2_i = |x_i-x_{i+1}|$. In the case of  a \EEE $3$-bonded $x_i$,  define  additionally \EEE $r^3_i = |x_i-x_{i}^b|$.  
By the cosine rule we obtain  
  \begin{align}\label{xidistance}
\begin{split}
&|x_{i+1}-x_{i-1}|=  \ell(\theta_i,r_i^1,r_i^2), \ \ \  |x_i^b-x_{i-1}| = \ell(\theta_i^1,r_i^1,r_i^3), \ \ \ |x_i^b-x_{i+1}| =  \ell(\theta_i^2,r_i^2,r_i^3),
\end{split}
\end{align}
where we have used the shorthand
\begin{align}\label{eq: ell} 
 \ell(\theta,r_1,r_2) = \sqrt{r_1^2 + r_2^2 - 2r_1r_2 \cos(\theta)}.
 \end{align}
\EEE We want to prove that for every boundary atom $x_i$ its contribution to the energy can be controlled by the energy contribution in a modified configuration which has \EEE  the same angles but   unit bond lengths instead of $r_i^1, r_i^2$. Recall that \EEE  by $[\mathrm{viii}]$ we have for all $r \in (1,r_0]$ \EEE
\begin{align}\label{eq: strict inequality}
\frac{1}{2(r-1)}(V_{\mathrm{a}}(1)-V_{\mathrm{a}}(r))  <  V'_{\mathrm{r},+}(1). \EEE
\end{align}
 Let $\theta \in [\pi/3,5\pi/3]$, $r_1,r_2 \in [1,r_0]$. \EEE Then $\ell(\theta,r_1,r_2)\geq 1$, $V_{\mathrm{r},+}'(r)\leq 0$ for $r \geq 1$ and $V_{\mathrm{r},+}'$ is monotone increasing \EEE due to the convexity assumption in  $(\mathrm{v})$ \EEE on $V_{\mathrm{r}}$.  Moreover, we have  $\partial_{r_1} \ell(\theta,r_1,r_2)\leq 1$ by an elementary computation. (This can also be seen by a geometric argumentation: by changing the length of $r_1$, the length change of $\ell$ is always smaller or equal  to \EEE the length change of $r_1$,  equal only if $\theta\in\lbrace 0 , \pi \rbrace$.)  We therefore obtain  for $r, s \in [1,r_0]$ \EEE
\begin{align*}
\frac{1}{2(r-1)}(V_{\mathrm{a}}(1)-V_{\mathrm{a}}(r))  \le  V_{\mathrm{r},+}'(1) \EEE \leq V_{\mathrm{r},+}'(\ell(\theta,s,r_2)) \leq V_{\mathrm{r},+}'(\ell(\theta,s,r_2)) \partial_{ r_1  }\ell(\theta,s,r_2).
\end{align*}
Integrating this from $1$ to $r$ in the variable $s$ and using the fundamental theorem of calculus, we get \EEE
\begin{align*}
\frac{1}{2}\left( V_{\mathrm{a}}(1)-V_{\mathrm{a}}(r) \right)\leq \int_{1}^r V_{\mathrm{r},+}'(\ell(\theta,s,r_2)) \partial_{r_1} \EEE \ell(\theta,s,r_2) \mathrm{d}s = V_{\mathrm{r}}(\ell(\theta,r,r_2))-V_{\mathrm{r}}(\ell(\theta,1,r_2)).
\end{align*}
Applying this estimate in the second \EEE as well as in the third \EEE component of $\ell(\theta,r_1,r_2)$ with  $1$ \EEE and $r_2$ respectively, we derive  \EEE
\begin{align}\label{Estimateunitlength}
V_{\mathrm{a}}(1) + V_{\mathrm{r}}(\ell(\theta,1,1)) \leq \frac{1}{2}\left( V_{\mathrm{a}}(r_1)+V_{\mathrm{a}}(r_2)\right)  +V_{\mathrm{r}}(\ell(\theta,r_1,r_2)).
\end{align}
Note that, due to ${\rm [i]}$, \eqref{eq: neighborhood}, and Lemma \ref{lemma:bondangles}, for $2$-bonded $x_i $, we have $r_i^1,r_i^2 \in [1,r_0] $, $\theta_i \in [\pi/3,5\pi/3]$, and for $3$-bonded $x_i$ we have  $r_i^1,r_i^2, r_i^3 \in [1,r_0]$, and $\theta_i^1,\theta_i^2 \in [\pi/3,5\pi/3]$. Now for all  $2$-bonded $x_i$,  using (\ref{Estimateunitlength})  with $r_i^1$, $r_i^2$, and $\theta_i$, \EEE   we have
\begin{align}\label{Twobondedestimate}
\frac{1}{2}\left( V_{\mathrm{a}}(r_i^1) + V_{\mathrm{a}}(r_i^2)\right) + V_{\mathrm{r}}(\ell(\theta_i,r_i^1,r_i^2)) \geq  V_{\mathrm{a}}(1) + V_{\mathrm{r}}(\ell( \theta_i, \EEE 1,1)).
\end{align}
On the other hand, for  all \EEE $3$-bonded $x_i$,  using (\ref{Estimateunitlength}) twice  with $r_i^1$, $r_i^3 $, and $\theta_i^1$ and $r_i^2$, $r_i^3 $, and $\theta_i^2$\EEE, we have
\begin{align}\label{Threebondedestimate}
\begin{split}
&\frac{1}{2}\left( V_{\mathrm{a}}(r_i^1) +V_{\mathrm{a}}(r_i^3)\right) + V_{\mathrm{r}}(\ell(\theta_i^1,r_i^1,r_i^3)) \geq V_{\mathrm{a}}(1) + V_{\mathrm{r}}(\ell(\theta_i^1,1,1)), \\&
\frac{1}{2} \left(V_{\mathrm{a}}(r_i^2) +  V_{\mathrm{a}}(r_i^3)\right) +  V_{\mathrm{r}}(\ell(\theta_i^2,r_i^2,r_i^3)) \geq V_{\mathrm{a}}(1)  + V_{\mathrm{r}}(\ell(\theta_i^2,1,1)).
\end{split}
\end{align}
Using (\ref{xidistance}), (\ref{Twobondedestimate})-(\ref{Threebondedestimate}), ${\rm [ii]}$, and $V_{\rm r} \ge 0$   we obtain  by \eqref{eq: boundary energy}
 \begin{align}\label{boundaryestimate1}
\mathcal{E}^{\mathrm{bnd}}(C_n) &\geq \sum_{i=1}^d \Big(\frac{1}{2}\big( V_{\mathrm{a}}(|x_i-x_{i+1}|)+V_{\mathrm{a}}(|x_i-x_{i-1}|)\big)+V_{\mathrm{r}}(|x_{i+1}-x_{i-1}|) \Big)\notag \\&\quad+ \sum_{x_i \in I_3}\Big(  V_{\mathrm{a}}(|x_i-x_i^b|) + V_{\mathrm{r}}(|x_{i-1}-x_i^b|) + V_{\mathrm{r}}(|x_{i+1}-x_i^b|)  \Big)\notag \\&  = \EEE  \sum_{x_i \in I_2}  \Big(\frac{1}{2}\big(V_{\mathrm{a}}(r_i^1) + V_{\mathrm{a}}(r_i^2)\big) + V_{\mathrm{r}}(\ell(\theta_i,r_i^1,r_i^2)) \Big)  + \sum_{x_i \in I_3} V_{\rm r}(|x_{i+1}-x_{i-1}|) \EEE \notag \\&\quad+\sum_{x_i \in I_3} \Big(\frac{1}{2}\big(V_{\mathrm{a}}(r_i^1)+V_{\mathrm{a}}(r_i^2) + 2 V_{\mathrm{a}}(r_i^3)\big) + V_{\mathrm{r}}(\ell(\theta_i^1,r_i^1,r_i^3))+  V_{\mathrm{r}}(\ell(\theta_i^2,r_i^2,r_i^3)) \Big) \notag\\&\geq -(\#I_2+2\#I_3) + \sum_{i \in I_2} V_{\mathrm{r}}(\ell(\theta_i,1,1))  + \sum_{i \in I_3}  \sum_{j=1,2}  V_{\mathrm{r}}(\ell(\theta_i^j,1,1)).
\end{align}
 For later purposes, we remark that this inequality is strict if one bond has not unit length. This follows from the strict inequality in \eqref{eq: strict inequality}. \EEE

 Recall   $\delta= \frac{\#I_2+2\#I_3}{d}$ and note that $\ell(\theta,1,1) = 2\sin(\theta/2)$ by  \eqref{eq: ell}. Using \EEE  $\theta_i = \theta_i^1 + \theta_i^2$ for $x_i \in I_3$ and \eqref{boundaryestimate1} \EEE we obtain
\begin{align*} 
\mathcal{E}^{\mathrm{bnd}}(C_n) &\geq -\delta d + \sum_{x_i \in I_2} V_{\mathrm{r}}(\ell(\theta_i,1,1)) + \sum_{x_i \in I_3}  \sum_{j=1,2}V_{\mathrm{r}}(\ell(\theta_i^j,1,1))  \EEE \notag \\& = -\delta d + \sum_{x_i \in I_2} V_{\mathrm{r}}\Big(2\sin\Big(\frac{\theta_i}{2}\Big)\Big) + \sum_{x_i \in I_3}  \sum_{j=1,2} V_{\mathrm{r}}\Big(2\sin\Big(\frac{\theta_i^j}{2}\Big)\Big). \EEE  
\end{align*}
This yields the first inequality in \eqref{boundarylowerbound-new}. We note that this inequality is strict if one bond has not unit length since then \eqref{boundaryestimate1} is strict. 

 The second inequality in \eqref{boundarylowerbound-new} follows by a convexity argument: \EEE 
  \EEE since $V_{\mathrm{r}}$ is convex and non-increasing  by ${\rm [v]}$ \EEE and $\sin(\theta/2)$ is concave for $\theta \in [0,2\pi]$, we have  for all $\lambda \in [0,1]$ \EEE
\begin{align*}
\lambda V_{\mathrm{r}}\Big(2\sin\Big(\frac{\theta_1}{2}\Big)\Big)+ (1-\lambda)V_{\mathrm{r}}\Big(2\sin\Big(\frac{\theta_2}{2}\Big)\Big) &\geq V_{\mathrm{r}}\Big(\lambda 2\sin\Big(\frac{\theta_1}{2}\Big) + (1-\lambda) 2\sin\Big(\frac{\theta_2}{2}\Big)\Big) \\& \geq V_{\mathrm{r}}\Big( 2\sin\Big(\frac{\lambda\theta_1+(1-\lambda)\theta_2}{2}\Big)\Big).
\end{align*}
Hence, $\theta \mapsto V_{\mathrm{r}}\Big(2\sin\Big(\frac{\theta}{2}\Big)\Big)$ is a convex function for $\theta \in [0,2\pi]$. This together  with \EEE the fact that  $\# I_2 + 2 \# I_3 = \delta d$   and 
$$
\pi(d-2)=\sum_{i=1}^d \theta_i = \sum_{x_i \in I_2} \theta_i + \sum_{x_i \in I_3}(\theta_i^1+\theta_i^2)
$$
 yields the second inequality in  \eqref{boundarylowerbound-new}. This \EEE concludes the proof. \EEE
\end{proof}

\begin{remark}\label{rem: angle/bond}

 (i)  We briefly  explain assumption [viii] from a technical point of view.  The condition prevents two phenomena concerning surface relaxation:   the  first is the occurrence of more atoms on the boundary of the configuration than  one would  expect   for hexagonal configurations. This is  achieved \EEE by the first condition of  [viii], cf.\  \eqref{deltageqdelta0}.   The second phenomenon is the presence of elastically deformed boundary bonds. This is prevented by the second condition,  cf.\ \eqref{Estimateunitlength}-\eqref{Threebondedestimate}.

 (ii) \EEE At this stage, let us highlight the difference of our analysis to \cite{Mainini}. In \cite{Mainini}, an empirical angular potential is considered which penalizes deviations of the bond-angles from $\frac{2\pi}{3}$,  modeling covalent bonding for carbon nanostructures. In our model for ionic compounds, the energy contribution can also be expressed in terms of the bond-angle.  More precisely, by \eqref{eq: ell} we have  energy contributions of the form 
$$V_{\rm r} \big( \ell(\theta,r_1,r_2) \big) = V_{\rm r}\Big(\sqrt{r_1^2 + r_2^2 - 2r_1r_2 \cos(\theta)}\Big).$$
In view of assumption [vii], only lengths $\ell(\theta,r_1,r_2) < \sqrt{3}$ are penalized. In particular, as $r_1,r_2 \ge 1$, bond-angles $\theta \ge \frac{2\pi}{3}$ never penalize the energy and, if the bond lengths $r_1$ and $r_2$ exceed one, also bond-angles  less  than $\frac{2\pi}{3}$ might not penalize the energy. 

In principle, this implies more geometric flexibility of ground-state configurations with respect to \cite{Mainini}. This  calls for refined arguments for controlling the boundary energy and, in particular, for characterizing the ground-state geometries in Section \ref{sec: characterization}. Let us highlight that, in spite of the weaker penalization of bond-angles deviating from $\frac{2\pi}{3}$, it is still possible to prove that ground states assemble themselves  in the hexagonal lattice. 
\end{remark}

 We recall the definition of $b$ in \eqref{eq:b}. Recall also \EEE the excess of edges $\eta = \sum_{j\geq 6} (j-6)  f_j $ introduced in \eqref{Excess}, where $f_j$ denotes the number of  elementary polygons \EEE with $j$ vertices in the bond graph.  \EEE By Lemma \ref{LemmaPolygon} for any ground state we have that $\eta \in 2\mathbb{N}$. Moreover, \EEE   $\eta=0$ if only if the bond graph consists of hexagons only. We now use this notion to estimate the cardinality of the  bulk. \EEE  It will turn out useful in Section \ref{sec: characterization} to exclude the existence \EEE of other  elementary polygons \EEE than hexagons. \EEE

\begin{lemma}[Cardinality of the bulk\EEE]\label{LemmaBoundaryestimate} Suppose that $C_n$ is a  connected ground state \EEE and that it does not contain any acyclic bonds. Then
\begin{align*}
n-d   = 4b \EEE + 6+\eta - 5n.
\end{align*}
\end{lemma}
\begin{proof}  By $f_j$ we denote the number of  elementary polygons \EEE in the bond graph with $j$ vertices and set $f = \sum_{j \ge 3} f_j$.  Note that $f$ is the number of faces of the bond graph (omitting the exterior face). \EEE We have
\begin{align*}
 \sum_{j \ge 3} jf_j \EEE  = 2  b\EEE -d,
\end{align*} 
since by the summation on the left all bonds  contained in the maximal polygon are counted only once whereas all other bonds \EEE are counted twice. From Lemma \ref{LemmaPolygon}   and the definition of $\eta$ \EEE we obtain
\begin{align*}
6  f  = \EEE  2  b \EEE -d-\eta.
\end{align*}
Using this together with Euler's formula $n-  b + f \EEE  =1$ (omitting the exterior face) we get
\begin{align*}
n-d  = 4b \EEE +6+\eta - 5n.
\end{align*}
\end{proof}

\section{Characterization of the  ground-state  energy}\label{sec: characterization-energy}

This section is devoted to the proof of Theorem \ref{TheoremEnergyGroundstates}. We only need to provide   a lower bound on the ground-state energy  since the  upper bound has already been obtained by an explicit construction, see Proposition \ref{PropositionDaisy}.

\EEE

We state two algebraic lemmas that will be used in the sequel.

\begin{lemma}\label{LemmaSquareroot} 
 Let $j,n,m \in \mathbb{N}$ and let \EEE $x \in \mathbb{R} $ satisfy
\begin{align*}
 \frac{m}{4} - \frac{5}{4}n \EEE  \ge x \geq -\frac{3}{2}n + j +\sqrt{\frac{3}{2}(-4x-5n+m)}.
\end{align*}
Then $x \geq -\frac{3}{2}n +j -3 + \sqrt{\frac{3}{2}(-4j+m+n+6)}  $.
\end{lemma}
\begin{proof} The proof is elementary: we note \EEE that the function
\begin{align*}
x \mapsto  x +\frac{3}{2}n - j -\sqrt{\frac{3}{2}(-4x-5n+m)}
\end{align*}
is  strictly increasing \EEE and vanishes for $x =-\frac{3}{2}n +j -3 + \sqrt{\frac{3}{2}(-4j+m+n+6)}$.
\end{proof}

We use the following properties of the function $\beta$ which has been defined in \eqref{eq: beta definition}. \EEE

\begin{lemma}\label{LemmaPropertiesbeta} The function $\beta : \mathbb{N} \to \mathbb{R}$ satisfies
\begin{itemize}
\item[1)] $\lfloor \beta(n-1)\rfloor +1 \leq \lfloor \beta(n)\rfloor$.
\item[2)] $\lfloor \beta(m)\rfloor + \lfloor \beta(n-m)\rfloor +1 \leq \lfloor \beta(n)\rfloor$ for all $n\geq 12$, $n \geq m \geq 6$ and equality holds if and only if $n=12$ and $m=6$.
\item[3)]  $\lfloor \beta(n)\rfloor \geq \lfloor \beta(n-k)\rfloor +2 + k$ for all  $n \ge 13$ and $ n \ge  k  \ge 6$. \EEE
\item[4)]  $\lfloor \beta(n)\rfloor \geq \lfloor \beta(n-5)\rfloor + 7$ for all $n \geq 13$ except for $n \in \lbrace 15,18,21,29\rbrace$. \EEE
\end{itemize}
\end{lemma}
\begin{proof}
The proof of 1) and 2) is elementary \EEE and can be found in \cite[Lemma 6.4, 6.5]{Mainini}.  It relies on monotonicity and convexity properties of $\beta$.  As a preparation for  3) and 4), we observe that for $n \geq 41$, $k = 5$ or for $n \ge 17$, $k=6$ or for $n \ge 13$, $n \ge k \ge 7$ one has 
\begin{align}\label{eq: help}
\frac{3}{2}n -\sqrt{\frac{3}{2}n}\geq \frac{3}{2}(n-k) -\sqrt{\frac{3}{2}(n-k)} + 2+ k.
\end{align}
Indeed, after some manipulations, we see that this  is equivalent to $\sqrt{\frac{3}{2}n} +\sqrt{\frac{3}{2}(n- k ) \EEE } \geq  3k/(k-4)\EEE$.  The latter holds true  for  $n \geq 41$,  $k = 5$ or  for $n \ge 17$, $k=6$ or for $n \ge 13$, $n \ge k \ge 7$. 

One can check directly that $\lfloor \beta(n)\rfloor = \lfloor \beta(n-6)\rfloor +8$ for $n =13,\ldots,16$,  see Table \ref{table2}. \EEE This together with  \eqref{eq: help} yields 3). Property 4) follows from \eqref{eq: help} and an explicit computation for the cases  $13 \le n \le 40$, see Table \ref{table}. \EEE
\end{proof}

   \begin{table}[h]\centering
\begin{tabular}{|c||c|c|c|c|c|c|c|c|c|c|c|c|c|c|c|c|c|c|c|c|c|c|c|c|c|c|c|}
\hline
$n$  & 2 & 3 & 4 & 5 & 6 & 7 & 8 & 9 & 10 & 11 & 12 & 13 & 14 &15     \\  \hline 
$\lfloor\beta(n)\rfloor$ & 1 & 2 & 3&4&6&7&8&9&11&12&13&15&16&17\\
\hline 
\hline
$n$ & 16 &17 &18 &19 &20 &21 & 22 & 23 & 24 &25 & 26 &27 &28 &29 \\
\hline
$\lfloor\beta(n)\rfloor$ &19 &20&21&23&24&25&27&28&30&31&32&34&35&36 \\
\hline 
  \end{tabular}
\vspace{8mm}
\caption{The function \EEE $\lfloor\beta(n)\rfloor$ for $2 \leq n \leq 29$. The table  together with Theorem \ref{TheoremEnergyGroundstates}  \EEE can be used to see that the configurations in Fig.~\ref{FigureFlexible}, Fig.~\ref{FigureBridge}, and Fig.~\ref{FigureOctagons} are ground states.\EEE}  \label{table2}
\end{table}

\begin{table}[h]\centering
\begin{tabular}{|c||c|c|c|c|c|c|c|c|c|c|c|c|c|c|c|c|c|c|c|c|c|c|c|c|c|c|}
\hline
$n$    & 13 & 14 &15 & 16 &17 &18 &19 &20 &21 & 22 & 23 & 24 &25 & 26  \\  \hline 
$\gamma(n)$   & 7&7&6&7&7&6&7&7&6&7&7&7&7&7\\
\hline 
\end{tabular}
\vspace{0.8cm}
\hspace{0.02cm}
\vspace{0.3cm}
\begin{tabular}{|c||c|c|c|c|c|c|c|c|c|c|c|c|c|c|c|c|c|c|c|c|c|c|c|c|c|c|}
\hline
$n$  &27 &28 &29 &30 &31 & 32 & 33 & 34 &35 & 36 &37 &38 &39 &40  \\
\hline
$\gamma(n)$  &7&7&6&7&7&7&7&7&7&7&7&7&7&7  \\
\hline 
  \end{tabular}
\vspace{-3mm}
\caption{The function $\gamma(n) := \lfloor \beta(n) \rfloor - \lfloor \beta(n-5)\rfloor$ for $13 \le n \le 40$.\EEE}  \label{table}
\end{table}

Before we proceed with the   proof of Theorem \ref{TheoremEnergyGroundstates}, we will consider the cases \EEE $1 \le n \le 6$ which will serve as the induction base. \EEE

 \begin{lemma}[Cases  $1 \le n \le 6$]\label{lemma: small n}
 For $1 \le n \le 6$ every ground state $C_n$ is connected and satisfies $\mathcal{E}(C_n) = -b = \lfloor \beta(n) \rfloor$. 
 \end{lemma}

 \begin{proof}
 Let $n \leq 5$. By Lemma \ref{LemmaPolygon} we have that the bond graph does not contain any  polygon \EEE with less  than or equal to $5$ edges. \EEE Hence, the bond graph is cycle free and $b \leq n-1$.  This provides  the lower bound $
\mathcal{E}(C_n) \geq -(n-1)$ by   Remark \ref{rem: repulsionsfree}. \EEE Clearly, one can construct configurations with $n$ atoms and $n-1$ bonds having alternating charge distribution.  Note also that $b = n-1 =  \lfloor \beta(n) \rfloor\EEE$ is only possible if the configuration is connected.  

 In the case $n=6$, note that the number of polygons $f$ satisfies $f \le 1$ since  by Lemma \ref{LemmaPolygon} we have that every polygon  has at least $6$ edges. \EEE By Euler's formula we get $6-b+f \geq 1$ where the inequality is due to the fact that we may have more than one connected component. This implies  $b \leq 6$ and thus $\mathcal{E}(C_6) \geq -b$ by   Remark \ref{rem: repulsionsfree}. Exactly a regular hexagon with alternating charge distribution and unit bond length yields a configuration with energy equal to $-6 = -\lfloor \beta(6) \rfloor$. This concludes the proof. \EEE
  \end{proof}

 We now proceed with the proof of Theorem \ref{TheoremEnergyGroundstates}. We follow the strategy devised in \cite[Theorem 6.1]{Mainini} with the adaptions needed due to the  presence of different atomic types and repulsive potentials \EEE instead of angle potentials. In contrast to  \cite{Mainini}, however, we split the proof of the characterization of the ground-state energy and the characterization of the ground states. Indeed, the latter is more involved in our setting and the investigation is deferred to Section \ref{sec: characterization}.

\begin{proof}[Proof of  Theorem \ref{TheoremEnergyGroundstates}] We start by noting that every ground state $C_n$ has alternating charge distribution by  Lemma \ref{LemmaNeighborhood}.  \EEE  By Proposition \ref{PropositionDaisy} we have that the ground-state energy satisfies
\begin{align}\label{Groundstateineq1}
\mathcal{E}(C_n) \leq -\lfloor \beta(n)\rfloor.
\end{align}
 We proceed by induction. Suppose that the statement has been proven for all $m < n$ (for $1 \le m \le 6$ see Lemma \ref{lemma: small n}). We first show connectedness of the ground state (Claim 1) and then the energy equality (Claim 2).  \EEE  
 
\textbf{Claim 1:}  $C_n$ is connected.

 \noindent \emph{Proof of Claim 1}: \EEE Assume by contradiction that $C_n$ was \EEE not connected, i.e., \EEE  $C_n$ consists of  two or more connected components.  Let $C'_m$ and $C'_{n-m}$ be two sub-configurations consisting of $m$ and $n-m$ atoms, respectively, which  do not have any bonds between them.  \EEE  If $m\geq 6, n\geq 12$, we can apply  the induction hypothesis,   Lemma \ref{LemmaPropertiesbeta} $2)$, and $V_{\rm r} \ge 0$ \EEE to get
\begin{align*}
\mathcal{E}(C_n) \ge \mathcal{E}(C_m') + \mathcal{E}(C_{n-m}') \EEE \geq -\lfloor\beta(n-m)\rfloor -\lfloor\beta(m)\rfloor >-\lfloor \beta(n)\rfloor.
\end{align*}
 If \EEE $m< 6$,  we can apply Lemma \ref{lemma: small n} \EEE and  Lemma \ref{LemmaPropertiesbeta} $1)$ iteratively $m$ times \EEE  to get
\begin{align*}
\mathcal{E}(C_n) \geq -\lfloor\beta(n-m)\rfloor -  \lfloor\beta(m)\rfloor  = -\lfloor\beta(n-m)\rfloor  - (m-1)\EEE  > -\lfloor \beta(n)\rfloor.
\end{align*}
The case  $n<12$ is already included in this argument since  \EEE at least one connected component consists of less than $6$ \EEE atoms.  In \EEE view of  (\ref{Groundstateineq1}), we obtain a contradiction to the fact that $C_n$ is a ground state.

\textbf{Claim 2:} Energy equality $\mathcal{E}(C_n) = -b = -\lfloor \beta(n) \rfloor$.

 \noindent \emph{Proof of Claim 2}: \EEE We divide the proof into three steps. We first treat the case that $C_n$ contains acyclic bonds (Step 1). Afterwards, we consider only configurations $C_n$ without acyclic bonds and show $\mathcal{E}(C_n) = -b$ (Step 2) and $\mathcal{E}(C_n) = -\lfloor \beta(n) \rfloor$ (Step 3). 

\emph{Step 1: $C_n$ contains acyclic bonds:}  If there exist flags, we can find an atom $x_i$ such that removing $x_i$ removes exactly one flag.  We can count the energy contribution of  this  flag by at least $-1$ and we estimate the energy of the rest of the configuration   by induction. By Lemma \ref{LemmaPropertiesbeta} 1) we get  
\begin{align*}
\mathcal{E}(C_n) \ge  -1 + \mathcal{E}(C_n \setminus \lbrace (x_i,q_i)\rbrace) \EEE  \geq -1 -\lfloor \beta(n-1)\rfloor \geq -\lfloor \beta(n)\rfloor.
\end{align*}
 Equality also shows that $C_n \setminus \lbrace (x_i,q_i)\rbrace$ has $\lfloor \beta(n-1) \rfloor$ bonds by induction and $C_n$ has  $\lfloor \beta(n-1) \rfloor+1 = \lfloor \beta(n) \rfloor$ bonds. \EEE

 We now suppose that a bridge exists. \EEE  Consider the two sub-configurations $C'_m$ and $C'_{n-m}$ which are connected by the bridge.  By the definition of   bridges we have that both $C'_m$ and $C'_{n-m}$ contain at least one simple cycle  and therefore, by Lemma \ref{LemmaPolygon}, $m,n-m \geq 6$. The energy contribution of the bridge is greater or equal to $-1$. Using the induction assumption and Lemma \ref{LemmaPropertiesbeta} 2) we get
\begin{align*}
\mathcal{E}(C_n)  \ge \mathcal{E}(C'_m) + \mathcal{E}(C'_{n-m})  - 1 \EEE \ge  -\lfloor \beta(m)\rfloor - \lfloor \beta(n-m)\rfloor -1 \geq -\lfloor \beta(n)\rfloor.
\end{align*}
 As before, equality also implies that $C_n$ has $\lfloor \beta(m)\rfloor + \lfloor \beta(n-m)\rfloor +1 = \lfloor \beta(n) \rfloor$ bonds.

\emph{Step 2: $\mathcal{E}(C_n) = -b$ for connected $C_n$ with no acyclic bonds:} Suppose by contradiction that \EEE the statement was false. Then,  by Remark \ref{rem: repulsionsfree}, \EEE   there exist  $x_1,x_2 \in X_n$ such that $q_1=q_2$ and $|x_1-x_2|<\sqrt{3}$ or $q_1 =-q_2$ and $1<|x_1-x_2|\leq r_0$. If \EEE  $x_1 \in \partial X_n$ or $x_2 \in \partial X_n$, by \EEE using Lemma \ref{LemmaBoundaryEnergy} and Remark \ref{rem: bdy}(ii) \EEE we have the strict inequality \EEE
\begin{align*}
\mathcal{E}^{\mathrm{bnd}}(C_n) > -\frac{3}{2}d + 3
\end{align*}
and by  induction assumption   we have
\begin{align*}
\mathcal{E}^{\mathrm{bulk}}(C_n) \geq -\lfloor \beta(n-d)\rfloor.
\end{align*}
On the other hand, if we have   $x_1, x_2 \notin \partial X_n$, we calculate by Lemma \ref{LemmaBoundaryEnergy} and induction \EEE
\begin{align*}
\mathcal{E}^{\mathrm{bnd}}(C_n) \geq -\frac{3}{2}d + 3, \ \ \ \ \ \ \  \mathcal{E}^{\mathrm{bulk}}(C_n) > -\lfloor \beta(n-d)\rfloor.
\end{align*}
 Here, the estimate for the bulk part is strict. Indeed, since in this case $C_n^{\rm bulk}$, which consists of $n-d$ particles, \EEE  is not repulsion-free or has bonds  longer than 1, \EEE  by the induction assumption  it cannot be a ground state, see Remark \ref{rem: repulsionsfree}. \EEE
In both cases it holds that
\begin{align*}
\mathcal{E}(C_n) > -\left\lfloor \frac{3}{2}n -\sqrt{\frac{3}{2}(n-d)}\right\rfloor +3.
\end{align*}
Since the right hand side is an integer, we obtain
\begin{align}\label{Integerinequality1}
-(\lfloor -\mathcal{E}(C_n)\rfloor +1) \geq -\frac{3}{2}n +\sqrt{\frac{3}{2}(n-d)} + 3.
\end{align}
 Recall that we assumed $\mathcal{E}(C_n) > -b$ by contradiction, \EEE which   implies  $-(\lfloor-\mathcal{E}(C_n)\rfloor+1) \geq -b$. \EEE Now by Lemma \ref{LemmaBoundaryestimate}  we obtain
\begin{align*}
n-d \geq 4(\lfloor-\mathcal{E}(C_n)\rfloor+1) + 6 -5n,
\end{align*}
 where we used that $\eta \ge 0$. \EEE Using the above inequality and (\ref{Integerinequality1}) we obtain
\begin{align*}
-(\lfloor -\mathcal{E}(C_n)\rfloor +1) \geq -\frac{3}{2}n +\sqrt{\frac{3}{2}( 4 \EEE (\lfloor-\mathcal{E}(C_n)\rfloor+1) + 6 -5n)} + 3.
\end{align*}
Now we can use Lemma \ref{LemmaSquareroot} with  $j=3$, $m=6$, and $x = -(\lfloor -\mathcal{E}(C_n)\rfloor +1)$ \EEE  to obtain
\begin{align*}
-(\lfloor -\mathcal{E}(C_n)\rfloor +1) \geq -\frac{3}{2}n + \sqrt{\frac{3}{2}n}.
\end{align*}
The last inequality implies $\mathcal{E}(C_n) >  -\lfloor \beta(n) \rfloor$ \EEE  contradicting \eqref{Groundstateineq1}.

 \emph{Step 3: $\mathcal{E}(C_n) = -\lfloor \beta(n) \rfloor$ for connected $C_n$ with no acyclic bonds:} \EEE
  Due to (\ref{Groundstateineq1}), it suffices to prove $\mathcal{E}(C_n) \geq -\lfloor \beta(n)\rfloor$. Again we proceed by induction.   By Lemma \ref{LemmaBoundaryestimate}  and the induction assumption \EEE we obtain 
\begin{align*}
\mathcal{E}^{\mathrm{bnd}}(C_n) \geq -\frac{3}{2}d+3,\quad \mathcal{E}^{\mathrm{bulk}}(C_n) \geq -\lfloor \beta(n-d)\rfloor.
\end{align*}
This gives
\begin{align*}
\mathcal{E}(C_n) \geq -\frac{3}{2}n +\sqrt{\frac{3}{2}(n-d)} + 3.
\end{align*}
By Lemma \ref{LemmaBoundaryestimate} and Step 2 \EEE we obtain $n-d \geq  - 4 \mathcal{E}(C_n) \EEE + 6 -5n$.
This yields
\begin{align*}
\mathcal{E}(C_n) \geq -\frac{3}{2}n + \sqrt{\frac{3}{2}(-4\mathcal{E}(C_n)-5n +6)}  + 3.\EEE
\end{align*}
Applying Lemma \ref{LemmaSquareroot} with $j=3$, $m=6$,  and $x = \mathcal{E}(C_n)$ \EEE  we obtain $\mathcal{E}(C_n) \geq -\beta(n)$. \EEE Finally, since $\mathcal{E}(C_n)$ is an integer due to Step 2, we conclude $\mathcal{E}(C_n) \ge -\lfloor \beta(n) \rfloor$. \EEE 
\end{proof}

For later purposes, we observe that the calculations of Step 2 and Step 3 in the previous proof can be refined.  Recall  the definition of the  excess of edges $\eta$ in \eqref{Excess}. \EEE

\begin{lemma}[Refined energy inequality for $\eta$]\label{lemma: eta energy}
 Let $n \ge 6$ and  let $C_n$ be a ground state  with no acyclic bonds. \EEE Then  
$$ \mathcal{E}(C_n) \ge -\frac{3}{2}n  + \sqrt{\frac{3}{2}(n+\eta)}.$$
If $\mathcal{E}^{\rm bnd}(C_n) > -\frac{3}{2}d + 3$ or $\mathcal{E}^{\rm bulk}(C_n) > \lfloor \beta(n-d)\rfloor$, then 
$$ \mathcal{E}(C_n) \ge -\frac{3}{2}n  + \sqrt{\frac{3}{2}(n+\eta-4)}+1.$$

\end{lemma}

\begin{proof}
By Lemma \ref{LemmaBoundaryEnergy} and Theorem \ref{TheoremEnergyGroundstates} applied on $C_n^{\mathrm{bulk}}$ we have
\begin{align}\label{eq: strict?}
\mathcal{E}^{\mathrm{bnd}}(C_n) \ge -\frac{3}{2}d +3, \ \ \ \ \  \mathcal{E}^{\mathrm{bulk}}(C_n) \geq -\lfloor \beta(n-d)\rfloor.
\end{align}
 By Lemma \ref{LemmaBoundaryestimate} and Theorem \ref{TheoremEnergyGroundstates} we get  $n-d = -4\mathcal{E}(C_n) + 6 + \eta -5n$.  This together with the summation of the two terms in \eqref{eq: strict?} yields \EEE
\begin{align}\label{eq: strict?2}
\mathcal{E}(C_n) \geq -\frac{3}{2}n + \sqrt{\frac{3}{2}(-4\mathcal{E}(C_n) -5n+6+\eta)} +3.
\end{align}
By applying Lemma \ref{LemmaSquareroot} with $j=3$ and $m=6+\eta$ we obtain
\begin{align*}
\mathcal{E}(C_n) \geq -\frac{3}{2}n + \sqrt{\frac{3}{2}(n+\eta)}.
\end{align*}
Finally, if $\mathcal{E}^{\rm bnd}(C_n) > -\frac{3}{2}d + 3$ or $\mathcal{E}^{\rm bulk}(C_n) > \lfloor \beta(n-d)\rfloor$, i.e., one inequality in \eqref{eq: strict?} is strict, we can replace $3$ by $4$ in \eqref{eq: strict?2} since $\mathcal{E}(C_n)$ is an integer. Then applying again Lemma  \ref{LemmaSquareroot} with $j=4$, $m=6+\eta$ we obtain 
$
\mathcal{E}(C_n) \geq -\frac{3}{2}n + \sqrt{\frac{3}{2}(n+\eta-4)} +1.
$
\end{proof}

\section{Characterization of  ground states}\label{sec: characterization}

 In this section we characterize the ground states of (\ref{Energy}). We do not provide a complete characterization for  $n < 10$ since the system is highly flexible  in those cases. Some of the ground states for $n<10$ are pictured in Fig.~\ref{FigureFlexible}. We will start by  providing some geometric facts about ground states.  Afterwards, \EEE we formulate and prove the first main result of the section which shows \EEE that ground states  consist of hexagonal cycles  except for possibly (at most)   two flags   or   one octagon at the boundary, see Proposition \ref{TheoremGroundstatesleq-new}. \EEE   Finally, for $n \ge 30$, we will be able to  prove \EEE that no octagons occur which will conclude \EEE    the proof of Theorem \ref{TheoremGroundstatesleq31}. The proof will also show that, among the ground states for $10 \le n \le 29$, an octagon can only occur in the cases  $n=12, 15,18,21,29$, \EEE see Remark \ref{rem: octogons} and Fig.~\ref{FigureOctagons}. \EEE

\begin{remark}\label{rem: method}
{\normalfont

We briefly remark that in the following our strategy deviates considerably from the one in \cite{Mainini} due to the different modeling assumptions concerning repulsive and angular potentials, see Remark \ref{rem: angle/bond} for  details. On the one hand, for $n \le 29$ indeed more flexible structures may occur \EEE which are not subsets  of the hexagonal lattice.  On the other hand, although we eventually will prove that for $n \ge 30$ ground states essentially assemble themselves in the hexagonal lattice, we need a different approach compared to \cite{HR, Mainini-Piovano, Mainini, Radin}: differently  to the proof by induction performed there, we cannot use the property that ground states are subsets of the hexagonal lattice in the induction hypothesis (consider, e.g., the step from $29$ to $30$). \EEE Therefore,  finer geometric considerations are necessary which are developed in two steps: first, we prove by induction that the breaking of the hexagonal symmetry due to presence of non-hexagonal  elementary \EEE polygons can only occur on the boundary. Then we show that for $n \ge 30$ the existence of such defects leads to an energy exceeding \eqref{Energygroundstates}.

}
\end{remark}

\begin{figure}[htp]
\centering
\begin{tikzpicture}[scale=.8]
\draw[thin](0,0)--(3,0);
\foreach \j in {0,1}{
\draw[fill=black](2*\j,0) circle(.05);
\draw[fill=white](2*\j+1,0) circle(.05);
}
\begin{scope}[shift={(4,0)}]
\draw[thin](0,0)--++(70:1)--++(10:1)--++(-10:1)--++(-70:1)--++(70:-1)--++(190:1)--++(170:1)--++(110:1);

\draw[fill=white](0,0)++(70:1)circle(.05)++(10:1)++(-10:1)circle(.05)++(-70:1)++(70:-1)circle(.05)++(190:1)++(170:1)circle(.05)++(110:1);

\draw[fill=black](0,0)circle(.05)++(70:1)++(10:1)circle(.05)++(-10:1)++(-70:1)circle(.05)++(70:-1)++(190:1)circle(.05)++(170:1)++(110:1);
\end{scope}

\begin{scope}[shift={(10.5,-.75)}]
%
%
\draw[thin](0,0)--++(-60:1);
\draw[thin](120:1)++(180:1)--++(180:1);
\draw[thin](120:1)++(60:1)--++(60:1);

\draw[fill=white](0,0)++(-60:1)circle(.05);
\draw[fill=white](120:1)++(180:1)++(180:1)circle(.05);
\draw[fill=white](120:1)++(60:1)++(60:1)circle(.05);
\draw[thin](60:1)--++(120:1)--++(180:1)--++(240:1)--++(300:1)--++(0:1)--++(60:1);
\draw[fill=black](60:1)++(120:1)circle(.05)++(180:1)++(240:1)circle(.05)++(300:1)++(0:1)circle(.05);

\draw[fill=white](60:1)circle(.05)++(120:1)++(180:1)circle(.05)++(240:1)++(300:1)circle(.05);
\end{scope}

\begin{scope}[shift={(12,0)},rotate=-90]

\draw[thin](0,0)--++(30:1)--++(90:1)--++(90:1)--++(150:1)--++(210:1)--++(270:1)--++(270:1)--++(330:1);

\draw[thin](0,0)++(150:1)++(90:1)++(90:1)--++(150:1);
\draw[fill=black](0,0)++(150:1)++(90:1)++(90:1)++(150:1)circle(.05);

\draw[fill=black](0,0)circle(.05)++(30:1)++(90:1)circle(.05)++(90:1)++(150:1)circle(.05)++(210:1)++(270:1)circle(.05)++(270:1)++(330:1);
\draw[fill=white](0,0)++(30:1)circle(.05)++(90:1)++(90:1)circle(.05)++(150:1)++(210:1)circle(.05)++(270:1)++(270:1)circle(.05)++(330:1);
\end{scope}
\end{tikzpicture}
\caption{Some ground states for $n<10$.}
\label{FigureFlexible}
\end{figure}
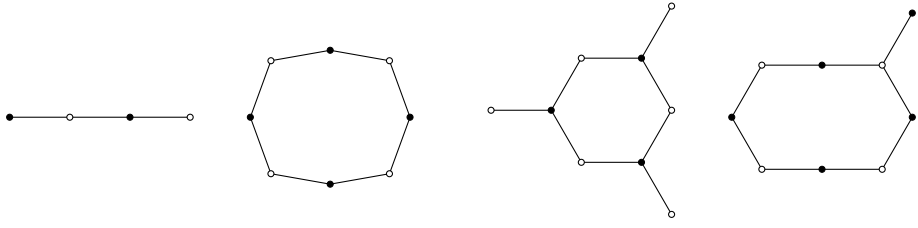

\subsection{Geometric properties of ground states}

In this section we collect some geometric properties of ground states. \EEE
 We start with an elementary property.

\begin{lemma}[Bridges\EEE]\label{lemma: bridges}
Ground states for $n \ge 13$ do not contain bridges.
\end{lemma}

\begin{proof}
Suppose that a bridge exists. Consider the two sub-configurations $C'_m$ and $C'_{n-m}$ which are connected by the bridge. As the energy contribution of the bridge is greater or equal to $-1$, we get by Theorem \ref{TheoremEnergyGroundstates}, Lemma \ref{LemmaPropertiesbeta} 2),  the fact that $n \ge 13$, and $V_{\rm r} \ge 0$ \EEE
\begin{align*}
\mathcal{E}(C_n) \ge -\lfloor \beta(m)\rfloor - \lfloor \beta(n-m)\rfloor -1 > -\lfloor \beta(n)\rfloor.
\end{align*}
This contradicts the fact that $C_n$ is a ground state.
\end{proof}

 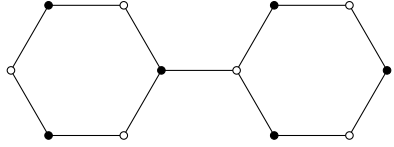
\begin{figure}[htp]
 \begin{tikzpicture}
 \draw[thin](0,0)--++(0:1)--++(60:1)--++(120:1)--++(180:1)--++(240:1)--++(300:1);
  \draw[thin](2,0)++(60:1)++(300:1)--++(0:1)--++(60:1)--++(120:1)--++(180:1)--++(240:1)--++(300:1);
  \draw[thin](1,0)++(60:1)--++(0:1);
  
 \draw[fill=black](0,0)circle(.05)++(0:1)++(60:1)circle(.05)++(120:1)++(180:1)circle(.05)++(240:1)++(300:1);  
 \draw[fill=white](0,0)++(0:1)circle(.05)++(60:1)++(120:1)circle(.05)++(180:1)++(240:1)circle(.05)++(300:1);  
 
  \draw[fill=black](2,0)++(60:1)++(300:1)circle(.05)++(0:1)++(60:1)circle(.05)++(120:1)++(180:1)circle(.05)++(240:1)++(300:1);  
 \draw[fill=white](2,0)++(60:1)++(300:1)++(0:1)circle(.05)++(60:1)++(120:1)circle(.05)++(180:1)++(240:1)circle(.05)++(300:1); 
  
 \end{tikzpicture}
 \caption{A ground state configuration for $n=12$ containing a bridge.}
 \label{FigureBridge}
 \end{figure}

\begin{remark}\label{rem: bridge}
{\normalfont
 Lemma \ref{lemma: bridges} is sharp in the sense that for $n=12$ there exists a ground state that contains a bridge connecting two hexagons  (and then $\eta=0$)\EEE, cf.\ Table~\ref{table2}  and Figure~\ref{FigureBridge}. Also note that for $n \le 11$ ground states cannot contain bridges as each polygon in the bond graph has at least $6$ vertices  by Lemma \ref{LemmaPolygon}.
 }
 \end{remark}
 \EEE
 
  The next lemma states that the number of flags is at most  two. \EEE  Let us mention that this property also applies to the ground states of \cite{Mainini} although this has not been observed explicitly there. \EEE

\begin{lemma}[Flags\EEE] \label{LemmaOctagonflag}  Let $n \geq 10$ and let $C_n$ be a ground state. Then the bond graph of $C_n$ contains at most $2$ flags.
\end{lemma}
\begin{proof} 
Assume  by contradiction \EEE that there exist $j \ge 3$ \EEE flags. Using the fact that a flag contributes at least $-1$ to the energy  and applying \EEE Theorem \ref{TheoremEnergyGroundstates} on the sub-configuration obtained after removing the flags,  we   have  $\mathcal{E}(C_n) \geq -j -\lfloor \beta(n-j)\rfloor$. By \eqref{eq: beta definition} we then get 
\begin{align*}
\mathcal{E}(C_n) & \ge  -\Big(\frac{3}{2}n -\frac{1}{2}j -\sqrt{\frac{3}{2}(n-j)}\Big)  = -\Big(\frac{3}{2}n - \sqrt{\frac{3}{2}n} +\sqrt{\frac{3}{2}n}-\frac{1}{2}j -\sqrt{\frac{3}{2}(n-j)}\Big) \\&  = -\Big( \frac{3}{2}n - \sqrt{\frac{3}{2}n} -\frac{1}{2}j +\frac{\frac{3}{2}j}{\sqrt{\frac{3}{2}n}+\sqrt{\frac{3}{2}(n-j)}}\Big).
\end{align*}
The function $  f(j) := \EEE -\frac{1}{2}j +\frac{\frac{3}{2}j}{\sqrt{\frac{3}{2}n}+\sqrt{\frac{3}{2}(n-j)}}$ is non-increasing in $j$, non-positive, and we have that  $f(3) \leq -1$ \EEE for $n\geq 16$. With the above estimate this implies \begin{align*}
\mathcal{E}(C_n) > -\lfloor \beta(n)\rfloor
\end{align*}
which leads to a contradiction  to Theorem \ref{TheoremEnergyGroundstates}  in the cases $n\geq 16$.  The cases  $10 \leq n \leq 15 $  can be checked directly 
 by comparing the above formula $\mathcal{E}(C_n) \geq -j -\lfloor \beta(n-j)\rfloor, j\geq 3$,  with $\lfloor\beta(n)\rfloor, n=10,\ldots,15$, cf.\ Table~\ref{table2}.  
\end{proof}

 Lemma \ref{LemmaOctagonflag} is sharp in the sense that for $n=9$ there exists a ground state that contains three flags in its bond graph, cf. Table~\ref{table2}  and the third configuration in Fig.~\ref{FigureFlexible}.  \EEE

\textbf{Equilibrated atoms:} We say an atom $x \in X_n$ is \emph{equilibrated} if all bond-angles at $x$ lie in $\lbrace \frac{2\pi}{3}, \frac{4\pi}{3} \rbrace$. By $\mathcal{A}$ we denote the atoms which are \emph{not} equilibrated. By $\mathcal{A}_{\rm bulk} \subset \mathcal{A} $ we denote the bulk atoms which are not equilibrated.  Note that if $\mathcal{A}=\emptyset$ and  $C_n$ \EEE is connected, then $X_n$ is a subset of the hexagonal lattice \EEE $\mathcal{L}$.  The following properties will be useful in the sequel. 

\begin{lemma}[Regular hexagons and bond-angles]\label{LemmaHexagon} 
Let $C_n$ be a ground state.
Then all hexagons are regular with unit bond length and  have \EEE alternating charge. All bond-angles $\theta$ satisfy  $\frac{2\pi}{3} \le \theta \le \frac{4\pi}{3}$.   If $x \in \mathcal{A}$, then $x$ is \EEE $2$-bonded and the bond angles lie in $(\frac{2\pi}{3},\frac{4\pi}{3})$.
\end{lemma}

\begin{proof}
By Theorem \ref{TheoremEnergyGroundstates} and Remark \ref{rem: repulsionsfree} we have that all bonds in the bond graph are of unit length and that the configuration is repulsion-free.  An additional necessary condition for equality is that all bond-angles $\theta$ satisfy 
\begin{align}\label{eq: larger than 2pi3}
\frac{2\pi}{3} \leq \theta \leq  \frac{4\pi}{3}. 
\end{align}
 In fact, suppose that $x_1,x_0,x_2$ form the angle $\theta$. Since $x_1,x_2$ are neighbors of $x_0$, we have  $q_1=q_2$ by Lemma \ref{LemmaNeighborhood}. The above mentioned necessary conditions imply $|x_1-x_0| = |x_2-x_0|=1$ and  $V_{\mathrm{r}}(|x_1-x_2|)=0$. The latter  only holds if $|x_1-x_2| \ge \sqrt{3}$ by $(\mathrm{vii})$. Simple trigonometry then yields $\frac{2\pi}{3} \leq \theta \leq  \frac{4\pi}{3}$

  From this discussion \EEE we derive that the edges of each hexagon necessarily need to have length $1$ and the interior angles are larger or equal to $2\pi/3$. As the sum of the interior angles in a  planar hexagon   equals \EEE  $4\pi$, we get that each interior angle is $2\pi/3$, i.e., each hexagon is indeed a regular hexagon with unit bond length.  The charge is alternating by Lemma \ref{LemmaNeighborhood}. \EEE

Consider an atom $x_i$ with a bond-angle $\theta_1$ which satisfies $\frac{4\pi}{3}>\theta_1 > \frac{2\pi}{3}$. Suppose by contradiction that $x_i$ had more than two bonds (i.e., three bonds, see Lemma \ref{LemmaNeighborhood}). Summing up all the three bond-angles $\theta_1,\theta_2,\theta_3$ at $x_i$ we get that $\min\lbrace \theta_2,\theta_3 \rbrace < 2\pi/3$. This, however, contradicts  \eqref{eq: larger than 2pi3}.  
\end{proof}

%

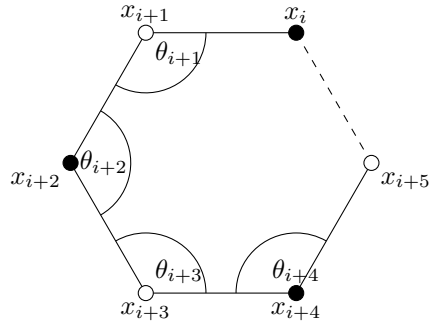
\begin{figure}[htp]
\begin{tikzpicture}[scale =2]
\draw[thin](0,0)--++(-60:1)--++(0:1)--++(60:1);
\draw[thin](0,0)--++(60:1)--++(0:1);
\draw[dashed,thin](0,0)++(60:1)++(0:1)--++(-60:1);

\draw[fill=black](0,0)circle(.05)++(-60:1)++(0:1)circle(.05)++(60:1);
\draw[fill=white](0,0)++(-60:1)circle(.05)++(0:1)++(60:1)circle(.05);
\draw[fill=white](0,0)++(60:1)circle(.05)++(0:1);
\draw[fill=black](0,0)++(60:1)++(0:1)circle(.05);
\draw(0,0) node[anchor=north east]{$x_{i+2} $} ++(-60:1)node[anchor=north]{$x_{i+3} $}++(0:1)node[anchor=north]{$x_{i+4} $}++(60:1)node[anchor=north west]{$x_{i+5} $};
\draw[thin](0,0)++(60:1) node[anchor=south]{$x_{i+1} $}++(0:1)  node[anchor=south]{$x_{i} $};

\draw[thin](-60:.4) arc(-60:60:.4);
\draw[thin](0,0) node[anchor=west]{$\theta_{i+2} $};

\draw[thin](60:1)++(240:.4) arc(240:360:.4);
\draw[thin](60:1) node[anchor=north west]{$\theta_{i+1} $};


\draw[thin](-60:1)++(0:.4) arc(0:120:.4);
\draw[thin](-60:1) node[anchor=south west]{$\theta_{i+3} $};

\draw[thin](-60:1)++(0:.6) arc(180:60:.4);
\draw[thin](-60:1)++(1,0) node[anchor=south]{$\theta_{i+4} $};

\end{tikzpicture}
\caption{Part of an \EEE octagon with $4$ consecutive angles $\theta_{i+1},\ldots, \theta_{i+4}$ \EEE equal to $\frac{2\pi}{3}$.}
\label{Figureangles}
\end{figure}

 We observe \EEE that octagons contain non-equilibrated atoms. More precisely, we have the following statement.

\begin{lemma}[Octagon]\label{lemma: octa}
Let $C_n$ be a   ground state  containing an octagon $\{x_0,\ldots,x_7\}$ in the bond graph. Let $\theta_i$, $i=0,\ldots,7$, be  the interior angles of the octagon. Then we have
\begin{align*}
(i) & \ \  (\theta_i,\theta_{i+1},\theta_{i+2},\theta_{i+3}) \neq  \left(\frac{2\pi}{3},\frac{2\pi}{3},\frac{2\pi}{3},\frac{2\pi}{3}\right) \text{ for all } i=0,\ldots,7,\\
(ii) & \ \ \theta_i < \frac{4\pi}{3} \text{ for all } i=0,\ldots,7.
\end{align*}
Here, the indices have to be understood $\mathrm{mod}\, 8$.
\end{lemma}

\begin{proof}
 By Remark \ref{rem: repulsionsfree} and Lemma \ref{LemmaHexagon} we have that
\begin{align}\label{E=Eimplication}
\frac{2\pi}{3} \leq \theta_i \leq \frac{4\pi}{3} \text{ and } |x_{i+1}-x_i|=1 \text{ for all } i=0,\ldots,7.
\end{align}
Suppose by contradiction that (i) was wrong. If there are more than three consecutive angles of size $\frac{2\pi}{3}$,  as indicated in Fig.~\ref{Figureangles},  then the bond graph would contain an additional hexagon and at least one triangle or square. This is a contradiction to the fact that the bond graph contains an octagon.

 We now show (ii). \EEE Assume by contradiction that, without restriction,  $\theta_0 =\frac{4\pi}{3}$. Then by (\ref{E=Eimplication}) and the fact that the interior angles of the octagon sum to $6\pi$,  there holds for all $i=1,\ldots,7$
\begin{align*}
\theta_i + \frac{4\pi}{3} + 6\frac{2\pi}{3} \leq\sum_{j=0}^7 \theta_j= 6\pi.
\end{align*}
This implies   $\theta_i \leq \frac{2\pi}{3}$ for all $i=1,\ldots,7$. Again using (\ref{E=Eimplication}) this yields $\theta_i = \frac{2\pi}{3}$ for all $i=1,\ldots, 7$ which contradicts (i).
\end{proof}

\EEE

The following lemma investigates the properties of a configuration in which a non-equilibrated bulk atom is present. Roughly speaking, it states that the existence of such an atom induces the existence of more non-equilibrated atoms and a certain excess of edges $\eta$. Note that at this point our analysis deviates significantly from \cite{Mainini}: in a model with angle potentials favoring $\frac{2\pi}{3}$ angles, \EEE it is obvious that non-equilibrated atoms cannot exist in ground states.   \EEE

\begin{lemma}[Non-equilibrated atoms\EEE]\label{LemmaOctagon} Let $C_n$ be a   ground state   with no acyclic bonds in the bond graph. Assume that $\# \mathcal{A}_{\rm bulk} \ge 1$. Then one of the following holds:    
\begin{itemize}
\item[i)] $\# \mathcal{A}  \ge 2$ and \EEE  $\eta \geq 6$,
\item[ii)] $\# (\mathcal{A} \setminus \mathcal{A}_{\rm bulk}) \ge 1$ and $\eta = 4$. \EEE
\end{itemize}
\end{lemma}

\begin{proof} 
 
We first prove that each polygon containing non-equilibrated atoms has at least eight vertices and contains at least two non-equilibrated atoms. Then we show that $\eta \ge 4$ and that in the case $\eta < 6$, we have $\# (\mathcal{A} \setminus \mathcal{A}_{\rm bulk}) \ge 1$. The statement clearly follows from these claims. \EEE

\textbf{Claim 1:}  Each polygon containing a non-equilibrated atom has at least eight vertices. 

\noindent \emph{Proof of Claim 1:} Due to Lemma \ref{LemmaHexagon}, both angles at  non-equilibrated atoms lie in $(\frac{2\pi}{3},\frac{4\pi}{3})$ and therefore \EEE each polygon containing a non-equilibrated atom is not a hexagon. Thus, it has at least eight vertices  by Lemma \ref{LemmaPolygon}. \EEE

\EEE

\textbf{Claim 2:}  Each polygon contains either no or at least two non-equilibrated atoms. \EEE

\noindent \emph{Proof of Claim 2:}  Consider a polygon with $k$ edges which contains a non-equilibrated atom \EEE with interior angle $\theta_1 \neq \frac{2\pi}{3}, \frac{4\pi}{3}$. By Lemma \ref{LemmaPolygon} we have that $k \in 2\mathbb{N}$. Assume by contradiction that all the other angles $\theta_2, \ldots, \theta_k$ \EEE are either $\frac{2\pi}{3}$ or $\frac{4\pi}{3}$.  We have
\begin{align*}
\sum_{j=1}^k \theta_j =\pi(k-2).
\end{align*}
Since  we assumed that  $\theta_j$, $j \ge 2$, \EEE   are integer multiples of $\frac{2\pi}{3}$, we have
\begin{align*}
\theta_1 + k' \frac{2\pi}{3}= \pi(k-2),
\end{align*}
where $k' \in \mathbb{N}$ is given by $k'=\#\{j : \theta_j =\frac{2\pi}{3}\}+2\#\{j : \theta_j =\frac{4\pi}{3}\}$. This implies
\begin{align*}
k =\frac{\theta_1}{\pi}+ 2 + \frac{2}{3}k'.
\end{align*}
Since both $k,k' \in \mathbb{N}$  and $\frac{2\pi}{3} < \theta_1 < \frac{4\pi}{3}$ by Lemma \ref{LemmaHexagon}, \EEE there exists only a solution to the equation if $\theta_1 =\pi$ and $k' \in 3\mathbb{N}$. This implies that $k$ is odd:  a contradiction.

 \textbf{Claim 3:} \EEE  We have $\eta \ge 4$. If $\eta<6$, then $\# (\mathcal{A} \setminus \mathcal{A}_{\rm bulk}) \ge 1$.   \EEE

 \noindent \emph{Proof of Claim 3:} \EEE  Since $\# \mathcal{A}_{\rm bulk} \ge 1$,  there exists a non-equilibrated bulk atom. As $C_n$ does not have acyclic bonds, we observe that this atom is a vertex of at least two polygons. \EEE  Claim 1 then yields that there have to be at least two polygons  with at least eight vertices, i.e., $\eta \geq 4$, $\eta \in 2\mathbb{N}$. \EEE

 It remains  to show that, if $\eta =4$, then   $\# (\mathcal{A} \setminus \mathcal{A}_{\rm bulk}) \ge 1$.     Assume by contradiction that $\eta=4$ and $\mathcal{A} \setminus \mathcal{A}_{\rm bulk} = \emptyset$. As $\eta=4$, the two non-hexagons have to be octagons.  By Claim 1, \EEE the assumption $\mathcal{A} \setminus \mathcal{A}_{\rm bulk} = \emptyset$,  and the fact that $C_n$ has no acyclic bonds, \EEE we find that all non-equilibrated atoms are contained in both octagons.

Denote the interior angles in the first octagon different from $ \{\frac{2\pi}{3},\frac{4\pi}{3}\}$ by $\alpha_1,\ldots,\alpha_k$, where $1 \le k \le 8$. Similarly, the angles in the second octagon  different from $ \{\frac{2\pi}{3},\frac{4\pi}{3}\}$ \EEE are denoted by $\beta_1,\ldots,\beta_k$, where without restriction $\alpha_i$ and $\beta_i$ lie at the same atom. Note that  $\beta_i = 2\pi - \alpha_i$ for $i=1,\ldots,k$ as non-equilibrated atoms are $2$-bonded. Due to Lemma \ref{lemma: octa}(ii), all other interior angles of the octagons are $\frac{2\pi}{3}$.  Thus, by the interior angle sum of the octagons we obtain
\begin{align*}
\sum_{j=1}^k \alpha_j +   (8-k) \EEE \frac{2\pi}{3} = 6\pi,  \ \ \ \  \sum_{j=1}^k \beta_j +  (8-k) \EEE \frac{2\pi}{3} = 6\pi.
\end{align*}
Using $\sum_{j=1}^k \beta_j = 2\pi k- \sum_{j=1}^k \alpha_j $ and summing the two equations, we obtain the unique solution $k=2$. In particular, this implies $\# \mathcal{A} = 2$. Denote the two atoms in $\mathcal{A}$ by $x_1$ and $x_2$. From Lemma \ref{lemma: octa}(i) we get that the atoms $x_1$ and $x_2$ lie `on opposite sides' of the octagons, \EEE i.e., the shortest path in the bond graph connecting $x_1$ and $x_2$ has length $4$. Then we also see that $\alpha_1=\alpha_2=\beta_1=\beta_2=\pi$ by a simple geometric argument.  (An octagon with this geometry is depicted in the rightmost configuration in Fig.~\ref{FigureFlexible}.) \EEE Finally, this yields that the two octagons are identical up to an isometry. This, however, contradicts the fact that both non-equilibrated atoms $x_1$ and $x_2$ are contained in both octagons. Thus, $\eta=4$ and $\mathcal{A} \setminus \mathcal{A}_{\rm bulk} = \emptyset$ is not possible. This concludes the proof. \EEE
\end{proof}

 Based on Lemma  \ref{LemmaOctagon}, we now show that non-equilibrated bulk atoms cannot exist in ground states with no acyclic bonds. \EEE
\begin{lemma}[Non-equilibrated bulk atoms\EEE] \label{LemmaAbulk} Let $n \geq 1$ and let $C_n$ be a ground state with no acyclic bonds. \EEE  Then $\mathcal{A}_{\mathrm{bulk}} = \emptyset$. 
\end{lemma}
\begin{proof} We prove the statement by induction. We first note that the statement is true for $1 \le n \le 9$. In fact, in this case the bond graph contains at most one polygon by Lemma \ref{LemmaPolygon} and Lemma \ref{LemmaHexagon}. This \EEE implies $\mathcal{A}_{\mathrm{bulk}} \subset X_n^{\mathrm{bulk}}=\emptyset$.  Let $n \ge 10$. \EEE  We  assume that the result has been proven for $m<n$ and proceed to show  the statement for $n$. Assume by contradiction that $\mathcal{A}_{\mathrm{bulk}} \neq \emptyset$. By Lemma \ref{LemmaOctagon} there are two cases to consider:
\begin{itemize}
\item[i)] $\#(\mathcal{A}\setminus \mathcal{A}_{\mathrm{bulk}}) \geq 1$ and $\eta=4$.
\item[ii)] $\# \mathcal{A} \geq 2$ and $\eta \geq 6$.
\end{itemize}

\noindent \emph{Proof for Case $\mathrm{i)}$:} By Lemma \ref{LemmaBoundaryEnergy} we obtain the strict inequality
\begin{align*}
\mathcal{E}^{\mathrm{bnd}}(C_n) >-\frac{3}{2}d +3.
\end{align*}
By Lemma \ref{lemma: eta energy}  (with strict inequality)  this gives  
\begin{align}\label{EnergyestimateOctagon}
\mathcal{E}(C_n) \geq -\frac{3}{2}n + \sqrt{\frac{3}{2}(n+\eta-4)}+1.
\end{align}
Note that $\eta=4$ and therefore we have
\begin{align}\label{Strictenergyboundetabig}
\mathcal{E}(C_n) \geq -\frac{3}{2}n + \sqrt{\frac{3}{2}n}+1 > -\lfloor \beta(n)\rfloor
\end{align}
which contradicts Theorem \ref{TheoremEnergyGroundstates} and the fact that $C_n$ is a ground state.

\noindent  \emph{Proof for  Case $\mathrm{ii})$}: We are now in the case that $\#\mathcal{A} \geq 2$ and $\eta \geq 6$. By the previous case we can assume that $\mathcal{A} \setminus \mathcal{A}_{\mathrm{bulk}} =\emptyset$. This implies  $\#\mathcal{A}_{\mathrm{bulk}} \geq 2$. \EEE After removing the boundary, we can suppose that
\begin{align}\label{eq: bulk-equ}
\mathcal{E}(C_n^{\rm bulk}) = - \lfloor \beta(n-d) \rfloor.
\end{align}
Indeed, if $\mathcal{E}(C_n^{\rm bulk}) > - \lfloor \beta(n-d) \rfloor$, we derive by  Lemma \ref{lemma: eta energy}  (with strict inequality) that \eqref{EnergyestimateOctagon} holds. Since $\eta \ge 4$, we get a contradiction exactly as in Case i), see \eqref{Strictenergyboundetabig}. Hence, $C_n$ as well as $C_n^{\mathrm{bulk}}$ are ground states. We now distinguish the following cases: 
\begin{itemize}
\item[a)] $C_n^{\mathrm{bulk}}$ contains at least two flags,
\item[b)] $C_n^{\mathrm{bulk}}$ contains a bridge,
\item[c)] $C_n^{\mathrm{bulk}}$ contains less than two flags and no bridge. \EEE
\end{itemize} 
\EEE

\noindent  \emph{Proof for  Case $\mathrm{a})$}: We use Lemma \ref{LemmaBoundaryEnergy} to obtain
 \begin{align}\label{ineq: bnd2flags}
 \mathcal{E}^{\mathrm{bnd}}(C_n) \geq -\frac{3}{2}d + 3.
 \end{align}
 Using the fact that a flag contributes at least $-1$ to the energy  and applying   Theorem \ref{TheoremEnergyGroundstates} on the sub-configuration obtained after removing exactly two flags from $C_n^{\mathrm{bulk}}$, \EEE we  obtain 
 \begin{align}\label{ineq: bulk2flags}
\mathcal{E}^{\rm bulk}(C_n)  = \mathcal{E}(C_n^{\mathrm{bulk}})  \EEE \geq  \EEE -2 - \lfloor\beta(n-d-2)\rfloor \geq -\frac{3}{2}(n-d) +\sqrt{\frac{3}{2}(n-d-2)} + 1.
\end{align}
Combining (\ref{ineq: bnd2flags})-(\ref{ineq: bulk2flags}) we obtain
\begin{align*}
\mathcal{E}(C_n) \geq -\frac{3}{2}n +\sqrt{\frac{3}{2}(n-d-2)} + 4.
\end{align*}
By Lemma \ref{LemmaBoundaryestimate}  and $b = - \mathcal{E}(C_n)$ we obtain $
\mathcal{E}(C_n) \geq -\frac{3}{2}n + 4 + \sqrt{\frac{3}{2}(-4\mathcal{E}(C_n) + 4+\eta-5n)}.
$
Lemma \ref{LemmaSquareroot} for $j=4$, $m= 4+ \eta$, and $x = \mathcal{E}(C_n)$ yields
\begin{align*}
\mathcal{E}(C_n) \geq -\frac{3}{2}n + \sqrt{\frac{3}{2}(n+10+\eta-16)} + 1 \geq -\frac{3}{2}n + \sqrt{\frac{3}{2}n}+1 > -\lfloor \beta(n)\rfloor,
\end{align*}
where we used $\eta \geq 6$. This contradicts Theorem \ref{TheoremEnergyGroundstates} and the fact that $C_n$ is a ground state. \EEE

 \noindent  \emph{Proof for  Case $\mathrm{b})$}: In view of Remark \ref{rem: bridge}, $C_n^{\mathrm{bulk}}$ can only contain a bridge if $n-d = 12$ and $C_n^{\mathrm{bulk}}$ consists of two regular hexagons connected with a bridge. This contradicts $\#\mathcal{A}_{\mathrm{bulk}} \geq 2$. \EEE

\noindent  \emph{Proof for  Case $\mathrm{c})$}:  Denote by $l \in \lbrace 0,1 \rbrace$ the number of flags of $C_n^{\mathrm{bulk}}$ and let $C_{n-d-l}^*$ be the configuration which arises by removing $l$ atoms from $C_n^{\mathrm{bulk}}$ such that $C_{n-d-l}^*$ has no acyclic bonds.  Observe that $\mathcal{E}(C_{n-d-l}^*)-l \le \mathcal{E}(C_n^{\mathrm{bulk}})$. Then also $C_{n-d-l}^*$ is a ground state since otherwise $\mathcal{E}(C_n^{\rm bulk}) > - \lfloor \beta(n-d) \rfloor$ by   Lemma \ref{LemmaPropertiesbeta} 1) which contradicts \eqref{eq: bulk-equ}. \EEE   As $\# \mathcal{A}_{\rm bulk} \ge 2$ and $l \le 1$, $C_{n-d-l}^*$ contains a non-equilibrated atom. Thus, due to the fact that all hexagons in the bond graph   are regular (see Lemma \ref{LemmaHexagon}),  we have that $\eta^{*} = \eta(C_{n-d-l}^*) \geq 2$, where $\eta^{*}$ denotes  the excess of edges \EEE of $C_{n-d-l}^*$.  By the induction assumption we have that $C_{n-d-1}^*$ has no non-equilibrated bulk atom and thus  \EEE   has a non-equilibrated boundary atom.   Thus, strict inequality holds in \eqref{BoundaryEnergyEstimate} for $C_{n-d-l}^*$, see Lemma \ref{LemmaBoundaryEnergy}. Therefore,  by Lemma \ref{lemma: eta energy}  (with strict inequality)  applied for  $C_{n-d-l}^*$ and  $\eta^{*}  \geq 2$ we obtain  \EEE 
\begin{align*}
\mathcal{E}^{\mathrm{bulk}}(C_n) = \mathcal{E}(C_n^{\mathrm{bulk}}) \ge \mathcal{E}(C_{n-d-l}^*)-l \EEE   \geq -\frac{3}{2}(n-d-l) + \sqrt{\frac{3}{2}(n-d-l-2)}+1-l.
\end{align*}
Using Lemma \ref{LemmaBoundaryEnergy} for $C_n$  and summing $\mathcal{E}^{\mathrm{bnd}}(C_n)$ and $\mathcal{E}^{\mathrm{bulk}}(C_n)$, we derive
\begin{align*}
\mathcal{E}(C_n) \geq -\frac{3}{2}n + 4 + \sqrt{\frac{3}{2}(n-d-2-l)}+ \frac{l}{2}.
\end{align*}
By Lemma \ref{LemmaBoundaryestimate}  and  $\eta \ge 6$ \EEE we obtain 
\begin{align*}
\mathcal{E}(C_n) \geq -\frac{3}{2}n + 4 + \frac{l}{2} + \sqrt{\frac{3}{2}(-4\mathcal{E}(C_n) + 10 \EEE -5n-l)}.
\end{align*}
Lemma \ref{LemmaSquareroot} for $j=4+\frac{l}{2}$, $m=  10 \EEE -l$, and $x = \mathcal{E}(C_n)$ yields
\begin{align*}
\mathcal{E}(C_n) &\geq -\frac{3}{2}n +1 +\frac{l}{2} + \sqrt{\frac{3}{2}(n-3l)} =  -\frac{3}{2}n +\sqrt{\frac{3}{2}n} + 1  +l \Big(\frac{1}{2} - \frac{9}{2(\sqrt{\frac{3}{2}n}+\sqrt{\frac{3}{2}(n-3l)})}\Big). \EEE
\end{align*}
 This estimate can be used to calculate $\mathcal{E}(C_n) \geq -\frac{3}{2}n + \sqrt{\frac{3}{2}n}+1 >-\lfloor\beta(n)\rfloor$ for all $n \ge 16$ when $l=1$  or for all  $n \geq 10$ when $l=0$. In the cases $10 \le n \le 15$, $l=1$, one can use $\mathcal{E}(C_n) \geq -\frac{3}{2}n +1 +\frac{1}{2} + \sqrt{\frac{3}{2}(n-3)}$ and compare this estimate directly with $\lfloor \beta(n)\rfloor$  to obtain $\mathcal{E}(C_n) > -\lfloor \beta(n)\rfloor$, \EEE  cf.\ Table \ref{table2}. In every case, this yields a  contradiction to the fact that $C_n$ is a ground state. \EEE 
\end{proof}

\subsection{Characterization of ground states: proof of Theorem \ref{TheoremGroundstatesleq31}}
 
 In this section we prove  Theorem \ref{TheoremGroundstatesleq31}. \EEE

\textbf{Boundary $k$-gon:} We say that a $k$-gon in the bond graph is a \textit{boundary} $k$-gon, whenever it shares at least one edge with the unbounded face.

The following proposition is the main ingredient for the proof of Theorem \ref{TheoremGroundstatesleq31}.

\begin{proposition}\label{TheoremGroundstatesleq-new} Let $n \geq10$ and let $C_n$ be a ground state. Then the bond graph consists only of hexagonal cycles except for  at most two flags and at most one boundary octagon. The bond graph cannot contain both flags and an octagon at the same time. \EEE
\end{proposition}

   Once this proposition is proven, for the proof of Theorem \ref{TheoremGroundstatesleq31} it remains to show that in the case $n \ge 30$ no octagons may occur. In fact, Theorem \ref{TheoremGroundstatesleq31} then follows from Theorem \ref{TheoremEnergyGroundstates}, Remark \ref{rem: main}, Lemma \ref{LemmaNeighborhood}, and Lemma \ref{lemma: bridges}. For \EEE $n = 9$, a ground state may contain an octagon and a flag at the same time, see  the rightmost configuration \EEE in Fig.~\ref{FigureFlexible}.  In this sense, the assumption $n \geq10$ in \EEE Proposition \ref{TheoremGroundstatesleq-new} is sharp.

\begin{proof}[Proof of Proposition \ref{TheoremGroundstatesleq-new}] 
Let $C_n$ be a ground state. Recall by Theorem \ref{TheoremEnergyGroundstates} that  $\mathcal{E}(C_n) = -b=-\lfloor \beta(n)\rfloor$ and  that   $C_n$ \EEE is connected. We divide the proof into several steps.  First, we  prove that ground states contain only hexagonal cycles if $\mathcal{A}=\emptyset$ (Claim 1). Then, in the case  $\mathcal{A} \neq \emptyset$, we prove  that at most one boundary octagon may exist (Claim 2). Finally, we show that the existence of a non-hexagonal cycle excludes the existence of flags (Claim 3).   The statement follows from Claim 1 - Claim 3 and Lemma \ref{LemmaOctagonflag}. \EEE    The claims are proven by contradiction.

 \noindent\textbf{Claim 1:} If $\mathcal{A}= \emptyset$, $C_n$ is defect-free.

 \noindent \emph{Proof of Claim 1}:  Suppose that there exists  a $k$-gon, $k \ge 8$. Since  $C_n$ \EEE is connected and $\mathcal{A}=\emptyset$, we have that  $C_n$ \EEE is a connected subset of the hexagonal lattice.  We observe that then our energy coincides with the one considered in \cite{Mainini}. We can repeat  the argument in  the proof of \cite[Proposition 6.7]{Mainini}, i.e., we can move boundary atoms inside this $k$-gon and observe that one can strictly lower the energy. This contradicts the fact that $C_n$ is a ground state. 

\noindent\textbf{Claim 2:}   For every ground state without flags  there exists at most one $k$-gon, $k\ge 8$. If it  exists, it has to be a boundary octagon.

 \noindent \emph{Proof of Claim 2}:   In view of  Claim 1, we can suppose $\mathcal{A} \neq \emptyset$ and $\eta \geq 2$.  As $\eta \ge 2$, $C_n$ contains no bridge (see Lemma \ref{lemma: bridges}, Remark \ref{rem: bridge}) and thus no acyclic bonds. \EEE  By   Lemma \ref{LemmaAbulk}   we have that $\mathcal{A}_{\mathrm{bulk}}=\emptyset$ and therefore $\mathcal{A}\setminus \mathcal{A}_{\mathrm{bulk}} \neq \emptyset$. Applying Lemma \ref{LemmaBoundaryEnergy}   and Lemma \ref{lemma: eta energy}  (with strict inequality) we obtain that
 \begin{align*}
 \mathcal{E}(C_n) \geq -\frac{3}{2}n +1 + \sqrt{\frac{3}{2}(n+\eta-4)}.
 \end{align*}
 In the case $\eta \geq 4$ we obtain a contradiction to the fact that $C_n$ is a ground state. Therefore,  we can assume that \EEE $\eta=2$, i.e., there exists exactly one octagon. We \EEE  apply Lemma \ref{LemmaHexagon}  to find that hexagons are regular which implies  $\mathcal{A}$ has to be contained in the octagon of the bond graph. By Lemma \ref{LemmaAbulk} \EEE we observe  $\mathcal{A} \subset \partial X_n$.  This implies that the octagon is a boundary octagon.

\noindent\textbf{Claim 3:}  A ground state cannot contain both  a  $k$-gon, $k \ge 8$, and a flag. 
 
  \noindent \emph{Proof of Claim 3}:  Assume by contradiction that there exists a $k$-gon and $l$ flags in the bond graph.  (As before in Claim 2, by Remark \ref{rem: bridge} there are no bridges.) \EEE Using the fact that a flag contributes at least $-1$ to the energy and removing the flags we obtain a sub-configuration $C_{n-l}$ satisfying $ \mathcal{E}(C_{n}) \ge \mathcal{E}(C_{n-l}) - l$. Then also $C_{n-l}$ is a ground state since otherwise $\mathcal{E}(C_{n})> - \lfloor \beta(n) \rfloor$ by   Lemma \ref{LemmaPropertiesbeta} 1) which contradicts Theorem \ref{TheoremEnergyGroundstates}. As $C_{n-l}$ has no acyclic bonds, we can use   Lemma \ref{LemmaAbulk}  to find \EEE $\mathcal{A}_{\mathrm{bulk}}=\emptyset$ and therefore $\mathcal{A}\setminus \mathcal{A}_{\mathrm{bulk}} \neq \emptyset$, where $\mathcal{A}_{\mathrm{bulk}}, \mathcal{A}$ correspond to configuration $C_{n-l}$.  Applying Lemma \ref{LemmaBoundaryEnergy}    and Lemma \ref{lemma: eta energy} on $C_{n-l}$ \EEE (with strict inequality for $\eta \ge 2$),  we obtain  
$$  \mathcal{E}(C_n) \geq   \mathcal{E}(C_{n-l}) - l \ge -l -\frac{3}{2}(n-l) +1 + \sqrt{\frac{3}{2}(n-l-2)}. $$
Setting $j = l +2$, we find 
\begin{align*}
\mathcal{E}(C_n)\ge -\Big( \frac{3}{2}n - \sqrt{\frac{3}{2}n} -\frac{1}{2}j +\frac{\frac{3}{2}j}{\sqrt{\frac{3}{2}n}+\sqrt{\frac{3}{2}(n-j)}}\Big).
\end{align*}
 Recall that $j \ge 3$. \EEE At this point, we can follow verbatim the proof of Lemma \ref{LemmaOctagonflag} to get the contradiction   $\mathcal{E}(C_n) > -\lfloor \beta(n)\rfloor$ for each $n \ge 10$. This contradicts Theorem \ref{TheoremEnergyGroundstates} and the fact that $C_n$ is a ground state. \EEE
\end{proof}

As a final preparation for the proof  of Theorem \ref{TheoremGroundstatesleq31}, we need the following elementary geometric lemma. \EEE

\begin{lemma}\label{LemmaOctagon3bdd} Let $C_n$ be a ground state with $\eta = 2$ that does not contain any acyclic bonds. Let $\{x_0,\ldots,x_7\}$ be the \EEE octagon in the bond graph.  Set $X^3:= \{x_i \in \{0,\ldots,7\}: x_i \text{ is } 3\text{-bonded}\}$. \EEE We have the following:
\begin{itemize}
\item[i)] If $ \# X^3 \EEE  \leq 3$, then $X^3$ \EEE is connected,
\item[ii)] If  $\#X^3 \in \{4,5\}$, then there exist $j_1,j_2,j_3 \in \{0,\ldots,7\}$ such that $X_n \setminus \{x_{j_1},x_{j_2},x_{j_3}\}$ is not connected,  $x_{j_1},x_{j_2},x_{j_3}$ are $2$-bonded, \EEE and  
\begin{align}\label{eq: remove3}
\mathcal{E}\big(C_n \setminus \{ (x_{j_1},q_{j_1}), (x_{j_2}, q_{j_2}),(x_{j_3},q_{j_3}) \} \big)   \le \EEE \mathcal{E}(C_n)+5.
\end{align}\EEE 
\item[iii)] If  $\# X^3 \geq 6$, then $C_n^{\mathrm{bulk}}$ is not connected.
\end{itemize}
\end{lemma}

\begin{proof} Let $\{x_0,\ldots,x_7\}$ be  the \EEE octagon in the bond graph.  All following \EEE statements are seen $\mathrm{mod}\, 8$ with respect to the numeration of $i$. Assume that $x_{i+1},x_{i-1} \in \mathcal{N}(x_i)$ for all $i=0,\ldots,7$ and $x_i \notin \mathcal{N}(x_j) $ for all $j \notin \{i-1,i+1\}$. Let $\theta_i$, $i=0,\ldots,7$,  be \EEE the interior angles of the octagon.  We start with three preliminary observations. \EEE

\textbf{Claim 1:}  If $x_i$ is $2$-bonded, \EEE then  $x_i \in \partial  X_n$.

\noindent \emph{Proof of Claim  1:}  Consider a  $2$-bonded $x_i$  and assume \EEE by contradiction that   $x_i \in X_n^{\mathrm{bulk}}$. By  Lemma \ref{lemma: octa}(ii) \EEE we have that $\theta_i < \frac{4\pi}{3}$ and therefore the other angle $\alpha$ at $x_i$ satisfies $\alpha >\frac{2\pi}{3}$. Since  $x_{i}$ is a bulk \EEE atom,  Lemma \ref{LemmaHexagon} and $\alpha > \frac{2\pi}{3}$ imply that $x_{i}$ is contained in the octagon and in another polygon which is not a hexagon \EEE since for hexagons all interior angles are equal to $\frac{2\pi}{3}$. This contradicts $\eta=2$.

\textbf{Claim 2:} If $x_i$ is $3$-bonded, then $x_{i-1}$ or $x_{i+1}$ is also $3$-bonded. \EEE

\noindent \emph{Proof of Claim 2:} Assume by contradiction that there exists a $3$-bonded $x_i$  such that  $x_{i-1}$ and $x_{i+1}$ are $2$-bonded.  Denote  by $z$ the third neighbor of $x_i$. \EEE Since the bond graph of $C_n$ does not contain acyclic bonds,  $x_i$ and $z$ are contained in a polygon. This contains necessarily either $x_{i-1}$ or $x_{i+1}$,  say $x_{i-1}$. \EEE Thus,  since $x_{i-1}$ is $2$-bonded, this polygon contains also $x_{i-2}$. Then  \EEE $x_{i-1}$ has to be a bulk \EEE atom.  This, however, contradicts Claim 1. \EEE

\textbf{Claim 3:} If $x_{i-1},x_i,x_{i+1}$ are $3$-bonded, then $x_{i-2}$ and $x_{i+2}$ are $2$-bonded.

\noindent \emph{Proof of Claim 3:} If there were four consecutive $3$-bonded atoms, there would be four consecutive interior angles of size $\frac{2\pi}{3}$, see Lemma \ref{LemmaHexagon}. This, however, contradicts Lemma \ref{lemma: octa}(i). \EEE

 We now proceed with the proof of the statement. \EEE Set $k=\# X^3\EEE$. First,  by Lemma \ref{LemmaHexagon}, Lemma \ref{lemma: octa}(ii), \EEE and the fact that 
$
\sum_i \theta_i = 6\pi
$
we have that $k \leq 6$.
 If $k=0$, there is nothing to prove. $k=1$ is not possible due to Claim 2. If $k=2$, again due to Claim 2, the  $3$-bonded atoms \EEE  have to be  bonded. \EEE  If $k=3$, the $3$-bonded atoms \EEE necessarily need to be of the form $x_{i-1},x_i,x_{i+1}$ for some $i=0,\ldots,7$. Otherwise, we have a contradiction to Claim 2. This proves $\mathrm{i)}$.
 
  Now suppose that $k \in \{4,5\}$. By  Lemma \ref{LemmaHexagon} we have that, if $x_i$ is $3$-bonded, then $\theta_i = \frac{2\pi}{3}$. By Lemma \ref{lemma: octa}(i) the $3$-bonded atoms  cannot be of the form $\{x_i,\ldots,x_{i+k-1}\}$. Hence, there exist $ 0 \le i_1< i_1+1 < j_1 <i_2 < i_2 +2 <j_2$ such that \EEE $x_{i_1},x_{i_1+1},  x_{i_2}, \EEE x_{i_2+1}$ are $3$-bonded and $x_{j_1},x_{j_2},x_{j_2+1}$ are $2$-bonded.  By Claim 1 we get that, if $x_j$ is $2$-bonded, then $x_j \in \partial X_n$.   Therefore, $x_{i_1}$ and $x_{i_2}$ are  connected only through paths going through $x_{j_1}$ or $x_{j_2}$  as otherwise one of the two atoms $x_{j_1}$ and $x_{j_2}$ would not be contained in the boundary. Thus, the set $X_n \setminus\{x_{j_1},x_{j_2},x_{j_2+1}\}$ is not connected.  Since $x_{j_1},x_{j_2},x_{j_2+1}$ are $2$-bonded and $x_{j_2}$ is bonded to $x_{j_2+1}$,  we remove exactly $5$ bonds. As each bond contributes at least $-1$ to the energy, this yields \EEE \eqref{eq: remove3}. \EEE
  This proves $\mathrm{ii)}$.
 
 
 Finally, suppose $k=6$. By Claim 3 we have that there exists $i \in \{0,\ldots,7\}$ such that $x_i,x_{i+1},x_{i+2},x_{i+4},x_{i+5},x_{i+6}$ are $3$-bonded. Since no acyclic bonds are present in the bond graph, we have that $x_{i+1},x_{i+5} \in X_n^{\mathrm{bulk}}$.  Moreover, arguing as in case  $\mathrm{ii)}$,  $x_{i+1}$ is not connected with $x_{i+5}$ \EEE in $X_n^{\mathrm{bulk}}$.  This \EEE proves $\mathrm{iii)}$.
\end{proof}

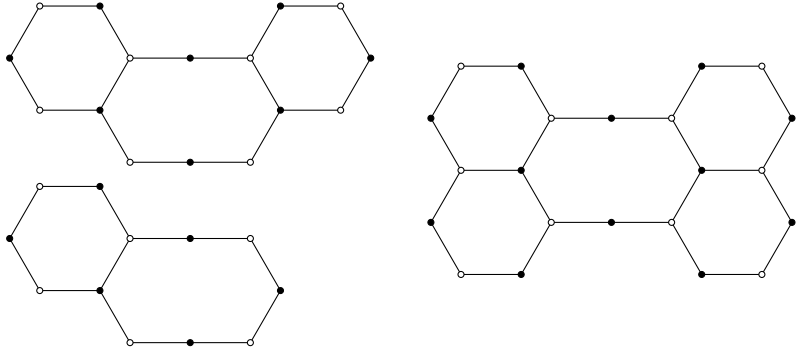
\begin{figure}[htp]
\centering
\begin{tikzpicture}[scale=.8]
\begin{scope}[shift={(0,-2)}]
\draw[thin](0:0)--++(60:1)--++(0:2)--++(300:1)--++(240:1)--++(180:2)--++(120:1);
\draw[thin](0:0)--++(180:1)--++(120:1)--++(60:1)--++(0:1)--++(-60:1);

\draw[fill=black](0:0)circle(.05)++(60:1)++(0:1)circle(.05)++(0:1)++(300:1)circle(.05)++(240:1)++(180:1)circle(.05)++(180:1)++(120:1);

\draw[fill=white](0:0)++(60:1)circle(.05)++(0:1)++(0:1)circle(.05)++(300:1)++(240:1)circle(.05)++(180:1)++(180:1)circle(.05)++(120:1);

\draw[fill=black](0:0)++(180:1)++(120:1)circle(.05)++(60:1)++(0:1)circle(.05)++(-60:1);

\draw[fill=white](0:0)++(180:1)circle(.05)++(120:1)++(60:1)circle(.05)++(0:1)++(-60:1);
\end{scope}

\begin{scope}[shift={(0,1)}]
\draw[thin](0:0)--++(60:1)--++(0:2)--++(300:1)--++(240:1)--++(180:2)--++(120:1);
\draw[thin](0:0)--++(180:1)--++(120:1)--++(60:1)--++(0:1)--++(-60:1);

\draw[thin](0:0)++(60:1)++(0:2)--++(60:1)--++(0:1)--++(-60:1)--++(-120:1)--++(-180:1);

\draw[fill=black](0:0)++(60:1)++(0:2)++(60:1)circle(.05)++(0:1)++(-60:1)circle(.05)++(-120:1)++(-180:1);
\draw[fill=white](0:0)++(60:1)++(0:2)++(60:1)++(0:1)circle(.05)++(-60:1)++(-120:1)circle(.05)++(-180:1);

\draw[fill=black](0:0)circle(.05)++(60:1)++(0:1)circle(.05)++(0:1)++(300:1)circle(.05)++(240:1)++(180:1)circle(.05)++(180:1)++(120:1);

\draw[fill=white](0:0)++(60:1)circle(.05)++(0:1)++(0:1)circle(.05)++(300:1)++(240:1)circle(.05)++(180:1)++(180:1)circle(.05)++(120:1);

\draw[fill=black](0:0)++(180:1)++(120:1)circle(.05)++(60:1)++(0:1)circle(.05)++(-60:1);

\draw[fill=white](0:0)++(180:1)circle(.05)++(120:1)++(60:1)circle(.05)++(0:1)++(-60:1);

\end{scope}

\begin{scope}[shift={(7,0)}]
\draw[thin](0:0)--++(60:1)--++(0:2)--++(300:1)--++(240:1)--++(180:2)--++(120:1);
\draw[thin](0:0)--++(180:1)--++(120:1)--++(60:1)--++(0:1)--++(-60:1);

\draw[thin](0:0)++(60:1)++(0:2)--++(60:1)--++(0:1)--++(-60:1)--++(-120:1)--++(-180:1);

\draw[thin](0:0)++(-60:1)--++(-120:1)--++(-180:1)--++(-240:1)--++(-300:1);

\draw[thin](0:0)++(-60:1)++(0:3)++(60:1)--++(-60:1)--++(-120:1)--++(-180:1)--++(-240:1);

\draw[fill=black](0:0)++(-60:1)++(-120:1)circle(.05)++(-180:1)++(-240:1)circle(.05)++(-300:1);

\draw[fill=white](0:0)++(-60:1)++(-120:1)++(-180:1)circle(.05)++(-240:1)++(-300:1);

\draw[fill=black](0:0)++(-60:1)++(0:3)++(60:1)++(-60:1)circle(.05)++(-120:1)++(-180:1)circle(.05)++(-240:1);

\draw[fill=white](0:0)++(-60:1)++(0:3)++(60:1)++(-60:1)++(-120:1)circle(.05)++(-180:1)++(-240:1);

\draw[fill=black](0:0)++(60:1)++(0:2)++(60:1)circle(.05)++(0:1)++(-60:1)circle(.05)++(-120:1)++(-180:1);
\draw[fill=white](0:0)++(60:1)++(0:2)++(60:1)++(0:1)circle(.05)++(-60:1)++(-120:1)circle(.05)++(-180:1);

\draw[fill=black](0:0)circle(.05)++(60:1)++(0:1)circle(.05)++(0:1)++(300:1)circle(.05)++(240:1)++(180:1)circle(.05)++(180:1)++(120:1);

\draw[fill=white](0:0)++(60:1)circle(.05)++(0:1)++(0:1)circle(.05)++(300:1)++(240:1)circle(.05)++(180:1)++(180:1)circle(.05)++(120:1);

\draw[fill=black](0:0)++(180:1)++(120:1)circle(.05)++(60:1)++(0:1)circle(.05)++(-60:1);

\draw[fill=white](0:0)++(180:1)circle(.05)++(120:1)++(60:1)circle(.05)++(0:1)++(-60:1);

\end{scope}

\end{tikzpicture}
\caption{ Illustration of the three cases in Lemma \ref{LemmaOctagon3bdd}.\EEE}
\label{Fig: Separating Octagon}
\end{figure}

 We are now in a position to prove \EEE Theorem \ref{TheoremGroundstatesleq31}.

\begin{proof}[Proof of Theorem \ref{TheoremGroundstatesleq31}]
As observed below Proposition \ref{TheoremGroundstatesleq-new},  it suffices to check that for $n \ge 30$    the bond graph of a ground state $C_n$ does not contain an octagon.  Assume  by contradiction that  $n\geq 30$ \EEE and that the bond graph contains  an \EEE  octagon. By  Proposition \ref{TheoremGroundstatesleq-new}  we know that there do not exist any flags in the bond graph and the bond graph contains only one  (boundary) \EEE octagon and hexagons otherwise. By Lemma \ref{lemma: bridges}  the bond graph does not contain any acyclic bonds. Denote by $\{x_0,\ldots,x_7\}$ the octagon in the bond graph. We need to consider the three cases
$${\rm (a)} \ \ \# X^3  \le 3, \ \ \ \ \ {\rm (b)} \ \  \# X^3 \in  \{4,5\}, \ \ \ \ \ {\rm (c)} \ \  \# X^3 \ge 6, $$
where $X^3:= \{x_i \in \{0,\ldots,7\}: x_i \text{ is } 3\text{-bonded}\}$. (See Fig.~\ref{Fig: Separating Octagon} for an illustration.) \EEE

\noindent \emph{Proof of Case $(\mathrm{a})$:} 
Since $\#  X^3 \EEE \leq 3$, by Lemma \ref{LemmaOctagon3bdd} we have that there are $k$ $2$-bonded atoms with $k\geq 5$ and they form a connected set. Hence, removing these $2$-bonded atoms, we remove exactly $k+1$ bonds.  Estimating the energy of every bond by $-1$ we get by Theorem \ref{TheoremEnergyGroundstates}  
\begin{align*}
\mathcal{E}(C_n) \geq -(k+1) - \lfloor\beta(n-k)\rfloor.
\end{align*}
 This implies  $\mathcal{E}(C_n) > \lfloor \beta(n) \rfloor$. Indeed, this follows from Lemma \ref{LemmaPropertiesbeta} 3)-4) and the fact that $k\ge 5$, $n \ge 30$. \EEE  This gives a contradiction in Case (a).  \EEE

%

\noindent \emph{Proof of Case $(\mathrm{b})$:} By the assumption and Lemma \ref{LemmaOctagon3bdd} we have that there exist $j_1,j_2,j_3 \in \{0,\ldots,7\}$ such that  
\begin{align*}
 \mathcal{E}\big(C_n \setminus \{ (x_{j_1},q_{j_1}), (x_{j_2}, q_{j_2}),(x_{j_3},q_{j_3}) \} \big)  \le \EEE \mathcal{E}(C_n)+5 \EEE
\end{align*}
and  $C_n \setminus \{ (x_{j_1},q_{j_1}), (x_{j_2}, q_{j_2}),(x_{j_3},q_{j_3}) \} $ \EEE is not connected. Denote by $n_1,n_2 \in \mathbb{N}, n_1+n_2 =n-3$ the cardinality of the two connected components of  $C_n \setminus \{ (x_{j_1},q_{j_1}), (x_{j_2}, q_{j_2}),(x_{j_3},q_{j_3}) \} $ \EEE which do not have any bonds between them. Since the bond graph of $C_n$ does not contain any acyclic bonds, as explained at the beginning of the proof, \EEE we have  $n_1,n_2 \geq 6$. By Lemma \ref{LemmaPropertiesbeta} 2) and Theorem \ref{TheoremEnergyGroundstates} \EEE we obtain
\begin{align*}
\mathcal{E}(C_n)   \ge \EEE \mathcal{E}\big(C_n \setminus \{ (x_{i},q_{i}): \  i = i_1,i_2,i_3\} \big) - 5  \ge \EEE -\lfloor\beta(n_1)\rfloor - \lfloor \beta(n_2)\rfloor - 5 > -\lfloor\beta(n-3)\rfloor - 4.
\end{align*}
By Theorem \ref{TheoremEnergyGroundstates} we have that $\mathcal{E}(C_n)$ is an integer. This implies
\begin{align}\label{eq: n1-NNN}
\mathcal{E}(C_n) &\geq -\lfloor\beta(n-3)\rfloor - 3 \geq -\frac{3}{2}n +\sqrt{\frac{3}{2}(n-3)} + \frac{3}{2}\notag\\&= -\frac{3}{2}n + \sqrt{\frac{3}{2}n} +\frac{3}{2}- \frac{9}{2\Big(\sqrt{\frac{3}{2}n}+\sqrt{\frac{3}{2}(n-3)}\Big)}.
\end{align}
It is elementary to check that for  $n\geq 16$  we have
$$
\frac{3}{2}- \frac{9}{2\Big(\sqrt{\frac{3}{2}n}+\sqrt{\frac{3}{2}(n-3)}\Big)} \geq  1.  
$$
This is a contradiction to the fact that $C_n$ is a ground state.

\noindent \emph{Proof of Case $(\mathrm{c})$:} By the assumption and Lemma \ref{LemmaOctagon3bdd} we have that $C_n^{\mathrm{bulk}}$ is not connected.  As the bond graph contains an octagon, there exists a non-equilibrated atom.  Lemma \ref{LemmaAbulk} \EEE implies $\mathcal{A}_{\rm bulk} = \emptyset$ and thus a boundary atom is not equilibrated. Then by Lemma \ref{LemmaBoundaryEnergy} we get
\begin{align*}
\mathcal{E}^{\mathrm{bnd}}(C_n) > -\frac{3}{2}d +3.
\end{align*}
 As $C_n^{\mathrm{bulk}}$ is not connected, $C_n^{\mathrm{bulk}}$ cannot be a ground state. \EEE Applying Theorem \ref{TheoremEnergyGroundstates} to $C_n^{\mathrm{bulk}}$ we thus obtain
\begin{align*}
\mathcal{E}^{\mathrm{bulk}}(C_n)>-\lfloor   \beta( n-d) \EEE \rfloor. 
\end{align*}
By Theorem \ref{TheoremEnergyGroundstates}, we have \EEE $\mathcal{E}(C_n)=-b$ and therefore $\mathcal{E}^{\mathrm{bulk}}(C_n)$ and $\mathcal{E}^{\mathrm{bnd}}(C_n)$ are integers. Hence, we obtain
\begin{align*}
\mathcal{E}(C_n)\ge -\lfloor   \beta( n-d) \EEE \rfloor -\frac{3}{2}d +5.
\end{align*}
Using Lemma \ref{LemmaBoundaryestimate}  and $\eta =2$  we get
\begin{align*}
\mathcal{E}(C_n) \geq -\frac{3}{2}n + 5 -\sqrt{\frac{3}{2}(-4\mathcal{E}(C_n)-5n+8)}.
\end{align*}
This together with Lemma \ref{LemmaSquareroot} applied for $j=5$, $m= 8$, and  $x = \mathcal{E}(C_n)$  leads to
\begin{align}\label{eq: n1}
\mathcal{E}(C_n) \geq -\frac{3}{2}n +2 + \sqrt{\frac{3}{2}(n-6)}= -\frac{3}{2}n +\sqrt{\frac{3}{2}n}+2 - \frac{9}{\sqrt{\frac{3}{2}n}+\sqrt{\frac{3}{2}(n-6)}}.
\end{align}
For   $n \geq 17$   we have
\begin{align*}
 2 - \frac{9}{\sqrt{\frac{3}{2}n}+\sqrt{\frac{3}{2}(n-6)}} \geq 1
\end{align*}
which leads to a contradiction to the fact that $C_n$ is a ground state. 
\end{proof}

\begin{remark}\label{rem: octogons}
{\normalfont  

Inspection of the previous proof shows that among $13 \le n \le 29$ only for $n=15,18,21,29$ boundary octagons may occur. Indeed, in Case $(\mathrm{a})$ this follows from Lemma \ref{LemmaPropertiesbeta} 4), see particularly Table \ref{table}. In Case $(\mathrm{b})$ and Case $(\mathrm{c})$ we obtain a contradiction for each $13 \le n \le 29$: in Case $(\mathrm{b})$, we necessarily have $n \ge 16$  (see upper left configuration in Fig.~\ref{Fig: Separating Octagon}) and thus a contradiction in \eqref{eq: n1-NNN}.  In Case $(\mathrm{c})$, we necessarily have $n \ge 22$  (see rightmost configuration in Fig.~\ref{Fig: Separating Octagon}) and thus a contradiction in \eqref{eq: n1}.

For $n=10,11$, the presence of an octagon is excluded by Proposition \ref{TheoremGroundstatesleq-new} and the fact that  each hexagon can share at most $2$ atoms with an octagon (see Lemma \ref{LemmaHexagon} and  Lemma \ref{lemma: octa}). Consequently, for $n \le 29$, ground states may contain a boundary octagon only for $n=8,9,12,15,18,21,29$. This is indeed possible as shown in Fig.~\ref{FigureFlexible} and Fig.~\ref{FigureOctagons}, cf.\ Table~\ref{table2}.  \EEE 
}
\end{remark}

\begin{figure}[htp]
\begin{tikzpicture}[scale=0.6]
\begin{scope}[shift={(-4,0)}]
\draw[thin](2,0)++(60:1)--++(1,0)--++(1,0)--++(300:1)--++(240:1)--++(180:1)--++(180:1);
\draw[fill=white](2,0)++(60:1)++(1,0)circle(.05)++(1,0)++(300:1)circle(.05)++(240:1)++(180:1)circle(.05)++(180:1);

\draw[fill=black](2,0)++(60:1)++(1,0)++(1,0)circle(.05)++(300:1)++(240:1)circle(.05)++(180:1)++(180:1);
\foreach \k in {0,...,5}{
\foreach \j in {0,2,4}{
\draw[thin](\k*60+30:{sqrt(3)})++(\j*60:1)--++(\j*60+120:1);
\draw[fill=black](\k*60+30:{sqrt(3)})++(\j*60:1) circle(.05);
}
\foreach \j in {1,3,5}{
\draw[thin](\k*60+30:{sqrt(3)})++(\j*60:1)--++(\j*60+120:1);
\draw[fill=white](\k*60+30:{sqrt(3)})++(\j*60:1) circle(.05);
}
}
\end{scope}

\begin{scope}[shift={(7,-4)},rotate=-90]

\draw[thin](0,0)--++(30:1)--++(90:1)--++(90:1)--++(150:1)--++(210:1)--++(270:1)--++(270:1)--++(330:1);

\draw[thin](0,0)--++(270:1)--++(330:1)--++(30:1)--++(90:1)--++(150:1);

\draw[thin](0,0)++(270:1)--++(210:1)--++(150:1)--++(90:1)--++(30:1);

\draw[fill=black](0,0)++(270:1)++(330:1)circle(.05)++(30:1)++(90:1)circle(.05)++(150:1);
\draw[fill=white](0,0)++(270:1)circle(.05)++(330:1)++(30:1)circle(.05)++(90:1)++(150:1);

\draw[fill=black](0,0)++(270:1)++(210:1)circle(.05)++(150:1)++(90:1)circle(.05)++(30:1);

\draw[fill=white](0,0)++(270:1)++(210:1)++(150:1)circle(.05)++(90:1)++(30:1);

\draw[fill=black](0,0)circle(.05)++(30:1)++(90:1)circle(.05)++(90:1)++(150:1)circle(.05)++(210:1)++(270:1)circle(.05)++(270:1)++(330:1);
\draw[fill=white](0,0)++(30:1)circle(.05)++(90:1)++(90:1)circle(.05)++(150:1)++(210:1)circle(.05)++(270:1)++(270:1)circle(.05)++(330:1);
\end{scope}

\begin{scope}[shift={(7,-4)}]

\draw[thin](-2,0)++(300:1)++(180:1)--++(240:1)--++(300:1)--++(360:1)--++(60:1);

\draw[fill=black](-2,0)++(300:1)++(180:1)++(240:1)circle(.05)++(300:1)++(360:1)circle(.05)++(60:1);

\draw[fill=white](-2,0)++(300:1)++(180:1)++(240:1)++(300:1)circle(.05)++(360:1)++(60:1);

\draw[thin](-1,0)++(120:1)--++(180:1)--++(240:1)--++(300:1)--++(360:1);
\draw[fill=black](-1,0)++(120:1)++(180:1)++(240:1)circle(.05)++(300:1)++(360:1);

\draw[fill=white](-1,0)++(120:1)++(180:1)circle(.05)++(240:1)++(300:1)circle(.05)++(360:1);

\end{scope}

\begin{scope}[shift={(5,0)}]

\draw[thin](-1,0)++(120:1)--++(180:1)--++(240:1)--++(300:1)--++(360:1);
\draw[fill=black](-1,0)++(120:1)++(180:1)++(240:1)circle(.05)++(300:1)++(360:1);

\draw[fill=white](-1,0)++(120:1)++(180:1)circle(.05)++(240:1)++(300:1)circle(.05)++(360:1);
\begin{scope}[rotate=-90]

\draw[thin](0,0)--++(30:1)--++(90:1)--++(90:1)--++(150:1)--++(210:1)--++(270:1)--++(270:1)--++(330:1);

\draw[thin](0,0)--++(270:1)--++(330:1)--++(30:1)--++(90:1)--++(150:1);

\draw[thin](0,0)++(270:1)--++(210:1)--++(150:1)--++(90:1)--++(30:1);

\draw[fill=black](0,0)++(270:1)++(330:1)circle(.05)++(30:1)++(90:1)circle(.05)++(150:1);
\draw[fill=white](0,0)++(270:1)circle(.05)++(330:1)++(30:1)circle(.05)++(90:1)++(150:1);

\draw[fill=black](0,0)++(270:1)++(210:1)circle(.05)++(150:1)++(90:1)circle(.05)++(30:1);

\draw[fill=white](0,0)++(270:1)++(210:1)++(150:1)circle(.05)++(90:1)++(30:1);

\draw[fill=black](0,0)circle(.05)++(30:1)++(90:1)circle(.05)++(90:1)++(150:1)circle(.05)++(210:1)++(270:1)circle(.05)++(270:1)++(330:1);
\draw[fill=white](0,0)++(30:1)circle(.05)++(90:1)++(90:1)circle(.05)++(150:1)++(210:1)circle(.05)++(270:1)++(270:1)circle(.05)++(330:1);
\end{scope}
\end{scope}

\begin{scope}[shift={(0,-5)}, rotate=-90]

\draw[thin](0,0)--++(30:1)--++(90:1)--++(90:1)--++(150:1)--++(210:1)--++(270:1)--++(270:1)--++(330:1);

\draw[thin](0,0)--++(270:1)--++(330:1)--++(30:1)--++(90:1)--++(150:1);

\draw[thin](0,0)++(270:1)--++(210:1)--++(150:1)--++(90:1)--++(30:1);

\draw[fill=black](0,0)++(270:1)++(330:1)circle(.05)++(30:1)++(90:1)circle(.05)++(150:1);
\draw[fill=white](0,0)++(270:1)circle(.05)++(330:1)++(30:1)circle(.05)++(90:1)++(150:1);

\draw[fill=black](0,0)++(270:1)++(210:1)circle(.05)++(150:1)++(90:1)circle(.05)++(30:1);

\draw[fill=white](0,0)++(270:1)++(210:1)++(150:1)circle(.05)++(90:1)++(30:1);

\draw[fill=black](0,0)circle(.05)++(30:1)++(90:1)circle(.05)++(90:1)++(150:1)circle(.05)++(210:1)++(270:1)circle(.05)++(270:1)++(330:1);
\draw[fill=white](0,0)++(30:1)circle(.05)++(90:1)++(90:1)circle(.05)++(150:1)++(210:1)circle(.05)++(270:1)++(270:1)circle(.05)++(330:1);
\end{scope}

\begin{scope}[shift={(-6,-5)},rotate=-90]

\draw[thin](0,0)--++(30:1)--++(90:1)--++(90:1)--++(150:1)--++(210:1)--++(270:1)--++(270:1)--++(330:1);


\draw[thin](0,0)++(270:1)--++(210:1)--++(150:1)--++(90:1)--++(30:1);

\draw[thin](0,0)--++(90:-1);
\draw[fill=white](90:-1)circle(.05);


\draw[fill=black](0,0)++(270:1)++(210:1)circle(.05)++(150:1)++(90:1)circle(.05)++(30:1);

\draw[fill=white](0,0)++(270:1)++(210:1)++(150:1)circle(.05)++(90:1)++(30:1);

\draw[fill=black](0,0)circle(.05)++(30:1)++(90:1)circle(.05)++(90:1)++(150:1)circle(.05)++(210:1)++(270:1)circle(.05)++(270:1)++(330:1);
\draw[fill=white](0,0)++(30:1)circle(.05)++(90:1)++(90:1)circle(.05)++(150:1)++(210:1)circle(.05)++(270:1)++(270:1)circle(.05)++(330:1);
\end{scope}
\end{tikzpicture}
\caption{The ground states containing an octagon for $n=12,15,18,21,29$ (up to isometry and changing of the charges).}  
\label{FigureOctagons}
\end{figure}
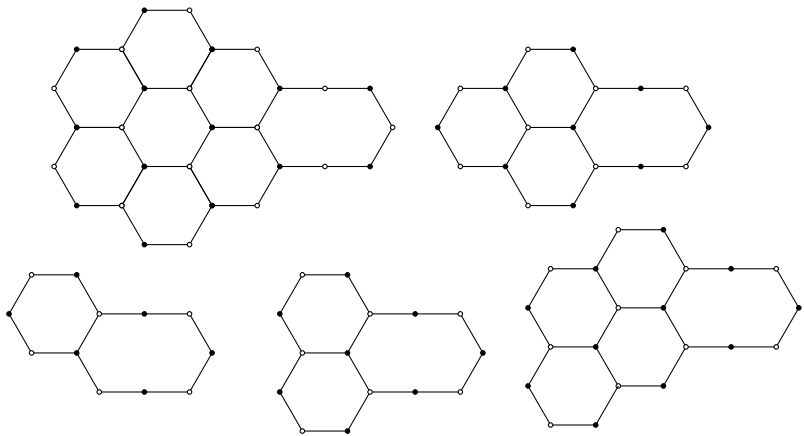

\EEE
 
\section{Characterization of the net charge}

This final section is devoted to the proof of Theorem \ref{TheoremCharge}(i). We recall that part (ii) of the statement has already been addressed by an explicit construction in Section \ref{sec: daisy + something}. 

We start with some preliminary definitions. First, recall the definition of the hexagonal lattice $\mathcal{L}$ in \eqref{eq: hex-lattice}. Set $u_1 = (1,0)$, $u_2 = (\frac{1}{2},\frac{\sqrt{3}}{2})$, $u_3 = (-\frac{1}{2},\frac{\sqrt{3}}{2})$, $u_4 = (-1,0)$.

\begin{definition}[Zig-zag paths]\label{def: zig-zag}
(i) A tuple $(p_1,\ldots,p_m) \subset \mathcal{L}$ is called a \emph{zig-zag path} if there exists $k \in \lbrace 1,2,3 \rbrace$ such that 
\begin{align}\label{eq: zig-zag}
(a)& \ \ p_j - p_{j-1} \in \lbrace u_k, u_{k+1} \rbrace  \ \ \ \text{ for all } j \in \lbrace 2,\ldots,m\rbrace, \notag\\
(b) & \ \  p_{j+1} - p_{j} \neq p_j - p_{j-1} \ \ \ \text{ for all } j \in \lbrace 2,\ldots,m-1\rbrace.
\end{align}

(ii) We say that two zig-zag paths have the same \emph{orientation} if the same $k \in \lbrace 1,2,3 \rbrace$ appears in \eqref{eq: zig-zag}.

(iii) Given a configuration $C_n$ with $X_n \subset \mathcal{L}$, we say the zig-zag path $(p_1,\ldots,p_m) \subset \mathcal{L}$, $m \ge 3$, is a \emph{bridging zig-zag path} for $C_n$ if $p_1, p_m \in X_n\EEE$ and $p_2,\ldots,p_{m-1} \notin  X_n\EEE$. 

\end{definition}

Let $C_n$ be a ground state for $n \ge 30$ with no acyclic bonds. \EEE By Theorem \ref{TheoremGroundstatesleq31} we get that \EEE  $X_n$ is defect-free and satisfies $X_n \subset \mathcal{L}$  (up to isometry). \EEE Let $(p_1,\ldots,p_m)$ be a bridging zig-zag path for $C_n$.  Since  $C_n$ \EEE is connected and defect-free, we find that $\mathcal{L} \setminus (X_n \cup \bigcup_{j=2}^{m-1} p_j)$ consists of two (one of them possibly empty) connected components. Exactly one of these components is bounded which we denote by  $\mathcal{B}( X_n; \EEE ( p_1, \ldots, p_m ))$.

\begin{lemma}[Bridging zig-zag paths of ground states]\label{lemma: zig-zag paths}
Let $n \ge 30$. Each ground state  with no acyclic bonds \EEE has at most one bridging zig-zag path.
\end{lemma}

\begin{proof}
As a preparation, we observe the following: suppose that  there exists a bridging zig-zag path for a ground state  $C_n$. Then we can choose a bridging zig-zag path $(p_1,\ldots,p_m) \subset \mathcal{L}$ for $C_n$ such that $\mathcal{B}(X_n;(p_1,\ldots,p_m)) = \emptyset$. To see this, we proceed as follows. Take an arbitrary bridging zig-zag path $(\bar{p}_1,\ldots,\bar{p}_{\bar{m}})$ and consider $\bar{\mathcal{B}} := \mathcal{B}(X_n;(\bar{p}_1,\ldots,\bar{p}_{\bar{m}}))$. If $\mathcal{\bar{B}} = \emptyset$, we have concluded.  If $\mathcal{\bar{B}} \neq \emptyset$, we can consider another  bridging zig-zag path $(\tilde{p}_1,\ldots, \tilde{p}_{\tilde{m}})$ for $C_n$ having the same orientation as $(\bar{p}_1,\ldots,\bar{p}_{\bar{m}})$ and satisfying $\tilde{p}_2,\ldots, \tilde{p}_{\tilde{m}-1} \in \mathcal{\bar{B}}$. Define $\tilde{\mathcal{B}} := \mathcal{B}(X_n;(\tilde{p}_1,\ldots,\tilde{p}_{\tilde{m}}))$ and note that $\tilde{\mathcal{B}} \subset \bar{\mathcal{B}}$,  $\#\tilde{\mathcal{B}} < \# \bar{\mathcal{B}}$. This construction can be iterated and after a finite number of iteration steps \EEE we find a  bridging zig-zag path $(p_1,\ldots,p_m) \subset \mathcal{L}$ for $C_n$ such that $\mathcal{B}(X_n;(p_1,\ldots,p_m)) = \emptyset$. \EEE

Now suppose  by contradiction \EEE that there was a ground state $C_n$, $n \ge 30$,  with no acyclic bonds \EEE which contains two bridging zig-zag paths. Consider a bridging zig-zag path $(p_1,\ldots, p_m)$ satisfying $\mathcal{B}(X_n;(p_1,\ldots,p_m)) = \emptyset$ and define $X_n' = X_n \cup \bigcup_{j=2}^{m-1} p_j$. \EEE Clearly, we can assign charges to the atoms  $p_2 ,\ldots, p_{m-1}$ \EEE to obtain a configuration $C_n'$ with \EEE alternating charge distribution. Note that $C_n'$ consists of  $n+m-2$ atoms. We now estimate the energy of $C_n'$.

First, we observe that between the atoms of the path $(p_1,\ldots,p_m)$ there are $m-1$ bonds. Since $\mathcal{B}(X_n;(p_1,\ldots,p_m)) = \emptyset$, either each \EEE $p_2,p_4, \ldots, p_{2 \lfloor (m-1)/2 \rfloor  }$ or  each \EEE $p_3,p_5,\ldots, p_{2 \lfloor m/2 \rfloor -1 }$  is \EEE bonded to an atom of $C_n$. (The latter set  is \EEE empty if $m=3$.) Thus, we obtain
\begin{align*}
\mathcal{E}(C_n') \le \mathcal{E}(C_n) - (m-1) - \min\lbrace \lfloor (m-1)/2 \rfloor , \lfloor m/2 \rfloor -1  \rbrace \le \mathcal{E}(C_n) - \frac{3}{2}m + \frac{5}{2}.
\end{align*}
In particular, as $C_n$ was supposed to be a ground state, this implies $\mathcal{E}(C_n') \le - \lfloor \beta(n) \rfloor - \frac{3}{2}m + \frac{5}{2} \le  - \lfloor \beta(n+m-2) \rfloor$. Here the second inequality is elementary to check. This shows that $C_n'$ is a ground state.

As $C_n$ has two bridging zig-zag paths, there is at least one \EEE bridging zig-zag path for $C_n'$. We now repeat the above procedure. Choose $(p'_1,\ldots, p'_{m'})$ with $\mathcal{B}(X'_n;(p'_1,\ldots,p'_{m'})) = \emptyset$ and  define a configuration $C_n''$ with $\#  X_n'' \EEE = n+m+m'-4$, with alternating charge distribution and consisting of the atoms $X_n' \cup \bigcup_{j=2}^{m'-1} p'_j$. Arguing as before, we calculate
\begin{align*}
\mathcal{E}(C_n'') \le \mathcal{E}(C'_n) - \frac{3}{2}m' + \frac{5}{2} \le \mathcal{E}(C_n) - \frac{3}{2}(m+m') + 5.
\end{align*}
Since $C_n$ was supposed to be a ground state, this implies 
\begin{align*}
\mathcal{E}(C_n'') &\le - \left\lfloor \frac{3}{2}n - \sqrt{\frac{3}{2}n} \right\rfloor - \frac{3}{2}(m+m') + 5 \\&\le - \left\lfloor\frac{3}{2} (n+ m + m'-4) - \sqrt{\frac{3}{2}(n + m + m'-4) }\right\rfloor  -\frac{1}{2}.
\end{align*}  
This implies $\mathcal{E}(C_n'') < -\lfloor \beta(n + m + m'-4) \rfloor$ which contradicts Theorem \ref{TheoremEnergyGroundstates}. This concludes the proof. 
\end{proof}

Recall the construction of daisies in Section \ref{sec: daisy}. From \cite{Davoli15} we obtain the following result.

\begin{proposition}[Deviation from Wulff-shape]\label{prop: davoli}
Let $ n \ge 30$ and let $C_n$ be a ground state with no acyclic bonds. \EEE Then, possibly after translation, we find two daisies $X^{\rm daisy}_{6k_1^2} \subset \mathcal{L}$ and $X^{\rm daisy}_{6k_2^2} \subset \mathcal{L}$ with $X^{\rm daisy}_{6k_1^2} \subset X_n \EEE \subset X_{6k_2^2}^{\rm daisy}$ such that
$$0 < k_2 - k_1 \le cn^{1/4}, $$  
where $c>0$ is a universal constant independent of $n$ and $C_n$.
\end{proposition}

 \begin{proof}
From Theorem \ref{TheoremGroundstatesleq31} and the fact that $C_n$ does not have acyclic bonds \EEE we get that the ground state $C_n$ is a subset of the hexagonal lattice. Moreover, $C_n$ is repulsion-free, see Remark \ref{rem: repulsionsfree}. Thus, the energy of a ground state  coincides with the one in \cite{Davoli15}, see  \cite[Equation (6)]{Davoli15}. The claim then follows from \cite[Theorem 1.2]{Davoli15}.  
 \end{proof}

We are now in the position to prove Theorem \ref{TheoremCharge}. 

\begin{proof}[Proof of Theorem \ref{TheoremCharge}]
As discussed at the beginning of the section, \EEE it remains to prove part (i) of the statement. In view of Theorem \ref{TheoremGroundstatesleq31} and Remark \ref{rem: main}(i), it suffices to treat the case that $C_n$ does not have acyclic bonds. \EEE We apply Proposition \ref{prop: davoli} to find two daisies with  $X^{\rm daisy}_{6k_1^2} \subset X_n \subset X^{\rm daisy}_{6k_2^2}$. It is elementary to see that $X^{\rm daisy}_{6k_2^2} \setminus X^{\rm daisy}_{6k_1^2}$ can be written as the union of $6(k_2-k_1)$ zig-zag paths as introduced in Definition \ref{def: zig-zag}. We claim that  $X_n \setminus X^{\rm daisy}_{6k_1^2}$ can be written as the union of at most $6(k_2-k_1)+ 1 \EEE$ zig-zag paths. 

 To see this, thanks to Lemma \ref{lemma: zig-zag paths}, observe that at most one of the $6(k_2-k_1)$ zig-zag paths may contain a bridging zig-zag path. Denote this zig-zag path by $\mathcal{P}_1=(p_1,\ldots,p_m)$ and the bridging zig-zag path by $(p_{k_1},\ldots,p_{k_2})$, where  $1 \leq k_1 < k_2\leq m$. We indicate   the two zig-zag paths $(p_1,\ldots,p_{k_1}) \cap X_n $ and $(p_{k_2},\ldots,p_{m}) \cap X_n$ by $ \gamma_0$ and $\gamma_1  $ respectively. 
For the remaining zig-zag paths segmenting $X^{\rm daisy}_{6k_2^2} \setminus X^{\rm daisy}_{6k_1^2}$, denoted by $\mathcal{P}_k$, $k=2,\ldots,6(k_2-k_1)$, we define $\gamma_k:=\mathcal{P}_k \cap X_n$.  Recalling that only  $\mathcal{P}_1$ may contain a bridging zig-zag path, we get that  $\gamma_k$ is a zig-zag path for all $k=0,\ldots,6(k_2-k_1)$. Moreover, $\gamma_k \subset X_n$ and  thus $X_n \setminus X_{6k^2_1}^\mathrm{daisy}$ can be written as the union of the $6(k_2-k_1)+ 1 \EEE$ (possibly empty) zig-zag paths $\gamma_k$, $k=0,\ldots,6(k_2-k_1)$.  \EEE 

Recall that $C_n$ has alternating charge distribution and therefore the net charge of each zig-zag path is in $\lbrace -1,0,1\rbrace$. Also recall from Section \ref{sec: daisy} that daisies always have net charge zero. This implies that the net charge of the configuration $C_n$ satisfies
$$|\mathcal{Q}(C_n) | \le 6(k_2-k_1)+1\EEE.$$
The statement follows from the fact that $k_2 - k_1 \le cn^{1/4}$, see Proposition \ref{prop: davoli}. 
\end{proof}

\section*{Acknowledgements} 
  M.\ F.\ acknowledges support from the Alexander von Humboldt Stiftung. 
 L.\ K.\ acknowledges support
from the Austrian Science Fund (FWF) project P~29681, and from the Vienna Science and Technology Fund (WWTF), the
City of Vienna, and the Berndorf Private Foundation through Project MA16-005. The authors would like to thank Ulisse Stefanelli for turning their attention to this problem. 

\EEE



\begin{thebibliography}{99}
%




\bibitem{Molecular}
 {\sc   N.L. Allinger}.
\newblock {\em Molecular structure: understanding steric and electronic effects from molecular mechanics}.
\newblock John Wiley \& Sons
\newblock (2010).

 

 


\bibitem{Yuen}
 {\sc   Y.~Au Yeung, G.~Friesecke, B.~Schmidt}.
\newblock {\em Minimizing atomic configurations of short range
pair potentials in two dimensions: crystallization in the Wulff-shape}.
\newblock Calc.\ Var.\ Partial Differential Equations
\newblock {\bf 44} (2012), 81--100.

 
\bibitem{Betermin}
{\sc L.~B\'etermin, H.~Kn\"upfer, F.~Nolte}.
\newblock  {\em Crystallization of  one-dimensional alternating two-component systems}.  
\newblock Preprint at arXiv:1804.05743.  
 
 


 
 
 \bibitem{Betermin2}
{\sc L.~B\'etermin, H.~Kn\"upfer}.
\newblock  {\em  On Born's conjecture about optimal distribution of charges for an infinite
ionic crystal.}  
\newblock J.\ Nonlinear Sci.\
\newblock {\bf 28} (2018),  1629--1656.
 
\EEE 
 

\bibitem{Blanc}
{\sc X.~Blanc, M.~Lewin}.
\newblock {\em The crystallization conjecture: a review}.
\newblock EMS Surv.\ Math.\ Sci.\  
\newblock {\bf 2} (2015), 255--306.

 
\bibitem{B43}
{\sc D.~C.~Brydges, P.~A.~Martin}. 
\newblock {\em Coulomb systems at low   density: A review}.
\newblock J.\ Stat.\ Phys.\
\newblock {\bf 96} (1999),  1163--1330.

  



 

  




\bibitem{Davoli15}
{\sc E.~Davoli, P.~Piovano, U.~Stefanelli}. 
\newblock {\em Wulff shape emergence in graphene}. 
\newblock Math.\ Models Methods Appl.\ Sci.\ 
\newblock {\bf 26} (2016), 12:2277--2310.


\bibitem{Davoli16}
{\sc E.~Davoli, P.~Piovano, U.~Stefanelli}. 
\newblock {\em Sharp $n^{3/4}$ 
 law for the minimizers of the edge-isoperimetric problem in the
 triangular lattice}. 
\newblock J.\  Nonlin.\ Sci.\
\newblock  {\bf 27} (2017),  627--660.



\bibitem{Lucia}
{\sc L.~De Luca, G.~Friesecke}. 
\newblock {\em  Crystallization in two dimensions and a discrete Gauss–Bonnet Theorem}. 
\newblock 
J.\ Nonlinear Sci.\
\newblock {\bf 28} (2017), 69--90.



 
 
 \bibitem{ELi}
 {\sc   W.~E, D.~Li}.
\newblock {\em  On the crystallization of 2D hexagonal lattices}.
\newblock Comm.\ Math.\ Phys.\
\newblock {\bf 286} (2009), 1099--1140.



 

\bibitem{Smereka15}
{\sc B.~Farmer, S.~Esedo\={g}lu, P.~Smereka}. 
\newblock {\em Crystallization for a
Brenner-like potential}.  
\newblock  Comm.\ Math.\ Phys.\
\newblock {\bf 349} (2017),  1029--1061. 
  

 
 
 \bibitem{Flateley1}
{\sc L.~Flatley, M.~Taylor, A.~Tarasov, F~ Theil}. 
\newblock {\em Packing twelve spherical caps to maximize tangencies}. 
\newblock J.\ Comput.\ Appl.\ Math.\
\newblock {\bf 254} (2013),220--225.



 \bibitem{Flateley2}
{\sc L.~Flatley, F.~Theil}. 
\newblock {\em Face-centered cubic crystallization of atomistic configurations}. 
\newblock Arch.\ Ration.\ Mech.\ Anal.\
\newblock {\bf 218} (2015), 363--416.



  

 

 \bibitem{emergence}
 {\sc  M.~Friedrich, U.~Stefanelli}.
 \newblock {\em Graphene ground states}.
 \newblock Z.\ Angew.\ Math.\ Phys.\
\newblock  {\bf 69} (2018): 70.




 \bibitem{FriedrichKreutzSquare}
 {\sc  M.~Friedrich, L.~Kreutz}.
 \newblock {\em Finite crystallization and Wulff shape emergence for ionic compounds in the square lattice}.
 \newblock Submitted, 2019. Preprint at  https://arxiv.org/abs/1903.00331.
 \EEE





\bibitem{Friesecke-Theil15}
{\sc G.~Friesecke, F.~Theil}. 
\newblock {\em Molecular geometry optimization,
  models}. In the Encyclopedia of Applied and Computational Mathematics,
B. Engquist (Ed.), Springer, 2015.











 



 \bibitem{Gardner}
{\sc C.~S.~Gardner, C.~Radin}. 
\newblock {\em The infinite-volume ground state of the Lennard-Jones potential}. 
\newblock J.\ Stat.\ Phys.\
\newblock {\bf 20} (1979), 719--724.



  
  
\bibitem{Geim}
 {\sc   A.~K.~Geim, K.~S.~Novoselov}.
\newblock {\em The rise of graphene}.
\newblock Nat.\ Mater.\
\newblock {\bf 6} (2007), 183--191. 




\bibitem{Hamrick}
{\sc G.~C.~Hamrick and C.~Radin}. 
\newblock {\em The symmetry of ground states under perturbation}. 
\newblock J.\ Stat.\ Phys.\
\newblock {\bf 21} (1979), 601--607.
  

\bibitem{HR}
{\sc  R.~Heitman, C.~Radin}. 
\newblock {\em Ground states for sticky disks}. 
\newblock J.\ Stat.\ Phys.\
\newblock {\bf 22} (1980), 3:281--287. 


 





\bibitem{DresselhausM}
{\sc K.~K.~Kim, A.~Hsu, X.~Jia, S.~M.~Kim, Y.~Shi, M.~Dresselhaus, T.~Palacios, J.~Kong}. 
\newblock {\em Synthesis and characterization of hexagonal boron nitride film as a dielectric layer for graphene devices}. 
\newblock Acs Nano 
\newblock {\bf 6} (2012), 8583--8590.

 





\bibitem{cronut}
{\sc G.~Lazzaroni, U.~Stefanelli}. {\em Chain-like minimizers in three
  dimensions}.
  \newblock Transactions of Mathematics and Its Applications
  \newblock {\bf 2} (2018), 1--22. \EEE


\bibitem{Lewars}
{\sc E.~G.~Lewars}.
\newblock {\em Computational Chemistry}. 2nd edition, Springer, 2011.

 
 

\bibitem{Mainini-Piovano}
{\sc E.~Mainini, P.~Piovano, U.~Stefanelli}. 
\newblock {\em Finite crystallization in the square lattice}. 
\newblock Nonlinearity
\newblock {\bf 27} (2014), 717--737.
 

 
\bibitem{Mainini}
{\sc E.~Mainini, U.~Stefanelli}. 
\newblock {\em Crystallization in carbon nanostructures}. 
\newblock Comm.\ Math.\ Phys.\
\newblock {\bf 328} (2014), 545--571. 



 
 
 

 \bibitem{Pauling}
{\sc L.~Pauling}. 
\newblock {\em The nature of the chemical bond and the structure of molecules and crystals: an introduction to modern structural chemistry}.
\newblock Ithaca, New York: Cornell University Press 1960. 


 
\bibitem{Radin}
{\sc  C.~Radin}. 
\newblock {\em The ground state for soft disks}. 
\newblock J.\ Stat.\ Phys.\
\newblock {\bf 26} (1981), 2:365--373. 




  \bibitem{Radi}
{\sc C.~Radin}. 
\newblock {\em Classical ground states in one dimension}. 
\newblock  J.\ Stat.\ Phys.\
\newblock {\bf 35} (1983), 109--117.
  

 

 

 
 \bibitem{B195}
{\sc C.~Radin}. 
\newblock {\it Crystals and quasicrystals: a continuum
  model}. 
  \newblock Comm.\ Math.\ Phys.\
  \newblock {\bf 105} (1986),  385--390.
  


 
 
 

 

 
 
\bibitem{Schmidt-disk}
{\sc  B.~Schmidt}. 
\newblock {\em Ground states of the 2D sticky disc model: fine properties and $N^{3/4}$ law for the
deviation from the asymptotic Wulff-shape}. 
\newblock J.\ Stat.\ Phys.\
\newblock {\bf 153} (2013), 727--738.

\bibitem{Suto06}
{\sc A.~S\"ut\H{o}}. 
\newblock {\em From bcc to fcc: Interplay between oscillation long-range
and repulsive short
range forces}.  
\newblock Phys.\ Rev.\ B
\newblock  {\bf 74} (2006), 104117.

 
 \bibitem{Theil}
{\sc   F.~Theil}. 
\newblock {\em A proof of crystallization in two dimensions}. 
\newblock Comm.\ Math.\ Phys.\ 
\newblock {\bf 262} (2006), 209--236.

 
 
  
  \bibitem{Ventevogel}
{\sc W.~J.~Ventevogel, B.~R.~A.~Nijboer}. 
\newblock {\em On the configuration of systems of interacting atom with
minimum potential energy per  particle\EEE}. 
\newblock  Phys.\ A
\newblock  {\bf 98} (1979), 274--288.\EEE


  

 \bibitem{Wang}
{\sc J.~Wang, F.~Ma, M.~Sun}. 
\newblock {\em Graphene, hexagonal boron nitride, and their heterostructures: properties and applications}. 
\newblock Rsc Adv.\
\newblock {\bf 7} (2017), 16801--16822.





 


\bibitem{Wagner83} {\sc H. J. Wagner}. {\em Crystallinity in two dimensions: a note on a paper of C. Radin}.  J.\ Stat.\ Phys.\   {\bf 33},  (1983),  523--526.




\end{thebibliography}
\end{document}